\documentclass{lmcs} 
\pdfoutput=1
\usepackage[utf8]{inputenc}

\usepackage{lastpage}
\lmcsdoi{20}{4}{16}
\lmcsheading{}{\pageref{LastPage}}{}{}%
{Oct.~28,~2021}{Nov.~21,~2024}{}

\usepackage{tikz} 
\usetikzlibrary{arrows}
\usetikzlibrary{arrows,automata}
\usetikzlibrary{automata,positioning,arrows}
\usetikzlibrary{shapes}
\usetikzlibrary{shapes,backgrounds}
\usetikzlibrary{matrix}
\usetikzlibrary{chains,fit,shapes}
\usepackage{amsmath}
\usepackage{amssymb}
\usepackage{amsfonts}
\usepackage{amsthm}
\usepackage[ruled]{algorithm2e}

\usepackage{wrapfig}

\usepackage{stmaryrd}

\usepackage{tensor}

\usepackage{booktabs}

\DeclareMathOperator{\undefin}{\perp}

\DeclareMathOperator{\varno}{\wp}

\DeclareMathOperator{\vset}{\mathsf{VSet}}

\DeclareMathOperator{\corespanners}{\mathsf{core-}\mathfrak{S}}

\newcommand{\msrep}[1]{\mathsf{msrep}(#1)}

\DeclareMathOperator{\TestProb}{\mathsf{Testing}}
\DeclareMathOperator{\NonemptProb}{\mathsf{NonEmptiness}}
\DeclareMathOperator{\SatProb}{\mathsf{Satisfiability}}
\DeclareMathOperator{\ContProb}{\mathsf{Containment}}
\DeclareMathOperator{\EquProb}{\mathsf{Equivalence}}
\DeclareMathOperator{\HierProb}{\mathsf{Hierarchicality}}

\DeclareMathOperator{\regspanners}{\mathsf{reg-}\mathfrak{S}}
\DeclareMathOperator{\reflspanners}{\mathsf{refl-}\mathfrak{S}}

\DeclareMathOperator{\regLanguages}{\mathsf{reg}-\mathfrak{L}}

\DeclareMathOperator{\regswmLanguages}{\mathsf{reg}-\mathsf{swm}-\mathfrak{L}}

\DeclareMathOperator{\regrefLanguages}{\mathsf{reg}-\mathsf{ref}-\mathfrak{L}}

\DeclareMathOperator{\regexform}{\mathsf{RGX}}

\DeclareMathOperator{\spanfusion}{\uplus}
\DeclareMathOperator{\bigspanfusion}{\biguplus}

\DeclareMathOperator{\RE}{\mathsf{RE}}

\DeclareMathOperator{\npclass}{\mathsf{NP}}
\DeclareMathOperator{\pspaceclass}{\mathsf{PSpace}}

\DeclareMathOperator{\expspaceclass}{\mathsf{ExpSpace}}

\DeclareMathOperator{\bigO}{O}

\newcommand{\erasemorphism}[1]{\mathsf{e}_{#1}}

\newcommand{\deref}[1]{\mathfrak{d}(#1)}
\newcommand{\derefmark}{\mathfrak{d}}

\newcommand{\getWord}[1]{\mathfrak{e}(#1)}

\newcommand{\getSpanTuple}[1]{\mathsf{st}(#1)}

\DeclareMathOperator{\NFA}{\mathsf{NFA}}

\DeclareMathOperator{\spans}{\textsf{Spans}}

\DeclareMathOperator{\eword}{\varepsilon}
\DeclareMathOperator{\lang}{\mathcal{L}}

\DeclareMathOperator{\altop}{\vee}

\DeclareMathOperator{\openmark}{\triangleright}
\DeclareMathOperator{\closemark}{\triangleleft}

\DeclareMathOperator{\lce}{\textsf{LCE}}

\newcommand{\open}[1]{\tensor[^{#1}]{\triangleright}{_{}}}
\newcommand{\close}[1]{\tensor[_{}]{\triangleleft}{^{#1}}}

\newcommand{\spann}[2]{\ensuremath{[#1,#2\rangle}}

\newcommand{\varsx}{\ensuremath{\mathsf{x}}}
\newcommand{\varsy}{\ensuremath{\mathsf{y}}}
\newcommand{\varsz}{\ensuremath{\mathsf{z}}}

\newcommand{\varset}{\ensuremath{\mathcal{X}}}
\newcommand{\extVarset}[1]{\textsf{ext}(#1)}
\newcommand{\splitset}[2]{\textsf{split}_{#1}(#2)}

\newcommand{\ta}{\ensuremath{\mathtt{a}}}
\newcommand{\tb}{\ensuremath{\mathtt{b}}}
\newcommand{\tc}{\ensuremath{\mathtt{c}}}

\newcommand{\ncc}[1]{\newcommand{#1}}
\newcommand{\rnc}[1]{\renewcommand{#1}}

\ncc{\myparagraph}[1]{\textbf{#1.}}

\rnc{\leq}{\ensuremath{\leqslant}}
\rnc{\geq}{\ensuremath{\geqslant}}

\rnc{\le}{\leq}
\rnc{\ge}{\geq}

\ncc{\isdef}{\ensuremath{:=}}
\ncc{\deff}{\isdef}

\ncc{\set}[1]{\ensuremath{\{#1\}}}
\ncc{\setsize}[1]{\ensuremath{|#1|}}
\ncc{\Setsize}[1]{\ensuremath{\big|#1\big|}}
\ncc{\Set}[1]{\ensuremath{\big\{#1\big\}}}
\ncc{\setc}[2]{\set{#1 \ | \ #2}}
\ncc{\Setc}[2]{\Set{#1 \ | \ #2}}

\ncc{\aufgerundet}[1]{\ensuremath{\lceil #1 \rceil}}
\ncc{\abgerundet}[1]{\ensuremath{\lfloor #1 \rfloor}}
\ncc{\restrict}[2]{\ensuremath{{#1}_{|#2}}}
\ncc{\extend}[3]{\ensuremath{{#1}\tfrac{#3}{#2}}}

\ncc{\dcup}{\ensuremath{\dot\cup}}

\ncc{\bigoh}{O}
\ncc{\bigOh}{\bigoh}

\ncc{\potenzmengeof}[1]{\ensuremath{\mathcal{P}({#1})}}
\ncc{\Pot}[1]{\potenzmengeof{#1}}

\ncc{\ov}[1]{\ensuremath{\overline{#1}}}

\ncc{\NN}{\ensuremath{\mathbb{N}}}
\ncc{\NNpos}{\ensuremath{\NN_{\scriptscriptstyle\geq 1}}}
\ncc{\RR}{\ensuremath{\mathbb{R}}}
\ncc{\RRpos}{\ensuremath{\RR_{\scriptscriptstyle\geq 0}}}
\ncc{\QQ}{\ensuremath{\mathbb{Q}}}
\ncc{\QQpos}{\ensuremath{\QQ_{\scriptscriptstyle\geq 0}}}

\ncc{\und}{\ensuremath{\wedge}}
\ncc{\Und}{\ensuremath{\bigwedge}}
\ncc{\oder}{\ensuremath{\vee}}
\ncc{\Oder}{\ensuremath{\bigvee}}
\ncc{\nicht}{\ensuremath{\neg}}
\ncc{\impl}{\ensuremath{\to}}
\ncc{\gdw}{\ensuremath{\leftrightarrow}}

\ncc{\free}{\ensuremath{\textrm{\upshape free}}}
\ncc{\quant}{\ensuremath{\textrm{\upshape quant}}}
\ncc{\ar}{\ensuremath{\operatorname{ar}}}

\ncc{\Dom}[1]{\ensuremath{\textup{dom}(#1)}}
\ncc{\Rg}[1]{\ensuremath{\textup{rg}(#1)}}

\ncc{\Structure}[1]{\ensuremath{\mathcal{#1}}}
\ncc{\A}{\Structure{A}}
\ncc{\B}{\Structure{B}}
\ncc{\C}{\Structure{C}}

\ncc{\isom}{\ensuremath{\cong}}
\ncc{\joinable}{\ensuremath{\sim}}

\ncc{\querycont}{\ensuremath{\sqsubseteq}}
\ncc{\eval}[2]{\ensuremath{\llbracket#1\rrbracket^{#2}}}
\ncc{\semantik}[1]{\ensuremath{\left\llbracket#1\right\rrbracket}}
\ncc{\sem}[1]{\semantik{#1}}

\ncc{\Yes}{\texttt{yes}}
\ncc{\No}{\texttt{no}}
\ncc{\True}{\ensuremath{\texttt{true}}}
\ncc{\False}{\ensuremath{\texttt{false}}}

\rnc{\phi}{\varphi}

\ncc{\emptytuple}{\ensuremath{()}}
\ncc{\emptyword}{\ensuremath{\varepsilon}}

\ncc{\proj}{\ensuremath{\pi}}
\ncc{\select}{\ensuremath{\sigma}}
\ncc{\union}{\ensuremath{\cup}}
\ncc{\intersect}{\ensuremath{\cap}}
\ncc{\join}{\ensuremath{\bowtie}}
\ncc{\eqconstr}{\ensuremath{\varsigma^{=}}}
\ncc{\streq}{\eqconstr}

\ncc{\X}{\ensuremath{\varset}}
\ncc{\Y}{\ensuremath{\mathcal{Y}}}
\ncc{\Z}{\ensuremath{\mathcal{Z}}}

\ncc{\vaFont}[1]{\ensuremath{#1}}
\ncc{\vaM}{\vaFont{M}}

\ncc{\tuple}{\ensuremath{t}}

\ncc{\x}{\varsx}
\ncc{\y}{\varsy}
\ncc{\z}{\varsz}

\ncc{\Letters}{\ensuremath{\textsf{letters}}}
\ncc{\suff}{\ensuremath{\textsf{suff}}}

\rnc{\TestProb}{\mathsf{ModelChecking}}
\DeclareMathOperator{\FunctProb}{\mathsf{Functionality}}

\newenvironment{mea}{\begin{enumerate}[(a)]}{\end{enumerate}}


\usepackage{hyperref}
\theoremstyle{plain} %\crefname{satz}{Satz}{S\"atze}
\begin{document}

\title[REFL-SPANNERS: A REGULAR APPROACH TO CORE SPANNERS]{Refl-Spanners: A Purely Regular Approach to Non-Regular Core Spanners}
\titlecomment{This is the full version of the article~\cite{SchmidSchweikardt2021}. The first author has been funded by the German Research Foundation (Deutsche Forschungsgemeinschaft, DFG) -- project number 416776735 (gef\"ordert durch die Deutsche Forschungsgemeinschaft (DFG) -- Projektnummer 416776735).
The second author has been partially supported by the ANR project EQUUS ANR-19-CE48-0019; funded by the Deutsche Forschungsgemeinschaft (DFG,  German Research Foundation) – project number 431183758 (gef\"ordert durch die Deutsche Forschungsgemeinschaft (DFG) -- Projektnummer 431183758).\\
\indent \emph{2012 ACM Subject Classification}: Information systems $\rightarrow$ Information retrieval; Theory of computation $\rightarrow$ Automata extensions; Theory of computation $\rightarrow$ Regular languages; Theory of computation $\rightarrow$ Design and analysis of algorithms; Theory of computation $\rightarrow$ Database query languages (principles)}

\author[M.\ L.\ Schmid]{Markus L.\ Schmid\lmcsorcid{0000-0001-5137-1504}}	%optional
\author[N.\ Schweikardt]{Nicole Schweikardt\lmcsorcid{0000-0001-5705-1675}}	%required

\address{Humboldt-Universit\"at zu Berlin, Unter den Linden 6, D-10099, Berlin, Germany}	%optional
\email{MLSchmid@MLSchmid.de, schweikn@informatik.hu-berlin.de}  %optional

\begin{abstract}
\noindent The regular spanners (characterised by vset-automata) are closed under the algebraic operations of union, join and projection, and have desirable algorithmic properties. The core spanners (introduced by Fagin, Kimelfeld, Reiss, and Vansummeren (PODS 2013, JACM 2015) as a formalisation of the core functionality of the query language AQL used in IBM's SystemT) additionally need string-equality selections and it has been shown by Freydenberger and Holldack (ICDT 2016, Theory of Computing Systems 2018) that this leads to high complexity and even undecidability of the typical problems in static analysis and query evaluation. We propose an alternative approach to core spanners: by incorporating the string-equality selections directly into the regular language that represents the underlying regular spanner (instead of treating it as an algebraic operation on the table extracted by the regular spanner), we obtain a variant of core spanners that, while being incomparable to the full class of core spanners, arguably still covers the intuitive applications of string-equality selections for information extraction and has much better upper complexity bounds for the typical problems in static analysis and query evaluation.
\end{abstract}

\maketitle

\section{Introduction}\label{sec:intro}

The information extraction framework of \emph{document spanners} has been introduced by Fagin, Kimelfeld, Reiss, and Vansummeren~\cite{FaginEtAl2015} as a formalisation of the query language AQL, which is used in IBM's information extraction engine SystemT. A document spanner performs information extraction by mapping a \emph{document}, formalised as a word $w$ over a finite alphabet $\Sigma$, to a relation over so-called \emph{spans} of $w$, which are intervals $\spann{i}{j}$ with $0 \leq i \leq j \leq |w|+1$. \par
The document spanners (or simply \emph{spanners}, for short)
introduced in~\cite{FaginEtAl2015} follow a two-stage approach:
\emph{Primitive} spanners extract relations directly from the input
document, which are then further manipulated by using particular algebraic operations. 
As primitive spanners,~\cite{FaginEtAl2015} introduces
\emph{vset-automata} and \emph{regex-formulas}, which are variants of
nondeterministic finite automata and regular expressions. They can use meta-symbols $\open{\varsx}$ and
$\close{\varsx}$, where $\varsx$ is a \emph{variable} from a set
$\varset$ of variables, in order to bind those variables to start and
end positions of spans, therefore extracting an $|\varset|$-ary
span-relation, or a table with columns labelled by the variables in
$\varset$. 
For example, $\alpha = (\open{\varsx} (\ta \altop \tb)^*
\close{\varsx}) \cdot (\open{\varsy} (\ta^* \altop \tb^*)
\close{\varsy}) \tc^*$ is a regex-formula and it describes a spanner
$\llbracket \alpha \rrbracket$ by considering for a given word $w$ all
possibilities of how $w$ can be generated by $\alpha$ and for each
such generation of $w$, the variables $\varsx$ and $\varsy$ extract
the spans that correspond to those subwords of $w$ that are generated
by the subexpressions $\open{\varsx} (\ta \altop \tb)^*
\close{\varsx}$ and $\open{\varsy} (\ta^* \altop \tb^*)
\close{\varsy}$, respectively. This means that $\llbracket \alpha
\rrbracket(w)$ is a binary relation over $w$'s spans. For example, on
input $w = \ta \tb \ta \ta \tc$, we have $\llbracket \alpha
\rrbracket(w) = \{(\spann{1}{3}, \spann{3}{5}), (\spann{1}{4},
\spann{4}{5}), (\spann{1}{5}, \spann{5}{5})\}$, since $\alpha$ can
generate $\open{\varsx} \ta \tb \close{\varsx} \open{\varsy} \ta \ta
\close{\varsy} \tc$, $\open{\varsx} \ta \tb \ta \close{\varsx}
\open{\varsy}  \ta \close{\varsy} \tc$ and $\open{\varsx} \ta \tb  \ta
\ta \close{\varsx} \open{\varsy} \close{\varsy} \tc$. The rows of the
extracted relation are also called \emph{span-tuples}. Vset-automata follow the same principle, but take the form of
nondeterministic finite automata.
It is known that, with respect to defining spanners, vset-automata are
more expressive than regex-formulas.
The class of spanners expressible by vset-automata are called
\emph{regular spanners}; for the sake of presentation, we denote this class of regular spanners by $\regspanners$ for the remainder of this introduction (there are different ways of characterising \emph{regular spanners} and also different semantics (see~\cite{MaturanaEtAl2018, FaginEtAl2015}); these aspects shall be discussed in more detail below).\par
The algebraic operations used for further manipulating the extracted
span-relations comprise the union $\cup$, natural join $\bowtie$,
projection $\pi$ (with the obvious meaning) and string-equality
selection $\varsigma^{=}_{\mathcal{Z}}$. The latter is a unary
operator that is parameterised by a set $\mathcal{Z} \subseteq
\varset$ of variables, and it selects exactly those rows of the table
for which all spans of columns in $\mathcal{Z}$ refer to (potentially different occurrences of) the same subwords of $w$.  \par
The \emph{core spanners} (capturing the \emph{core} of SystemT's query language AQL) introduced in~\cite{FaginEtAl2015} are defined as $\regspanners^{\{\union, \bowtie, \pi, \varsigma^{=}\}}$, i.\,e., the closure of regular spanners under the operations $\union$, $\bowtie$, $\pi$ and $\varsigma^{=}$ (these operations are interpreted as operations on spanners in the natural way). A central result of~\cite{FaginEtAl2015} is that the operations $\union$, $\bowtie$ and $\pi$ can be directly incorporated into the regular spanners, i.\,e., $\regspanners^{\{\union, \bowtie, \pi\}} = \regspanners$. This is due to the fact that regular spanners are represented by finite automata and therefore the closure properties for regular languages carry over to regular spanners by similar automaton constructions. 
This also holds in the case of so-called \emph{schemaless semantics}~\cite{MaturanaEtAl2018} (i.\,e., variables in span-tuples can be undefined).  However, as soon as we also consider the operator of string-equality selection, the picture changes considerably. \par
In terms of expressive power, it can be easily seen that not all core spanners are regular spanners, simply because for all regular spanners $S$ the language $\{w \in \Sigma^* \mid S(w) \neq \emptyset\}$ is regular, which is not necessarily the case for core spanners. As shown in~\cite{FaginEtAl2015}, we can nevertheless represent any core spanner $S \in \regspanners^{\{\union, \bowtie, \pi, \varsigma^{=}\}}$ in the form $\pi_{\mathcal{Y}} \varsigma_{\mathcal{Z}_1}^{=} \varsigma_{\mathcal{Z}_2}^{=} \ldots \varsigma_{\mathcal{Z}_k}^{=} (S')$ for a regular spanner $S'$ (this is called the \emph{core-simplification lemma} in~\cite{FaginEtAl2015}).\par
Regular spanners have excellent algorithmic properties: enumerating $S(w)$ can be done with linear preprocessing and constant delay, even if the spanner is given as vset-automaton (see~\cite{AmarilliEtAl2021, FlorenzanoEtAl2020}), while spanner containment or equivalence is generally decidable, and can even be decided efficiently if we additionally require the spanner to be represented by a certain deterministic vset-automaton (see~\cite{DoleschalEtAl2019}).\footnote{As is common in the literature, the input vset-automata are assumed to be \emph{sequential}, which means that every accepting run describes a valid span-tuple; it is well known that dropping this natural requirement yields intractability (see, e.\,g.,~\cite{AmarilliEtAl2021}).} However, in terms of complexity, we have to pay a substantial price for adding string-equality selections to regular spanners. It has been shown in~\cite{FreydenbergerHolldack2018} that for core spanners the typical problems of query evaluation and static analysis are $\npclass$- or $\pspaceclass$-hard, or even undecidable (see Table~\ref{comparisonTable}). \par
The results from~\cite{FreydenbergerHolldack2018} identify features that are, from an intuitive point of view, sources of complexity for core spanners. Thus, the question arises whether tractability can be achieved by restricting core spanners accordingly. We shall illustrate this with some examples. \par
Consider a regex formula $\alpha = \open{\varsx_1} \Sigma^* \close{\varsx_1} \open{\varsx_2} \Sigma^* \close{\varsx_2} \ldots \open{\varsx_n} \Sigma^* \close{\varsx_n}$. Then checking, for some $\mathcal{Z}_1, \mathcal{Z}_2, \ldots, \mathcal{Z}_k \subseteq \{\varsx_1, \varsx_2, \ldots, \varsx_n\}$, whether the empty tuple is in $(\pi_{\emptyset} \varsigma_{\mathcal{Z}_1}^{=} \varsigma_{\mathcal{Z}_2}^{=} \ldots \varsigma_{\mathcal{Z}_k}^{=}(\llbracket \alpha \rrbracket))(w)$, is identical to checking whether $w$ can be factorised into $n$ factors such that for each $\mathcal{Z}_i$ all factors that correspond to the variables in $\mathcal{Z}_i$ are the same. This is the \emph{pattern matching problem with variables} (or the \emph{membership problem for pattern languages}), a well-known $\npclass$-complete problem (see, e.\,g.,~\cite{ManeaSchmid2019}). However, checking for a (non-empty) span-tuple $t$ whether it is in $(\varsigma_{\mathcal{Z}_1}^{=} \varsigma_{\mathcal{Z}_2}^{=} \ldots \varsigma_{\mathcal{Z}_k}^{=}(\llbracket \alpha \rrbracket))(w)$ can be easily done in polynomial time, since the task of checking the existence of a suitable factorisation boils down to the task of evaluating a factorisation that is implicitly given by $t$. Hence, instead of blaming the string-equality selections for intractability, we could as well blame the projection operator. Can we achieve tractability by restricting projections instead of string-equality selections?\par
Another feature that yields intractability is that we can use string-equality selections in order to concisely express the \emph{intersection non-emptiness of regular languages} (a well-known $\pspaceclass$-complete problem). For example, let $r_1, r_2, \ldots, r_n$ be some regular expressions, and let $\alpha = \open{\varsx_1} r_1 \close{\varsx_1} \open{\varsx_2} r_2 \close{\varsx_2} \ldots \open{\varsx_n} r_n \close{\varsx_n}$. Then there is a word $w$ with $(\varsigma_{\{\varsx_1, \varsx_2, \ldots, \varsx_n\}}^{=} (\llbracket \alpha \rrbracket))(w) \neq \emptyset$ if and only if $\bigcap^n_{i = 1} \lang(r_i) \neq \emptyset$. So string-equality selections do not only check whether the same subword has several occurrences, but also, as a ``side-effect'', check membership of this repeated subword in the intersection of several regular languages. Can we achieve tractability by somehow limiting the power of string-equality selections to the former task?\par

\begin{table}
\begin{tabular}{llll}
\textbf{Problem} & \textbf{Regular sp.} & \textbf{Refl-sp.} & \textbf{Core sp.} \cite{FreydenbergerHolldack2018}\\\toprule
$\TestProb$& $\bigO(|w|\cdot|M| \log |\X|)$ & $\bigO(|w|\cdot|M| \log |\X|)$ \hfill [T.~\ref{checkingTheorem}] & $\npclass$-c \\
$\NonemptProb$& $\bigO(|w|\cdot|M|)$ & $\npclass$-c \hfill [T.~\ref{thm:nonempRefl}] &  $\npclass$-h \\\midrule
$\SatProb$& $\bigO(|M|)$ & $\bigO(|M|)$ \hfill [T.~\ref{thm:SatAndHier}] & $\pspaceclass$-c\\
$\ContProb$& $\pspaceclass$-c~\cite{MaturanaEtAl2018} & $\pspaceclass$-c (for str.~ref.) \hfill [T.~\ref{mainContainmentDecidabilityTheorem}] & undec.\\
$\EquProb$& $\pspaceclass$-c~\cite{MaturanaEtAl2018} & $\pspaceclass$-c (for str.~ref.) \hfill [T.~\ref{mainContainmentDecidabilityTheorem}] & undec.\\
$\HierProb$& $\bigO(|M|\cdot|\varset|^3)$ & $\bigO(|M|\cdot|\varset|^3)$ \hfill [T.~\ref{thm:SatAndHier}] & $\pspaceclass$-c\\\bottomrule
\end{tabular}
\caption{Comparison of decision problems of regular spanners, core
  spanners and refl-spanners. A formal definition of the problems can
  be found in Section~\ref{sec:decisionProblems}. In the case of
  regular spanners and refl-spanners, the input spanner is represented
  by an $\NFA$ $M$ (accepting a subword-marked language or a ref-language, respectively). The abbreviation ``str.~ref.'' means
  \emph{strongly reference extracting}, a restriction for
  refl-spanners to be formally defined in
  Section~\ref{sec:decisionProblems}. The bounds for $\TestProb$ are
  under the (rather weak) assumption $|\varset| = \bigO(|w|)$; without
  this assumption, $|w|$ has to be replaced with $|w| + |\varset|$.
}
\label{comparisonTable}
\end{table}

A third observation is that by using string-equality selections on \emph{overlapping} spans, we can use core spanners to express rather complex word-combinatorial properties. For example, we can even express word equations as core spanners (see~\cite[Proposition~3.7, Example~3.8, Theorem~3.13]{FreydenbergerHolldack2018} for details). Can we achieve tractability by requiring all variables that are subject to string-equality selections to extract only pairwise non-overlapping spans?

\subsection{Our Contribution} 

We introduce \emph{refl-spanners} (based on \emph{regular
  ref-languages}), a new formalism for spanners that properly extends
regular spanners, describes a large class of core spanners, and has
better upper complexity bounds than core spanners. Moreover, the
formalism is purely based on regular language description
mechanisms. The main idea is a paradigm shift in the two-stage
approach of core spanners: instead of extracting a span-relation with
a regular spanner and then applying string-equality selections on it,
we handle string-equality selections directly with the finite
automaton (or regular expression) that describes the regular
spanner. However, checking the equality of unbounded factors in
strings is a task that, in most formalisms, can be considered highly
``non-regular'' (the well-known \emph{copy-language} $\{ww \mid w \in
\Sigma^*\}$ is a textbook example for demonstrating the limits of
regular languages in this regard). We deal with this obstacle by
representing the factors that are subject to string-equality
selections as \emph{variables} in the regular language. For example,
while $L = \{\ta^n \tb \ta^n \mid n \geq 0\}$ is non-regular, the
language $L' = \{\open{\varsx} \ta^n \close{\varsx} \tb \varsx \mid n
\geq 0\}$ can be interpreted as a regular description of $L$ by means
of meta-symbols $\open{\varsx}$ and $\close{\varsx}$ to \emph{capture}
a factor, and a meta-symbol $\varsx$ to \emph{copy} or
\emph{reference} the captured factor (conceptionally, this is similar
to so-called backreferences in practical regular expressions;
see~\cite{FaginEtAl2015, FreydenbergerHolldack2018} for comparisons of
regular expressions with backreferences and core spanners). In
particular, all words of $L$ can be easily obtained from the words of
$L'$ by simply replacing the occurrence of $\varsx$ with the factor it
refers to. As long as core spanners use string-equality selections in
a not too complicated way, this simple formalism seems also to be
suited for describing particular core spanners, e.\,g., the core spanner $\pi_{\{\varsx, \varsy\}}\varsigma^{=}_{\{\varsx, \varsx'\}} \varsigma^{=}_{\{\varsy, \varsy'\}} (\llbracket \alpha \rrbracket)$ with $\alpha = \open{\varsx} \ta^* \tb \open{\varsy} \tc \close{\varsx} \tb^* \open{\varsx'} \ta^* \tb \tc \close{\varsx'} \close{\varsy} \open{\varsy'} \tc \tb^* \ta^* \tb \tc \close{\varsy'}
$ could be represented as $\llbracket \open{\varsx} \ta^* \tb \open{\varsy} \tc \close{\varsx} \tb^* \varsx \close{\varsy} \varsy \rrbracket$. \par
The class of refl-spanners can now informally be described as the
class of all spanners that can be represented by a regular language
over the alphabet $\Sigma \cup \varset \cup \{\open{\varsx},
\close{\varsx} \mid \varsx \in \varset\}$ that has the additional
property that the meta-symbols $\varset \cup \{\open{\varsx},
\close{\varsx} \mid \varsx \in \varset\}$ are ``well-behaved'' in the
sense that each word describes a valid span-tuple (one of this paper's main conceptional contributions is to formalise this idea in a sound way). \par
The refl-spanner formalism automatically avoids exactly the features of core spanners that we claimed above to be sources of complexity. More precisely, refl-spanners cannot project out variables, which means that they cannot describe the task of checking the existence of some complicated factorisation. Furthermore, it can be easily seen that in the refl-spanner formalism, we cannot describe \emph{intersection non-emptiness of regular languages} in a concise way, as is possible by core spanners. Finally, we can only have overlaps with respect to the spans captured by $\open{\varsx} \ldots \close{\varsx}$, but all references $\varsx$ represent pairwise non-overlapping factors, which immediately shows that we cannot express word equations as core spanners can. This indicates that refl-spanners are restricted in terms of expressive power, but it also gives hope that for refl-spanners we can achieve better upper complexity bounds for the typical decision problems compared to core spanners, and, in fact, this is the case (see Table~\ref{comparisonTable}).\par

It is obvious that not all core spanners can be represented as refl-spanners, but we can nevertheless show that a surprisingly large class of core spanners can be handled by the refl-spanner formalism. 
Recall that the core simplification lemma from~\cite{FaginEtAl2015} states that, in every core spanner $S \in \regspanners^{\{\cup, \pi, \bowtie, \varsigma^{=}\}}$, we can ``push'' all applications of $\cup$ and $\bowtie$ into the automaton that represents the regular spanner, leaving us with an expression $\pi_{\mathcal{Y}}\varsigma_{\mathcal{Z}_1}^{=} \varsigma_{\mathcal{Z}_2}^{=} \ldots \varsigma_{\mathcal{Z}_k}^{=}(M)$ for an automaton $M$ that represents a regular spanner. We can show that if the string-equality selections $\varsigma_{\mathcal{Z}_1}^{=} \varsigma_{\mathcal{Z}_2}^{=} \ldots \varsigma_{\mathcal{Z}_k}^{=}$ apply to a set of variables that never capture overlapping spans, then we can even ``push'' all string-equality selections into $M$, turning it into a representation of a refl-spanner that ``almost'' describes $S$: in order to get $S$, we only have to merge certain columns into a single one by creating the fusion of the corresponding spans. 

\subsection{Related Work} 

Spanners have recently received a lot of attention~\cite{FaginEtAl2015, FreydenbergerEtAl2018, PeterfreundEtAl2019, AmarilliEtAl2021, MaturanaEtAl2018, FlorenzanoEtAl2020, PeterfreundEtAl2019_2, Freydenberger2019, FreydenbergerHolldack2018,FreydenbergerThompson2020, Peterfreund2019PhD, SchmidSchweikardt2022, SchmidSchweikardt2021_PODS, DBLP:journals/sigmod/AmarilliBMN20, AmarilliEtAl2022, DoleschalEtAl2020, DoleschalEtAl2021, FreydenbergerThompson2022, MunozRiveros2023}. However, as it seems, most of the recent progress on document spanners concerns regular spanners. For example, it has recently been shown that results of regular spanners can be enumerated with linear preprocessing and constant delay~\cite{AmarilliEtAl2021, FlorenzanoEtAl2020}, the paper~\cite{MaturanaEtAl2018} is concerned with different semantics of regular spanners and their expressive power, and~\cite{PeterfreundEtAl2019} investigates the evaluation of relational algebra expressions over regular spanners.\par
Papers that are concerned with string-equality selection are the following. In~\cite{FreydenbergerHolldack2018} many negative results for core spanner evaluation are shown. By presenting a logic that exactly covers core spanners, the work~\cite{Freydenberger2019} answers questions on the expressive power of core spanners. That datalog over regular spanners covers the whole class of core spanners is shown in~\cite{PeterfreundEtAl2019_2}. In~\cite{FreydenbergerEtAl2018}, the authors consider conjunctive queries on top of regular spanners and, among mostly negative results, they also show the positive result that such queries with equality-selections can be evaluated efficiently if the number of string equalities is bounded by a constant. The paper~\cite{FreydenbergerThompson2020} investigates the dynamic descriptive complexity of regular spanners and core spanners. While all these papers contribute deep insights with respect to document spanners, positive algorithmic results for the original core spanners from~\cite{FaginEtAl2015} seem scarce and the huge gap in terms of tractability between regular and core spanners seems insufficiently bridged by tractable fragments of core spanners.\par
A rather recent paper that also deals with non-regular document
spanners is~\cite{Peterfreund2021}. However, the non-regular aspect
of~\cite{Peterfreund2021} does not consist in string-equality
selections, but rather that spanners are represented by context-free
language descriptors (in particular grammars) instead of regular
ones. The spanner class from~\cite{Peterfreund2021} is incomparable with core spanners, and the main focus of~\cite{Peterfreund2021} is on enumeration.

\subsection{Differences to the Preliminary Conference Version}\label{subsection:DifferencesToPreviousVersion}

This paper is the full and substantially revised version of the extended abstract~\cite{SchmidSchweikardt2021}, presented at the 24th International Conference on Database Theory (ICDT~2021). We will briefly describe the main changes.\par
The upper and lower complexity bounds for evaluation and static analysis problems of refl-spanners presented in Section~\ref{sec:decisionProblems} (see also Table~\ref{comparisonTable}) have been improved as follows. The upper bound for $\TestProb$ for refl-spanners reported in~\cite{SchmidSchweikardt2021} has been improved substantially, and, in fact, coincides with the (to our knowledge) best known upper bound for $\TestProb$ for regular spanners (note that~\cite{SchmidSchweikardt2021} did not mention any upper bound for $\TestProb$ for regular spanners). The upper bounds for $\ContProb$ and $\EquProb$ for strongly reference extracting refl-spanners have been improved from membership in $\expspaceclass$ to $\pspaceclass$-completeness. \par
With respect to the results on the expressive power of refl-spanners, this full version serves also as an erratum to~\cite{SchmidSchweikardt2021}. The main result with respect to the expressive power of refl-spanners holds as stated in~\cite[Theorem 6.4]{SchmidSchweikardt2021}, and is proven here in full detail (see Theorem~\ref{coreSpannersToReflSpannersTheorem}). Unfortunately, the secondary result~\cite[Theorem 6.5]{SchmidSchweikardt2021} does not hold as stated in~\cite{SchmidSchweikardt2021}. This error is corrected here.

\subsection{Organisation}
The rest of the paper is structured as follows. 
Section~\ref{sec:prelim} fixes the basic notation concerning spanners and lifts the core-simplification lemma of \cite{FaginEtAl2015} to the schemaless case.
In Section~\ref{section:declarative}, we develop a simple declarative approach to spanners by establishing a natural one-to-one correspondence between spanners and so-called subword-marked languages.
In Section~\ref{section:refl} we extend the concept of subword-marked
languages in order to describe spanners with string-equality
selections which we call refl-spanners.
Section~\ref{sec:decisionProblems} is devoted to the complexity of evaluation and static analysis problems for refl-spanners. 
Section~\ref{sec:ExpressivePower} studies the expressive power of refl-spanners. We conclude the paper in Section~\ref{sec:conclusions}.

\section{Preliminaries}\label{sec:prelim}

Let $\mathbb{N} = \{1, 2, 3, \ldots\}$ and $[n] = \{1, 2, \ldots, n\}$ for $n \in \mathbb{N}$.
For a (partial) mapping $f \colon X \to Y$, we write $f(x) = \undefin$ for some $x \in X$ to denote that $f(x)$ is not defined; we also set $\Dom{f} = \{x \mid f(x) \neq \undefin\}$.
By $\mathcal{P}(A)$ we denote the power set of a set $A$, and
$A^+$ denotes the set of non-empty words over $A$, and $A^* = A^+ \cup \{\eword\}$, where $\eword$ is the empty word. For a word $w \in A^*$, $|w|$ denotes its length (in particular, $|\eword| = 0$), and for every $b \in A$, $|w|_{b}$ denotes the number of occurrences of $b$ in $w$. 
Let $A$ and $B$ be alphabets with $B \subseteq A$, and let $w \in
A^*$. Then $\erasemorphism{B} \colon A \to A \cup \{\eword\}$ is a mapping
with $\erasemorphism{B}(b) = \eword$ if $b \in B$ and
$\erasemorphism{B}(b) = b$ if $b \in A \setminus B$; we also write
$\erasemorphism{B}$ to denote the natural extension of
$\erasemorphism{B}$ to a morphism $A^* \to A^*$. Intuitively, for
$w\in A^*$, the word $\erasemorphism{B}(w)$ is obtained from $w$ by
erasing all occurrences of letters in $B$. 
Technically, $\erasemorphism{B}$ depends on the alphabet $A$, but whenever we use $\erasemorphism{B}(w)$ we always assume that $\erasemorphism{B} \colon A \to A \cup \{\eword\}$ for some alphabet $A$ with $w \in A^*$.

\subsection{Regular Language Descriptors}

For an alphabet $\Sigma$, the set $\RE_{\Sigma}$ of \emph{regular
  expressions} (\emph{over $\Sigma$}) is defined as usual:
$\emptyset$ is in $\RE_{\Sigma}$ with $\lang(\emptyset) = \emptyset$.
Every $a \in \Sigma \cup \{\eword\}$ is in $\RE_{\Sigma}$ with
$\lang(a) = \{a\}$. For $r, s \in \RE_{\Sigma}$, $(r \cdot s), (r \altop s), (r)^+ \in \RE_{\Sigma}$ with $\lang((r \cdot s)) = \lang(r) \cdot \lang(s)$, $\lang((r \altop s)) = \lang(r) \cup \lang(s)$, $\lang((r)^+) = (\lang(r))^+$. For $r \in \RE_{\Sigma}$, we use $r^*$ as a shorthand form for $((r)^+ \altop \eword)$, and we usually omit the operator `$\cdot$', i.\,e., we use juxtaposition. For the sake of readability, we often omit parentheses (and use the usual
precedences of operands), if this does not cause ambiguities.\par
A \emph{nondeterministic finite automaton} ($\NFA$ for short) is a
tuple $M = (Q, \Sigma, \delta, q_0, F)$ with a finite set $Q$ of
states, a finite alphabet $\Sigma$, a start state $q_0$, a set
$F\subseteq Q$ of accepting states, and a transition function $\delta \colon Q \times (\Sigma \cup \{\eword\}) \to \mathcal{P}(Q)$. We also interpret $\NFA$ as directed, edge-labelled graphs in the obvious way. A word $w \in \Sigma^*$ is accepted by $M$ if there is a path from $q_0$ to some $q_f \in F$ that is labelled by $w$; $\lang(M)$ is the accepted language, i.\,e., the set of all accepted words. The size $|M|$ of an $\NFA$ is measured as $|Q| + |\delta|$. However, we will mostly consider $\NFA$ with constant out-degree, which means that $|M| = \bigO(|Q|)$. For a language descriptor $D$ (e.\,g., an $\NFA$ or a regular expression), we denote by $\lang(D)$ the language defined by $D$. The class of languages described by $\NFA$ or regular expressions is the class of regular languages, denoted by $\regLanguages$.

\subsection{Spans and Spanners} \label{sec:spanSpanners}

For a word $w \in \Sigma^*$ and for every $i, j \in [|w| {+} 1]$ with
$i\leq j$, $\spann{i}{j}$ is a \emph{span of $w$} and its
\emph{value}, denoted by $w\spann{i}{j}$, is the substring of $w$ from
symbol $i$ to symbol $j{-}1$. In particular, $w\spann{i}{i} = \eword$
(this is called an \emph{empty span}) and $w\spann{1}{|w|{+}1} = w$. By
$\spans(w)$, we denote the set of spans of $w$, and by $\spans$ we
denote the set 
$\setc{\spann{i}{j}}{i,j\in\NN,\ i\leq j}$ (elements from $\spans$ shall simply be called \emph{spans}). A span
$\spann{i}{j}$ can also be interpreted as the set $\{i, i{+}1, \ldots,
j{-}1\}$ and therefore we can use set operations for spans (note,
however, that the union of two spans is not necessarily a span
anymore). Two spans $s = \spann{i}{j}$ and $s' = \spann{i'}{j'}$ are
\emph{equal} if $s = s'$ (i.e., $i=i'$ and $j=j'$), they are \emph{disjoint} if $j \leq i'$ or $j' \leq i$ and they are \emph{non-overlapping} if they are equal or disjoint. Note that $s$ and $s'$ being disjoint is sufficient but not necessary for $s \cap s' = \emptyset$, e.\,g., $\spann{3}{6}$ and $\spann{5}{5}$ are not disjoint, but $\spann{3}{6} \cap \spann{5}{5} = \emptyset$. \par
For a finite set of variables $\varset$, an \emph{$(\varset, w)$-tuple} (also simply called \emph{span-tuple}) is a partial function $\varset \to \spans(w)$,
and a  \emph{$(\varset, w)$-relation}
is a set of $(\varset, w)$-tuples. For simplicity, we usually denote $(\varset, w)$-tuples in tuple-notation, for which we assume an order on $\varset$ and use the symbol `$\undefin$' for undefined variables, e.\,g., $(\spann{1}{5}, \undefin, \spann{5}{7})$ describes a $(\{\varsx_1, \varsx_2, \varsx_3\}, w)$-tuple that maps $\varsx_1$ to $\spann{1}{5}$, $\varsx_3$ to $\spann{5}{7}$, and is undefined for $\varsx_2$. Since the dependency on the word $w$ is often negligible, we also use the term \emph{$\varset$-tuple} or \emph{$\varset$-relation} to denote an $(\varset, w)$-tuple or $(\varset, w)$-relation, respectively.\par
An $(\varset, w)$-tuple $t$ is \emph{functional} if it is a total function, $t$ is \emph{hierarchical} if, for every $\varsx, \varsy \in \Dom{t}$, $t(\varsx) \subseteq t(\varsy)$ or $t(\varsy) \subseteq t(\varsx)$ or $t(\varsx) \cap t(\varsy) = \emptyset$, and $t$ is \emph{non-overlapping} if, for every $\varsx, \varsy \in \Dom{t}$, $t(\varsx)$ and $t(\varsy)$ are non-overlapping. An $(\varset, w)$-relation is \emph{functional}, \emph{hierarchical} or \emph{non-overlapping}, if all its elements are functional, hierarchical or non-overlapping, respectively. \par
A \emph{spanner} (\emph{over terminal alphabet $\Sigma$ and variables $\varset$}) is a function that maps every $w \in \Sigma^*$ to an $(\varset, w)$-relation (note that the empty relation $\emptyset$ is also a valid image of a spanner). 

\begin{exa}\label{firstSpannerExample}
Let $\Sigma = \{\ta, \tb\}$ and let $\varset = \{\varsx, \varsy, \varsz\}$. Then the function $S$ that maps words $w \in \Sigma^*$ to the $(\varset, w)$-relation $\{(\spann{1}{i}, \spann{i}{i + 1}, \spann{i+1}{|w| + 1}) \mid 1 \leq i < |w|, w\spann{i}{i+1} = \tb\}$ is a spanner. For example, $S(\ta \tb \ta \tb \tb \ta \tb) = \{t_1, t_2, t_3, t_4\}$ with $t_1 = (\spann{1}{2}, \spann{2}{3}, \spann{3}{8})$, $t_2 = (\spann{1}{4}, \spann{4}{5}, \spann{5}{8})$, $t_3 = (\spann{1}{5}, \spann{5}{6}, \spann{6}{8})$ and $t_4 = (\spann{1}{7}, \spann{7}{8}, \spann{8}{8})$.
\end{exa}

Let $S_1$ and $S_2$ be spanners over $\Sigma$ and $\varset$. Then $S_1$ and $S_2$ are said to be \emph{equal} if, for every $w \in \Sigma^*$, $S_1(w) = S_2(w)$ (this coincides with the usual equality of functions and shall also be denoted by $S_1 = S_2$).  We say that $S_2$ \emph{contains} $S_1$, written as $S_1 \subseteq S_2$, if, for every $w \in \Sigma^*$, we have $S_1(w) \subseteq S_2(w)$. A spanner $S$ over $\Sigma$ and $\varset$ is \emph{functional}, \emph{hierarchical} or \emph{non-overlapping} if, for every $w$, $S(w)$ is functional, hierarchical or non-overlapping, respectively. Note that, for span-tuples, span-relations and spanners, the non-overlapping property implies hierarchicality. \par
Next, we define operations on spanners.
The \emph{union} $S_1\cup S_2$ of two spanners $S_1$ and $S_2$ over $\Sigma$ and $\varset$ is defined via  $(S_1 \cup S_2)(w) = S_1(w) \cup S_2(w)$ for all
 $w \in \Sigma^*$.\par
To define the \emph{natural join} $S_1\join S_2$ we need further
notation:
Two $(\varset, w)$-tuples $t_1$ and $t_2$ are \emph{compatible}
(denoted by $t_1 \sim t_2$) if $t_1(\varsx) = t_2(\varsx)$ for every
$\varsx \in \Dom{t_1} \cap \Dom{t_2}$. For compatible $(\varset,
w)$-tuples $t_1$ and $t_2$, the $(\varset, w)$-tuple $t_1 \bowtie t_2$
is defined by $(t_1 \bowtie t_2)(\varsx) = t_i(\varsx)$ if $\varsx \in
\Dom{t_i}$ for $i\in\set{1,2}$.
For two $(\X,w)$-relations $R_1, R_2$ we let 
$R_1 \bowtie R_2 = \{t_1 \bowtie t_2 \mid t_1 \in R_1, t_2 \in R_2,
t_1 \sim t_2\}$.  
Finally, the \emph{natural join} $S_1\join S_2$ is defined via $(S_1
\join S_2)(w) = S_1(w) \join S_2(w)$ for all $w \in \Sigma^*$. 
\par
The \emph{projection} $\pi_{\mathcal{Y}}(S_1)$ for a set $\Y\subseteq\X$ is defined by letting 
$(\pi_{\mathcal{Y}} (S_1))(w)=\setc{t_{|\Y}}{t\in S_1(w)}$, where $t_{|\Y}$ is the restriction of $t$ to domain $\Dom{t}\cap \Y$.\par
The \emph{string-equality selection} $\varsigma_{\mathcal{Y}}^{=}(S_1)$ for a set $\Y\subseteq\X$ is defined by letting $(\varsigma_{\mathcal{Y}}^{=}(S_1))(w)$ contain all $t \in S_1(w)$ such that, for every $\varsx, \varsy \in \mathcal{Y} \cap \Dom{t}$ with $t(\varsx) = \spann{i}{j}$ and $t(\varsy) = \spann{i'}{j'}$, we have that $w\spann{i}{j} = w\spann{i'}{j'}$. 
Note that here we require $w\spann{i}{j} = w\spann{i'}{j'}$ only for
$\varsx, \varsy \in \mathcal{Y} \cap \Dom{t}$ instead of all $\varsx, \varsy \in \mathcal{Y}$. \par
For convenience, we omit the parentheses if we apply sequences of unary operations of spanners, e.\,g., we write $\pi_{\mathcal{Z}}\varsigma^{=}_{\mathcal{Y}_1}\varsigma^{=}_{\mathcal{Z}_1}(S)$ instead of $\pi_{\mathcal{Z}}(\varsigma^{=}_{\mathcal{Y}_1}(\varsigma^{=}_{\mathcal{Z}_1}(S)))$. For any $E = \{\mathcal{Y}_1, \mathcal{Y}_2, \ldots, \mathcal{Y}_\ell\} \subseteq \mathcal{P}(\varset)$, we also write $\varsigma^{=}_{E}(S)$ instead of $\varsigma^{=}_{\mathcal{Y}_1} \varsigma^{=}_{\mathcal{Y}_2} \ldots \varsigma^{=}_{\mathcal{Y}_\ell}(S)$, and in this case we will call $\varsigma^{=}_{E}$ a \emph{generalised} string-equality selection, or also just string-equality selection if it is clear from the context that $E \subseteq \mathcal{P}(\varset)$. 

For a class $\mathfrak{S}$ of spanners and a set $P$ of operations on spanners, $\mathfrak{S}^{P}$ (or $\mathfrak{S}^{p}$ if $P = \{p\}$) denotes the closure of $\mathfrak{S}$ under the operations from $P$.\par
Whenever formulating complexity bounds, we consider the terminal alphabet $\Sigma$ to be constant, but we always explicitly state any dependency on $|\varset|$. 

\subsection{Regular Spanners and Core Spanners}

In~\cite{FaginEtAl2015}, the class of \emph{regular spanners}, denoted
by $\regspanners$, is defined  as the class of spanners represented by
\emph{vset-automata}, and the class of \emph{core spanners} is defined
as $\corespanners = \llbracket \regexform \rrbracket^{\{\cup, \pi,
  \bowtie, \varsigma^{=}\}}$, where $\regexform$ is the class of
so-called \emph{regex-formulas} (we refer to~\cite{FaginEtAl2015} for
a formal definition of vset-automata and regex-formulas).
It is known that, with respect to defining spanners, vset-automata are
more expressive than regex-formulas;
but for
defining the class of core spanners, it does not matter whether one uses
the class of regex-formulas or the class of vset-automata, i.\,e.,
$\llbracket \regexform \rrbracket^{\{\cup, \pi, \bowtie,
  \varsigma^{=}\}} = \llbracket \vset \rrbracket^{\{\cup, \pi,
  \bowtie, \varsigma^{=}\}}$). A crucial result
from~\cite{FaginEtAl2015} is the \emph{core-simplification lemma}:
every $S \in \corespanners$ can be represented as
$\pi_{\mathcal{Y}}\varsigma^{=}_E(S')$, where $S'$ is a regular
spanner, $\Y \subseteq\X$ and $E \subseteq \mathcal{P}(\varset)$. 
The setting in~\cite{FaginEtAl2015} uses a \emph{function semantics} for spanners, i.\,e., $(\varset, w)$-tuples are always functional. In our definitions above, we allow variables in span-tuples and spanners to be undefined, i.\,e., we use partial mappings as introduced in~\cite{MaturanaEtAl2018}, and in the terminology of~\cite{PeterfreundEtAl2019}, we consider the \emph{schemaless} semantics.

In~\cite{MaturanaEtAl2018} it is shown that the classical framework
for \emph{regular spanners} with function semantics introduced
in~\cite{FaginEtAl2015} can be extended to the schemaless case,
i.\,e., vset-automata and regex-formulas are extended to the case of
schemaless semantics, and it is shown that the basic results still
hold (e.\,g., vset-automata are equally powerful as regex-formulas when considering the closure under union, projection and natural join,
i.e.,
$
\llbracket \vset \rrbracket^{\{\cup, \pi, \bowtie\}}
=
\llbracket \regexform \rrbracket^{\{\cup, \pi, \bowtie\}}
$).
However, the string-equality selection operator --- which turns regular spanners into the more powerful core spanners --- is not treated in~\cite{MaturanaEtAl2018}. Our definition of the string-equality selection operator given above extends the definition from~\cite{FaginEtAl2015} from the functional to the schemaless case by interpreting $\varsigma^{=}_{\mathcal{Y}}$ to apply only to those variables from $\mathcal{Y}$ that are in the domain of the span-tuple. This way of treating undefined variables is natural and also corresponds to how the join operator is extended to the schemaless case in~\cite{MaturanaEtAl2018}.
Due to~\cite{MaturanaEtAl2018}, we can also in the schemaless case define $\regspanners$ as the class of spanners defined by vset-automata (with schemaless semantics), and we can also define the class of core spanners with schemaless semantics as $\corespanners = \llbracket \regexform \rrbracket^{\{\cup, \pi, \bowtie, \varsigma^{=}\}}$. However, to the knowledge of the authors, the core-simplification lemma from~\cite{FaginEtAl2015} has so far not been extended to the schemaless semantics. Since we wish to apply the core-simplification lemma in the context of our results for schemaless semantics (see Theorem~\ref{coreSpannersToReflSpannersTheorem}), and since this seems to be a worthwhile task in its own right, we show that the core-simplification lemma from~\cite{FaginEtAl2015} holds verbatim for the schemaless case. For those parts of the proof's argument that are not concerned with string-equality selections, we heavily rely on the results from~\cite{MaturanaEtAl2018}. 

\begin{lem}[Core Simplification Lemma]\label{lemma:CoreSimplificationLemma}
For every $S \in \corespanners$ over $\X$ there are $S' \in \regspanners$, $\Y \subseteq \X$ and $E \subseteq \Pot{\X}$ such that $S = \proj_{\Y}\eqconstr_E(S')$.
\end{lem}

\begin{proof}
We proceed by induction on the construction of $S\in \corespanners = \llbracket \regexform \rrbracket^{\{\cup, \pi, \bowtie, \varsigma^{=}\}}$.

For the induction base we consider $S\in \llbracket \regexform \rrbracket$ over $\X$. 
By Theorem~4.4 of \cite{MaturanaEtAl2018} there exists a (hierarchical) 
vset-automaton $\vaM$ such that $S=\sem{\vaM}$.
In particular, $S\in\regspanners$. Furthermore, it is obvious that 
$S=\proj_{\X}\eqconstr_\emptyset (S)$, and hence we are done.

For the induction step we distinguish between two cases.
\medskip

\emph{Case~1:} $S$ is of the form $\proj_{\Y}(S_1)$ or $\eqconstr_{E}(S_1)$, where $S_1\in\corespanners$ over $\X$ and $\Y\subseteq\X$ or $E\subseteq\Pot{\X}$.
Applying the induction hypothesis to $S_1$, we obtain that there exist 
$S'_1\in\regspanners$, $\Y_1\subseteq\X$, and $E_1\subseteq\Pot{\X}$ such that 
$S_1=\proj_{\Y_1}\eqconstr_{E_1}(S'_1)$.

If $S$ is of the form $\proj_{\Y}(S_1)$, we are done by noting that 
$S=\proj_{\Y\cap\Y_1}\eqconstr_{E_1}(S'_1)$.

If $S$ is of the form $\eqconstr_E(S_1)$, we are done by noting that 
$S=\eqconstr_E(\proj_{\Y_1}\eqconstr_{E_1}(S'_1)) = \proj_{\Y_1}\eqconstr_{E'\union E_1}(S'_1)$ where
$E'\deff \setc{C\cap \Y_1}{C\in E}$.

\medskip
\emph{Case~2:} $S$ is of the form $(S_1\ast S_2)$ where $\ast\in\set{\union,\join}$ and $S_i\in\corespanners$ over $\X$ for each $i\in\set{1,2}$.
Applying the induction hypothesis to $S_i$ for each $i\in\set{1,2}$, we obtain that there exist $S'_i\in\regspanners$, $\Y_i\subseteq\X$, and $E_i\subseteq\Pot{\X}$ such that 
$S_i=\proj_{\Y_i}\eqconstr_{E_i}(S'_i)$.
In particular, for each $i\in\set{1,2}$ there is a vset-automaton $\vaM_i$ such that 
$S'_i=\sem{\vaM_i}$.

Let $\tilde{\Y}\deff\Y_1\cap\Y_2$. 
Let $\Z_i$ denote the set of all variables that occur in $\vaM_i$ and let $\Z'_i\deff \Z_i\setminus\tilde{\Y}$.
By suitably renaming variables we can assume w.l.o.g.\ that $\Z'_1\cap \Z'_2=\emptyset$.

Furthermore, we assume w.l.o.g.\
(this can be achieved by suitably modifying $M_i$) 
that for every $y\in\tilde{\Y}$ and for each $i\in\set{1,2}$, there is a unique $z(y,i)\in \Z_i\setminus\Y_i$ such that whenever $\vaM_i$ reads $\open{y}$, it immediately afterwards has to read $\open{z(y,i)}$, and whenever it reads $\close{y}$, this immediately follows after having read $\close{z(y,i)}$. 
This ensures that for every $w\in\Sigma^*$ and every $\tuple\in S'_i(w)$ we have 
\[
\big(\, y\in\Dom{\tuple} \iff z(y,i)\in\Dom{\tuple}\,\big)
\quad\text{and}\quad \big(\, y\in\Dom{\tuple}\ \Longrightarrow \ \tuple(y)=\tuple(z(y,i)) \,\big)\,.
\]
Therefore, we can assume w.l.o.g.\ that $E_i\subseteq\Pot{\Z'_i}$
(this can be achieved by replacing $y$ with $z(y,i)$ for every $y\in\tilde{\Y}$ that occurs in some string-equality constraint specified by $E_i$).

Recall that we assume that $S=(S_1\ast S_2)$ with $\ast\in\set{\union,\join}$.
Let 
\[
\tilde{S} \ = \ \ \proj_{\Y_1\cup\Y_2}\streq_{E_1\cup E_2}(S'_1\ast S'_2)\,.
\]
Since $S'_1,S'_2\in \regspanners$ and $\ast\in\set{\union,\join}$ we obtain from Theorem~4.5 of 
\cite{MaturanaEtAl2018} that $(S'_1\ast S'_2)\in \regspanners$.
Therefore, in order to finish the proof of Lemma~\ref{lemma:CoreSimplificationLemma} it suffices to prove the following claim.
\begin{clm}\label{claim:ClaimInCoreSimplificationProof}
$S=\tilde{S}$, i.e., $S(w)=\tilde{S}(w)$ for every $w\in\Sigma^*$.
\end{clm}
The proof is straightforward, but somewhat tedious: Consider an arbitrary $w\in\Sigma^*$ and show that $S(w)\subseteq\tilde{S}(w)$ and $S(w)\supseteq\tilde{S}(w)$. 
\medskip

\noindent
We first consider the $\subseteq$-part and choose an arbitrary $\tuple\in S(w)$. Our aim is to show that $\tuple \in\tilde{S}(w)$.

\emph{Case 1:} $\ast=\union$.
Since $S(w)=S_1(w)\cup S_2(w)$, 
there exists $i\in\set{1,2}$ such that $\tuple\in S_i(w)$.
Since $S_i=\proj_{\Y_i}\streq_{E_i}(S'_i)$, there exists $\tuple'\in S'_i(w)$ such that
$\tuple = \proj_{\Y_i}(\tuple')$ and $\tuple'$ satisfies the string-equality constraints
specified by $E_i$ --- i.e., for all $C\in E$ and all variables $u,v\in C$ with
 $u,v\in\Dom{\tuple'}$ we have $\tuple'(u)=\tuple'(v)$. 

Furthermore, from $S'_i=\sem{\vaM_i}$ we know that $\Dom{\tuple'}\subseteq \Z'_i\cup\tilde{\Y}$.
Since $E_{3-i}\subseteq \Pot{\Z'_{3-i}}$ and $(\Z'_i\cup\tilde{\Y})\cap \Z'_{3-i}=\emptyset$, 
$\tuple'$ trivially satisfies the string-equality constraints specified by $E_{3-i}$.
Therefore, $\tuple'\in (\streq_{E_1\cup E_2}(S'_1\cup S'_2))(w)$, and hence
$\tuple = \proj_{\Y_i}(\tuple') = \proj_{\Y_1\cup\Y_2}(\tuple')\in(\proj_{\Y_1\cup\Y_2}\streq_{E_1\cup E_2}(S'_1\cup S'_2))(w) =\tilde{S}(w)$.

\emph{Case 2:} $\ast={\join}$.
Since $S(w)=S_1(w)\join S_2(w)$, for each $i\in\set{1,2}$ there is $\tuple_i\in S_i(w)$ such that 
$\tuple_1\joinable\tuple_2$ and $\tuple=\tuple_1\join\tuple_2$ (i.e., $\tuple_1$ and $\tuple_2$ are compatible, and $\tuple$ is their join result).

Since $S_i=\proj_{\Y_i}\streq_{E_i}(S'_i)$, there exists $\tuple'_i\in S'_i(w)$ such that
$\tuple_i = \proj_{\Y_i}(\tuple'_i)$ and $\tuple'_i$ satisfies the string-equality constraints
specified by $E_i$.
Moreover, we know that $\Dom{\tuple'_i}\subseteq \Z_i$. Hence, since $\Z_1\cap\Z_2\subseteq\tilde{\Y}=\Y_1\cap\Y_2$, we obtain that $\tuple'_1$ and $\tuple'_2$ are compatible (because $\tuple_1\joinable\tuple_2$ and $\tuple'_i=\proj_{\Y_i}(\tuple'_i)$).
Let $\tuple'\deff\tuple'_1\join\tuple'_2$ be their join result. 
Clearly, $\tuple'\in (S'_1\join S'_2)(w)$. Furthermore, $\tuple'$ satisfies all string-equality constraints specified by $E_1\cup E_2$, and therefore, 
$\tuple'\in (\streq_{E_1\cup E_2} (S'_1\join S'_2))(w)$.
Finally, observe that $\proj_{\Y_1\cup\Y_2}(\tuple')=\tuple$, and hence
$\tuple\in (\proj_{\Y_1\cup\Y_2}\streq_{E_1\cup E_2} (S'_1\join S'_2))(w) = \tilde{S}(w)$.
This finishes the proof of the $\subseteq$-part.
\medskip

\noindent
Now, we consider the $\supseteq$-part and choose an arbitrary $\tuple\in\tilde{S}(w)$.
Our aim is to show that $\tuple \in S(w)$.
By definition of $\tilde{S}$ there exists $\tuple'\in (S'_1\ast S'_2)(w)$ such that 
$\tuple=\proj_{\Y_1\cup\Y_2}(\tuple')$ and $\tuple'$ satisfies the string-equality constraints specified by $E_1\cup E_2$.

\emph{Case 1:} $\ast=\union$.
Since $\tuple'\in(S'_1\union S'_2)(w)$, there is an $i\in\set{1,2}$ such that $\tuple'\in S'_i(w)$. 
In particular, $\Dom{\tuple'}\subseteq \Z_i$, and $\tuple'$ satisfies the string-equality constraints specified by $E_i$.
Hence, $\tuple'\in (\streq_{E_i} (S'_i))(w)$.
Furthermore, 
$\tuple=\proj_{\Y_1\cup \Y_2}(\tuple')=\proj_{\Y_i}(\tuple')$, and hence
$\tuple \in (\proj_{\Y_i}\streq_{E_i} (S'_i)(w) = S_i(w)\subseteq S_1(w)\cup S_2(w) = S(w)$.

\emph{Case 2:} $\ast=\,\join$.
Since $\tuple'\in(S'_1\join S'_2)(w)$, for each $i\in\set{1,2}$ there exists $\tuple'_i\in S'_i(w)$ such that $\tuple'_1\joinable\tuple'_2$ and $\tuple'=\tuple'_1\join \tuple'_2$ (i.e., $\tuple'_1$ and $\tuple'_2$ are compatible and $\tuple'$ is their join result). In particular, $\Dom{\tuple'_i}\subseteq \Z_i$, and $\tuple'_i$ satisfies the string-equality constraints specified by $E_i$. Thus, $\tuple_i\deff \proj_{\Y_i}(\tuple'_i) \in (\proj_{\Y_i}\streq_{E_i}(S'_i))(w) = S_i(w)$. Furthermore, $\tuple_1\joinable \tuple_2$ and $\tuple=\tuple_1\join\tuple_2$. I.e., $\tuple\in S_1(w)\join S_2(w) = S(w)$.

This finishes the proof of the $\supseteq$-part, the proof of Claim~\ref{claim:ClaimInCoreSimplificationProof}, and the proof of Lemma~\ref{lemma:CoreSimplificationLemma}.
\end{proof}

\section{A Declarative Approach to Spanners}\label{section:declarative}

In this section, we develop a simple declarative approach to spanners by establishing a natural one-to-one correspondence between spanners over $\Sigma$ and so-called \emph{subword-marked languages} over $\Sigma$. This approach conveniently allows to investigate or define non-algorithmic properties of spanners completely independently from any machine model or other description mechanisms (e.\,g., types of regular expressions, automata, etc.), while at the same time we can use the existing algorithmic toolbox for formal languages whenever required (instead of inventing special-purpose variants of automata or regular expressions to this end).\footnote{In the literature on spanners, subword-marked words have previously been used as a tool to define the semantics of regex-formulas or vset-automata (see, e.\,g.,~\cite{DoleschalEtAl2019, Freydenberger2019, FreydenbergerEtAl2018, FreydenbergerThompson2020}). However, in these papers, the term \emph{ref-word} is used instead of subword-marked word, which is a bit of a misnomer due to the following reasons. Ref-words have originally been used in~\cite{Schmid2016} (in a different context) as words that contain \emph{references} to some of their subwords, which are explicitly marked. In the context of spanners, only ref-words with marked subwords, but \emph{without} any references have been used so far. Since in this work we wish to use ref-words in the sense of~\cite{Schmid2016}, i.\,e., with actual references, but also the variants without references, we introduce the term subword-marked word for the latter.} \par
In particular, this declarative approach is rather versatile and provides some modularity in the sense that we could replace ``regular languages'' by any kind of language class (e.\,g., (subclasses of) context-free languages, context-sensitive languages, etc.) to directly obtain (i.\,e., without any need to adopt our definitions) a formally sound class of document spanners and also have the full technical machinery that exists for this language class at our disposal. Note that the idea of using non-regular languages from the Chomsky hierarchy to define more powerful classes of document spanners has been recently used in~\cite{Peterfreund2021}. \par
In the context of this paper, however, the main benefit is that this approach provides a suitable angle to treat string-equality selections in a regular way. 

\subsection{Subword-Marked Words}

For any set $\varset$ of variables, we shall use the set $\Gamma_{\varset} = \{\open{\varsx}, \close{\varsx} \mid \varsx \in \varset\}$ as an alphabet of meta-symbols. In particular, for every $\varsx \in \varset$, we interpret the pair of symbols $\open{\varsx}$ and $\close{\varsx}$ as a pair of opening and closing parentheses. 

\begin{defi}[Subword-Marked Words]\label{subwordMarkedWordsDefinition}
A \emph{subword-marked} word (\emph{over terminal alphabet $\Sigma$ and variables $\varset$}) is a word $w \in (\Sigma \cup \Gamma_{\varset})^*$ such that, for every $\varsx \in \varset$, $\erasemorphism{\Sigma \cup \Gamma_{\varset\setminus\{\varsx\}}}(w) \in \{\eword, \open{\varsx} \close{\varsx}\}$. A subword-marked word is \emph{functional} if $|w|_{\open{\varsx}} = 1$ for every $\varsx \in \varset$. For a subword-marked word $w$ over $\Sigma$ and $\varset$, we set $\getWord{w} = \erasemorphism{\Gamma_{\varset}}(w)$.
\end{defi}

A subword-marked word $w$ can be interpreted as a word over $\Sigma$, i.\,e., the word $\getWord{w}$, in which some subwords are marked by means of the parentheses $\open{\varsx}$ and $\close{\varsx}$. In this way, it represents an $(\varset, \getWord{w})$-tuple, i.\,e., every $\varsx \in \varset$ is mapped to $\spann{i}{j} \in \spans(\getWord{w})$, where $w = w_1 \open{\varsx} w_2 \close{\varsx} w_3$ with $i = |\getWord{w_1}| + 1$ and $j = |\getWord{w_1 w_2}| + 1$. In the following, the $(\varset, \getWord{w})$-tuple defined by a subword-marked word $w$ is denoted by $\getSpanTuple{w}$. 
We note that $\getSpanTuple{w}$ is a total function if and only if $w$ is functional. Moreover, we say that a subword-marked word $w$ is \emph{hierarchical} or \emph{non-overlapping}, if $\getSpanTuple{w}$ is hierarchical or non-overlapping, respectively.

\begin{exa}
Let $\varset = \{\varsx, \varsy, \varsz\}$ and $\Sigma = \{\ta, \tb,
\tc\}$. Then $\open{\varsx} \ta \ta \close{\varsx} \ta \tb
\open{\varsy} \open{\varsz} \tc \ta \close{\varsz} \ta \close{\varsy}$
is a functional and hierarchical subword-marked word. The
subword-marked word $u = \tb \open{\varsx} \ta \open{\varsy} \ta \tb
\ta \open{\varsz} \ta \close{\varsz} \tc \close{\varsx} \ta \tb
\close{\varsy} \tc$ is functional, but not hierarchical, while $v =
\open{\varsx} \ta \open{\varsy} \tb \ta \close{\varsy} \tc \ta \tb
\close{\varsx} \tc \ta \ta$ is a non-functional, but hierarchical
subword-marked word. Moreover, $\getSpanTuple{u} = (\spann{2}{8},
\spann{3}{10}, \spann{6}{7})$ and $\getSpanTuple{v} = (\spann{1}{7},
\spann{2}{4}, \undefin)$. On the other hand, neither $\open{\varsx}
\ta \ta \close{\varsx} \ta \tb \open{\varsy} \open{\varsx} \tc \ta
\close{\varsx} \ta \close{\varsy}$ nor $\open{\varsx} \ta
\open{\varsy} \tb \ta \close{\varsx} \tc$ are valid subword-marked
words. 
\end{exa}

\subsection{Subword-Marked Languages and Spanners}

A set $L$ of subword-marked words (over $\Sigma$ and $\varset$) is a \emph{subword-marked language} (\emph{over $\Sigma$ and $\varset$}). A subword-marked language $L$ is called \emph{functional}, \emph{hierarchical} or \emph{non-overlapping} if all $w \in L$ are functional, hierarchical or non-overlapping, respectively. Since every subword-marked word $w$ over $\Sigma$ and $\varset$ describes a $(\varset, \getWord{w})$-tuple, subword-marked languages can be interpreted as spanners as follows.

\begin{defi}\label{subwordMarkedLanguageSpannersDefinition}
Let $L$ be a subword-marked language (over $\Sigma$ and $\varset$). Then the spanner $\llbracket L \rrbracket$ (over $\Sigma$ and $\varset$) is defined as follows: for every $w \in \Sigma^*$, $\llbracket L \rrbracket(w) = \{\getSpanTuple{v} \mid v \in L, \getWord{v} = w\}$. For a class $\mathcal{L}$ of subword-marked languages, we set $\llbracket \mathcal{L} \rrbracket = \{\llbracket L \rrbracket \mid L \in \mathcal{L}\}$.
\end{defi}

\begin{exa}
Let $\Sigma = \{\ta, \tb\}$ and $\varset = \{\varsx_1, \varsx_2,
\varsx_3\}$. Let $\alpha = \open{\varsx_1}(\ta \altop
\tb)^*\close{\varsx_1} \open{\varsx_2} \tb \close{\varsx_2}
\open{\varsx_3} (\ta \altop \tb)^* \close{\varsx_3}$ be a regular
expression over the alphabet $\Sigma \cup \Gamma_{\varset}$. We can
note that $\lang(\alpha)$ is a subword-marked language (over terminal
alphabet $\Sigma$ and variables $\varset$) and therefore $\llbracket
\lang(\alpha) \rrbracket$ is a spanner over $\varset$. In fact, $\llbracket \lang(\alpha) \rrbracket$ is exactly the spanner described by the function $S$ in Example~\ref{firstSpannerExample}. 
\end{exa}

In this way, every subword-marked language $L$ over $\Sigma$ and $\varset$ describes a spanner $\llbracket L \rrbracket$ over $\Sigma$ and $\varset$, and since it is also easy to transform any $(\varset, w)$-tuple $t$ into a subword-marked word $v$ with $\getWord{v} = w$ and $\getSpanTuple{v} = t$, also every spanner $S$ over $\Sigma$ and $\varset$ can be represented by a subword-marked language over $\Sigma$ and $\varset$. Moreover, for a subword-marked language $L$ over $\Sigma$ and $\varset$, $\llbracket L \rrbracket$ is a functional, hierarchical or non-overlapping spanner if and only if $L$ is functional, hierarchical or non-overlapping, respectively. This justifies that we can use the concepts of spanners (over $\Sigma$ and $\varset$) and the concept of subword-marked languages (over $\Sigma$ and $\varset$) completely interchangeably. By considering only \emph{regular} subword-marked languages, we automatically obtain the class of regular spanners (usually defined as the class of spanners that can be described by vset-automata~\cite{FaginEtAl2015,MaturanaEtAl2018}). More formally, let $\regswmLanguages_{\Sigma, \varset}$ be the class of regular subword-marked languages over $\Sigma$ and $\varset$ and let $\regswmLanguages = \bigcup_{\Sigma, \varset} \regswmLanguages_{\Sigma, \varset}$.

\begin{prop}\label{regSpannerProposition}
$\regspanners = \llbracket \regswmLanguages \rrbracket$.
\end{prop}

\begin{proof}
Let $M$ be a vset-automaton over $\Sigma$ and with variables
$\varset$. Then we can intrepret $M$ as an $\NFA$ $M'$ over alphabet
$\Sigma \cup \Gamma_{\varset}$, by interpreting every variable
operation $\varsx \vdash$ and $\dashv \varsx$ as an $\open{\varsx}$- and
$\close{\varsx}$-transition, respectively. We can then further modify $M'$ such that it rejects all inputs that are not subword-marked words (this can be done by simply keeping track in the finite state control which markers have been read so far). With this further modification, we have that $\llbracket M \rrbracket = \llbracket \lang(M') \rrbracket$.\par
On the other hand, let $M$ be an $\NFA$ such that $\lang(M)$ is a subword-marked language over $\Sigma$ and $\varset$. Then we can simply interpret all $\open{\varsx}$- and $\close{\varsx}$-transitions as variable operations $\varsx \vdash$ and $\dashv \varsx$ in order to obtain a vset-automaton $M'$ with $\llbracket \lang(M) \rrbracket = \llbracket M' \rrbracket$.
\end{proof}

Proposition~\ref{regSpannerProposition} justifies that, throughout this paper,
instead of using the vset-automata of \cite{FaginEtAl2015}, we will
represent regular spanners by ordinary nondeterministic finite
automata ($\NFA$) over $\Sigma \cup \Gamma_{\varset}$ that accept
subword-marked languages. It is an easy exercise to see that any
vset-automaton (as defined in \cite{FaginEtAl2015}) can be transformed into an according $\NFA$ by
increasing the automaton's state space by
at most the factor $3^{|\varset|}$. If the vset-automaton is \emph{sequential} (see, e.\,g.,~\cite{AmarilliEtAl2021}), then the $\NFA$ is even of linear size.

It is a straightforward but important observation that for any given $\NFA$ over $\Sigma \cup \Gamma_{\varset}$, we can efficiently check whether $\lang(M)$ is a subword-marked language. Since most description mechanisms for regular languages (e.\,g., expressions, grammars, logics, etc.) easily translate into $\NFA$, they can potentially all be used for defining regular spanners.

\begin{prop}\label{CheckSubwordMarkedLanguagesProposition}
Given an $\NFA$ $M$ over alphabet $\Sigma \cup \Gamma_{\varset}$, we can decide in time $\bigO(|M|\cdot|\varset|^2)$ if $\lang(M)$ is a subword-marked language, and, if so, whether $\lang(M)$ is functional in time $\bigO(|M|\cdot|\varset|^2)$, and whether it is hierarchical or non-overlapping in time $\bigO(|M|\cdot|\varset|^3)$.
\end{prop}

\begin{proof}
A word $w \in (\Sigma \cup \Gamma_{\varset})^*$ is \emph{not} a subword-marked word if and only if there is an $\varsx \in \varset$ such that one of the following properties is satisfied:
\begin{itemize}
\item $w = w_1 \close{\varsx} w_2$ with $|w_1|_{\open{\varsx}} = 0$,
\item $w = w_1 \open{\varsx} w_2$ with $|w_2|_{\close{\varsx}} = 0$,
\item $|w|_{\open{\varsx}} \geq 2$ or $|w|_{\close{\varsx}} \geq 2$.
\end{itemize}
Moreover, a subword-marked word over $\Sigma$ and $\varset$ is \emph{not} functional if and only if there is an $\varsx \in \varset$ such that $|w|_{\open{\varsx}} = 0$; it is \emph{not} hierarchical if there are $\varsx, \varsy \in \varset$ such that $w = w_1 \open{\varsx} w_2 \open{\varsy} w_3 \close{\varsx} w_4 \close{\varsy} w_5$ with $\getWord{w_2} \neq \eword$, $\getWord{w_3} \neq \eword$ and $\getWord{w_4} \neq \eword$; and it is \emph{not} non-overlapping if there are $\varsx, \varsy \in \varset$ such that one of the following properties is satisfied:
\begin{itemize}
\item $w = w_1 \open{\varsx} w_2 \open{\varsy} w_3 \close{\varsx} w_4 \close{\varsy} w_5$ and $\getWord{w_3} \neq \eword$ and $\getWord{w_2 w_4} \neq \eword$
\item $w = w_1 \open{\varsx} w_2 \open{\varsy} w_3 \close{\varsy} w_4 \close{\varsx} w_5$ and 
\begin{itemize}
\item $\getWord{w_3} = \eword$ and $\getWord{w_2} \neq \eword$ and $\getWord{w_4} \neq \eword$,
\item $\getWord{w_3} \neq \eword$ and $\getWord{w_2w_4} \neq \eword$.
\end{itemize}
\end{itemize}
These considerations show that we can construct an $\NFA$ $N$ that checks whether a given word over $\Sigma \cup \Gamma_{\varset}$ is \emph{not} a subword-marked word over $\Sigma$ and $\varset$, or whether a given subword-marked word over $\Sigma$ and $\varset$ is \emph{not} functional. Note that $N$ needs a constant number of states per $\varsx \in \varset$, and the outdegree is bounded by $\bigO(|\Sigma \cup \Gamma_{\varset}|) = \bigO(|\varset|)$; thus, $|N| = \bigO(|\varset|^2)$.\par
Moreover, we can construct an $\NFA$ $N$ that checks whether a given subword-marked word over $\Sigma$ and $\varset$ is \emph{not} hierarchical or \emph{not} non-overlapping. Since we now need a constant number of states for each two variables $\varsx, \varsy \in \varset$, we can assume that $|N| = \bigO(|\varset|^3)$. \par
Finally, we can check whether $\lang(M)$ is a subword-marked language by checking non-emptiness for the cross-product automaton of $M$ and $N$ (i.\,e., the automaton that accepts the intersection $\lang(M) \cap \lang(N)$). Since the cross-product automaton has size $\bigO(|M|\cdot|N|) = \bigO(|M|\cdot|\varset|^2)$, this can be done in time $\bigO(|M|\cdot|\varset|^2)$. Analogously, checking functionality for a subword-marked language given by an $\NFA$ $M$ can be done in time $\bigO(|M|\cdot|\varset|^2)$ by exploiting the respective automata $N$ for checking non-functionality, and checking hierarchicality or the non-overlapping property can be done in time $\bigO(|M|\cdot|\varset|^3)$ by exploiting the respective automata $N$ for checking non-hierarchicality or the overlapping property, respectively.
\end{proof}

\section{Refl-Spanners: Spanners with Built-In String-Equality Selections}\label{section:refl}

In this section, we extend the concept of subword-marked words and
languages in order to describe spanners with string-equality
selections.

\subsection{Ref-Words and Ref-Languages}

We consider subword-marked words with extended terminal alphabet $\Sigma \cup \varset$, i.\,e., in addition to symbols from $\Sigma$, also variables from $\varset$ can appear as terminal symbols (the marking of subwords with symbols $\Gamma_{\varset}$ remains unchanged).

\begin{defi}[Ref-Words]\label{refWordDefinition}
A \emph{ref-word} \emph{over $\Sigma$ and $\varset$} is a subword-marked word over terminal alphabet $\Sigma \cup \varset$ and variables $\varset$, such that, for every $\varsx \in \varset$, if $w = w_1 \varsx w_2$, then there exist words $v_1, v_2, v_3$ such that $w_1 = v_1 \open{\varsx} v_2 \close{\varsx} v_3$.
\end{defi}

Since ref-words are subword-marked words, the properties \emph{functional}, \emph{hierarchical} and \emph{non-overlapping} are well-defined.

\begin{exa}
Let $\Sigma = \{\ta, \tb, \tc\}$ and $\varset = \{\varsx, \varsy\}$. The subword-marked words $u = \ta \tb \open{\varsx} \ta \tb \close{\varsx} \tc \open{\varsy} \varsx \ta \ta \close{\varsy} \varsy$ and $v = \ta \open{\varsx} \ta \tb \open{\varsy} \ta \tb \close{\varsx} \ta \close{\varsy} \varsx \varsy$ (over terminal alphabet $\Sigma \cup \varset$ and variables $\varset$) are valid ref-words (over terminal alphabet $\Sigma$ and variables $\varset$). Note that both $u$ and $v$ are functional, and $u$ is also hierarchical, while $v$ is not. On the other hand, $\ta \varsx \tb \open{\varsx} \ta \tb \close{\varsx} \tc \open{\varsy} \varsx \ta \ta \close{\varsy} \varsy$ and $\ta \ta \open{\varsx} \ta \tb \close{\varsx} \tc \open{\varsy} \varsy \ta \close{\varsy}$ are subword-marked words (over terminal alphabet $\Sigma \cup \varset$ and variables $\varset$), but not ref-words.
\end{exa}

The idea of ref-words is that occurrences of $\varsx \in \varset$ are interpreted as \emph{references} to the subword $\open{\varsx} v \close{\varsx}$, which we will call the \emph{definition} of $\varsx$.  Note that while a single variable can have several references, there is at most one definition per variable, and if there is no definition for a variable, then it also has no references. Variable definitions may contain other references or definitions of other variables, i.\,e., there may be chains of references, e.\,g., the definition of $\varsx$ contains references of $\varsy$, and the definition of $\varsy$ contains references of $\varsz$ and so on. Next, we formally define this nested referencing process encoded by ref-words. Recall that for a subword-marked word $w$ we denote by $\getWord{w}$ the word obtained by removing all meta-symbols from $\Gamma_{\varset}$ from $w$ (however, if $w$ is a ref-word, then $\getWord{w}$ is a word over $\Sigma \cup \varset$).

\begin{defi}[Deref-Function]\label{derefDefinition}
For a ref-word $w$ over $\Sigma$ and $\varset$, the subword-marked word $\deref{w}$ over $\Sigma$ and $\varset$ is obtained from $w$ by repeating the following steps until we have a subword-marked word over $\Sigma$ and $\varset$:
\begin{enumerate} 
\item\label{refStepTwo} Let $\open{\varsx} v_{\varsx} \close{\varsx}$ be a definition such that $\getWord{v_{\varsx}} \in \Sigma^*$.  
\item\label{refStepThree} Replace all occurrences of $\varsx$ in $w$ by $\getWord{v_{\varsx}}$.
\end{enumerate} 
\end{defi}

It is straightforward to verify that the function $\deref{\cdot}$ is well-defined. By using this function and because ref-words encode subword-marked words, they can be interpreted as span-tuples. More precisely, for every ref-word $w$ over $\Sigma$ and $\varset$, $\deref{w}$ is a subword-marked word over $\Sigma$ and $\varset$, $\getWord{\deref{w}} \in \Sigma^*$ and $\getSpanTuple{\deref{w}}$ is an $(\varset, \getWord{\deref{w}})$-tuple. 

\begin{exa}
Let $\Sigma = \{\ta, \tb, \tc\}$, let $\varset = \{\varsx_1, \varsx_2, \varsx_3, \varsx_4\}$ and let
\[
w = \ta \ta \open{\varsx_1} \ta \tb \open{\varsx_2} \ta \tc \tc \close{\varsx_2} \ta \varsx_2 \close{\varsx_1} \open{\varsx_4} \varsx_1 \ta \varsx_2 \close{\varsx_4} \varsx_4 \tb \varsx_1\,.
\]
Due to the definition $\open{\varsx_2} \ta \tc \tc \close{\varsx_2}$, the procedure of Definition~\ref{derefDefinition} will initially replace all references of $\varsx_2$ by $\ta \tc \tc$. Then, the definition for variable $\varsx_1$ is $\open{\varsx_1} \ta \tb \open{\varsx_2} \ta \tc \tc \close{\varsx_2} \ta \ta \tc \tc \close{\varsx_1}$, so $\ta \tb \ta \tc \tc \ta \ta \tc \tc$ can be substituted for the references of $\varsx_1$. After replacing the last variable $\varsx_4$, we obtain 
\[
\deref{w} = \ta \ta \open{\varsx_1} \ta \tb \open{\varsx_2} \ta \tc \tc \close{\varsx_2} \ta \ta \tc \tc \close{\varsx_1} \open{\varsx_4} \ta \tb \ta \tc \tc \ta \ta \tc \tc \ta \ta \tc \tc \close{\varsx_4} \ta \tb \ta \tc \tc \ta \ta \tc \tc \ta \ta \tc \tc \tb \ta \tb \ta \tc \tc \ta \ta \tc \tc\,.
\]
Moreover, we have $\getSpanTuple{\deref{w}} = (\spann{3}{12}, \spann{5}{8}, \undefin, \spann{12}{25})$.\par
As a non-hierarchical example, consider $u = \open{\varsx_1} \ta \open{\varsx_2} \ta \ta \close{\varsx_1} \tc \varsx_1 \open{\varsx_3} \ta \tc \close{\varsx_2} \varsx_2 \ta \varsx_1 \close{\varsx_3}$. It can be easily verified that $\deref{u} = \open{\varsx_1} \ta \open{\varsx_2} \ta \ta \close{\varsx_1} \tc \ta \ta \ta \open{\varsx_3} \ta \tc \close{\varsx_2} \ta \ta \tc \ta \ta \ta \ta \tc \ta \ta \ta \ta \close{\varsx_3}$ and $\getSpanTuple{\deref{u}} = (\spann{1}{4}, \spann{2}{10}, \spann{8}{22}, \undefin)$.
\end{exa}

A set $L$ of ref-words is called a \emph{ref-language}. We extend the $\deref{\cdot}$-function to ref-languages $L$ in the obvious way, i.\,e., $\deref{L} = \{\deref{w} \mid w \in L\}$. Note that $\deref{L}$ is necessarily a subword-marked language. As for subword-marked languages, we are especially interested in ref-languages that are regular. By $\regrefLanguages_{\Sigma, \varset}$ we denote the class of regular ref-languages over $\Sigma$ and $\varset$, and we set $\regrefLanguages = \bigcup_{\Sigma, \varset} \regrefLanguages_{\Sigma, \varset}$. \par
In analogy to Proposition~\ref{CheckSubwordMarkedLanguagesProposition}, we can easily check for a given $\NFA$ over alphabet $\Sigma \cup \Gamma_{\varset} \cup \varset$ whether it accepts a ref-language over $\Sigma$ and $\varset$:

\begin{prop}\label{checkRefPropProposition}
Given an $\NFA$ $M$ where $\lang(M)$ is a subword-marked language over $\Sigma \cup \varset$ and $\varset$, we can decide in time $\bigO(|M|\cdot|\varset|^2)$ if $\lang(M)$ is a ref-language over $\Sigma$ and $\varset$.
\end{prop}

\begin{proof}
Since $M$ accepts a subword-marked language over $\Sigma \cup \varset$
and $\varset$, it does \emph{not} accept a ref-language over $\Sigma$
and $\varset$ if and only if it accepts a word $w \in (\Sigma\cup \varset \cup \Gamma_{\varset})^*$ such that, for some $\varsx \in \varset$, $w = w_1 \varsx w_2$ with $|w_1|_{\close{\varsx}} = 0$.
It is straightforward to construct an $\NFA$ $N$ that accepts all words over $\Sigma \cup \Gamma_{\varset} \cup \varset$ with this property, and also $|N| = \bigO(|\varset|^2)$ (per each variable, we need a constant number of states with an outdegree bounded by $\bigO(|\Sigma \cup \Gamma_{\varset} \cup \varset|) = \bigO(|\varset|)$). Therefore, we can check whether $\lang(M)$ is a ref-language over $\Sigma$ and $\varset$ by constructing the cross-product automaton of $M$ and $N$ and check whether it accepts $\emptyset$. This can be done in time $\bigO(|M|\cdot|N|) = \bigO(|M|\cdot|\varset|^2)$.
\end{proof}

\subsection{Refl-Spanners} 

We shall now define spanners based on regular ref-languages.

\begin{defi}[Refl-Spanners]
Let $L$ be a ref-language (over $\Sigma$ and $\varset$). Then the \emph{refl-spanner} $\llbracket L \rrbracket_{\derefmark}$ (over $\Sigma$ and $\varset$) is defined by $\llbracket L \rrbracket_{\derefmark} = \llbracket \deref{L} \rrbracket$. For a class $\mathfrak{L}$ of ref-languages, we set $\llbracket \mathfrak{L} \rrbracket_{\derefmark} = \{\llbracket L \rrbracket_{\derefmark} \mid L \in \mathfrak{L}\}$. The class of \emph{refl-spanners} is $\reflspanners = \llbracket \regrefLanguages \rrbracket_{\derefmark}$.
\end{defi}   

Since any regular ref-language $L$ over $\Sigma$ and $\varset$ is also a regular subword-marked language over $\Sigma \cup \varset$ and $\varset$, $\llbracket L \rrbracket$ is, according to Definition~\ref{subwordMarkedLanguageSpannersDefinition}, also a well-defined spanner (but over $\Sigma \cup \varset$ and $\varset$). However, whenever we are concerned with a ref-language $L$ over $\Sigma$ and $\varset$ that is not also a subword-marked language over $\Sigma$ and $\varset$ (i.\,e., $L$ contains actual occurrences of symbols from $\varset$), then we are never interested in $\llbracket L \rrbracket$, but always in $\llbracket L \rrbracket_{\derefmark}$. Consequently, by a slight abuse of notation, we denote in this case $\llbracket L \rrbracket_{\derefmark}$ simply by $\llbracket L \rrbracket$. \par
For a regular ref-language $L$ the corresponding refl-spanner $\llbracket L \rrbracket$ produces for a given $w \in \Sigma^*$ all $(\varset, w)$-tuples $t$ that are represented by some $u \in \deref{L}$ with $\getWord{u} = w$, or, equivalently, all $(\varset, w)$-tuples $t$ with $t = \getSpanTuple{\deref{v}}$ and $\getWord{\deref{v}} = w$ for some $v \in L$. It is intuitively clear that the use of variable references of refl-spanners provides a functionality that resembles string-equality selections for core spanners. However, there are also obvious differences between these two spanner formalisms (as already mentioned in the introduction and as investigated in full detail in Section~\ref{sec:ExpressivePower}).\par
Before moving on to the actual results about refl-spanners, we shall briefly discuss another example. 

\begin{exa}\label{JACMExample}
Assume that we have a document $w = \# p_1 \# p_2 \# \ldots \# p_n \#$, where each $p_i \in \Sigma^*$ is the title page of a scientific paper and $\# \notin \Sigma$ is some separator symbol (e.\,g., a list of all title pages of papers in the issues of \emph{Journal of the ACM} (JACM) from 2000 to 2010). Let $\Sigma' = \Sigma \cup \{\#\}$. We define a refl-spanner 
\[
\alpha \quad = \quad\textsf{$\Sigma'^* \: \# \: \Sigma^*$ email: $\open{\varsx} \Sigma^+ \close{\varsx}$ @ $r_{\textsf{dom}} \: \Sigma^* \: \# \: \Sigma'^*$ email: $\varsx$ @ $r_{\textsf{dom}} \: \Sigma^* \: \# \: \Sigma'^*$}\,,
\]
where $r_{\textsf{dom}} = \textsf{hu-berlin.de} \altop \textsf{tu-berlin.de} \altop \textsf{fu-berlin.de}$ is a regular expressions that matches the email-domains of the three universities in Berlin. It can be easily seen that $\llbracket \alpha \rrbracket (w)$ contains the first parts of the email-addresses of authors that have at least two JACM-papers between year 2000 and 2010 while working at a university in Berlin.
\end{exa}

\section{Evaluation and Static Analysis of Refl-Spanners}\label{sec:decisionProblems}

The problem $\TestProb$ is to decide whether $t \in S(w)$ for a given
spanner $S$ over $\Sigma$ and $\varset$, $w \in \Sigma^*$ and
$(\varset, w)$-tuple $t$. The problem $\NonemptProb$ is to decide whether $S(w) \neq \emptyset$ for given $S$ and $w$. 
For the problem $\SatProb$, we get a single spanner $S$ as input and ask whether there is a $w \in \Sigma^*$ with $S(w) \neq \emptyset$, and for the problems $\HierProb$ and $\FunctProb$, we get a single spanner $S$ and ask whether $S$ is hierarchical, or whether $S$ is functional, respectively. 
Finally, $\ContProb$ and $\EquProb$ is to decide whether $S_1 \subseteq S_2$ or $S_1 = S_2$, respectively, for given spanners $S_1$ and $S_2$. The input refl-spanners are always given as $\NFA$ that accept a ref-language. Recall that a summary of our results is provided by Table~\ref{comparisonTable} in the introduction.\par
We first define some concepts that shall be helpful for dealing with the issue that different subword-marked words $w$ and $w'$ with $\getWord{w} = \getWord{w'}$ can nevertheless describe the same span-tuple, i.\,e., $\getSpanTuple{w} = \getSpanTuple{w'}$, since the order of consecutive occurrences of symbols from $\Gamma_{\varset}$ has no impact on the represented span-tuple. \par
Let $w$ be a subword-marked word over $\Sigma$ and $\varset$, let $t = \getSpanTuple{w}$, and let $\getWord{w} = w_1 w_2 \ldots w_n$ with $w_i \in \Sigma$ for every $i \in [n]$. The \emph{marker-set representation} of $w$ is a word $\msrep{w}$ over the extended alphabet $\Sigma\cup\Pot{\Gamma_\X}$, where each subset of $\Gamma_{\X}$ serves as a letter of the alphabet, as follows. For each $i\in \set{1,\ldots,n{+}1}$ let $\Gamma_i$ be the set comprising of all symbols $\open{\x}$ where $\x\in\Dom{t}$ and $t(\x)=\spann{i}{j}$ for some $j$, and all symbols $\close{\x}$ where $\x\in \Dom{t}$ and $t(\x)=\spann{j}{i}$ for some $j$.  Then $\msrep{w} = (\Gamma_1)^{|\Gamma_1|} w_1 (\Gamma_2)^{|\Gamma_2|} w_2 \ldots (\Gamma_n)^{|\Gamma_n|} w_n (\Gamma_{n+1})^{|\Gamma_{n+1}|}$. Let us illustrate these definitions with an example.\footnote{Also note that the marker-set representation is similar to the extended vset-automata used in~\cite{AmarilliEtAl2021}.}

\begin{exa}\label{msrepExample}
If $\X = \{\varsx_1, \varsx_2, \varsx_3\}$, $w = \open{\varsx_1} \open{\varsx_3} \ta \tb \close{\varsx_1} \open{\varsx_2} \tc \tb \close{\varsx_2} \ta \tb \tc \ta \close{\varsx_3}$ (and therefore $\getSpanTuple{w} = (\spann{1}{3}, \spann{3}{5}, \spann{1}{9})$), then $\Gamma_1 = \{\open{\varsx_1}, \open{\varsx_3}\}$, $\Gamma_3 = \{\close{\varsx_1}, \open{\varsx_2}\}$, $\Gamma_5 = \{\close{\varsx_2}\}$, $\Gamma_9 = \{\close{\varsx_3}\}$, and $\Gamma_i = \emptyset$ for every $i \in \{2,4,6,7,8\}$. Consequently, 
\begin{equation*}
\msrep{w} = \{\open{\varsx_1}, \open{\varsx_3}\} \{\open{\varsx_1}, \open{\varsx_3}\} \ta \tb \{\close{\varsx_1}, \open{\varsx_2}\} \{\close{\varsx_1}, \open{\varsx_2}\} \tc \tb \{\close{\varsx_2}\} \ta \tb \tc \ta \{\close{\varsx_3}\}\,.
\end{equation*}
\end{exa}

We also extend the function $\msrep{\cdot}$ in the natural way to a subword-marked language $L$ by setting $\msrep{L} = \{\msrep{w} \mid w \in L\}$. \par
As indicated above, the main idea of the marker-set representation is to represent a subword-marked word in a way that abstracts from the differences caused by reading the same markers from $\Gamma_{\varset}$ just in a different order, which, for subword-marked words considered as strings makes a difference, but has no effect on the represented document or the represented span-tuple. Such a representation could also be defined without repeating each $\Gamma_i$ for $|\Gamma_i|$ times, i.\,e., by defining $\msrep{w} = \Gamma_1 w_1 \Gamma_2 w_2 \ldots \Gamma_n w_n \Gamma_{n+1}$ (or, illustrated with Example~\ref{msrepExample} from above, as $\msrep{w} = \{\open{\varsx_1}, \open{\varsx_3}\} \ta \tb \{\close{\varsx_1}, \open{\varsx_2}\} \tc \tb \{\close{\varsx_2}\} \ta \tb \tc \ta \{\close{\varsx_3}\}$). However, the definition with repeated occurrences of $\Gamma_i$ will be more convenient for proving the results of this section. Later on, in Section~\ref{nonOverlappingCoreSpanner}, we shall also use this simplified version of the marker-set representation.\par
The next propositions point out the relevant properties of the marker-set representation. The first one is a direct consequence of the definition of the marker-set representation, the second one follows from the first one. We state these propositions mainly for demonstrating the meaning of the marker-set representation and for general illustrational purposes; for our results, the ref-word analogues of Proposition~\ref{setRepSimpleProp}~and~\ref{setRepInclusionProp} (formally stated as Lemmas~\ref{setRepSimpleLemmaRefVersion}~and~\ref{setRepInclusionLemmaRefVersion}) are more important and shall be proven in detail.

\begin{prop}\label{setRepSimpleProp}
Let $w_1, w_2$ be subword-marked words. Then $\getWord{w_1} = \getWord{w_2}$ and $\getSpanTuple{w_1} = \getSpanTuple{w_2}$ if and only if $\msrep{w_1} = \msrep{w_2}$.
\end{prop}

\begin{prop}\label{setRepInclusionProp}
Let $L_1, L_2$ be subword-marked languages. Then $\llbracket L_1 \rrbracket \subseteq \llbracket L_2 \rrbracket$ if and only if $\msrep{L_1} \subseteq \msrep{L_2}$.
\end{prop}

Note that the above proposition states that $\msrep{L_1} \subseteq \msrep{L_2}$ is characteristic for the inclusion of the corresponding spanners (i.\,e., $\llbracket L_1 \rrbracket \subseteq \llbracket L_2 \rrbracket$), but not for the inclusion of the corresponding subword-marked languages $L_1$ and $L_2$. More precisely, $\msrep{L_1} \subseteq \msrep{L_2}$ is not a sufficient condition for $L_1  \subseteq  L_2$ (although, as can be easily seen, it is a necessary condition).

\subsection{Model Checking and Non-Emptiness}

We now consider the problem $\TestProb$ (recall that this is the problem to decide whether $t \in S(w)$ for given spanner $S$ over $\Sigma$ and $\varset$, $w \in \Sigma^*$ and $(\varset, w)$-tuple $t$).

\begin{thm}\label{checkingTheorem}
$\TestProb$ for $\reflspanners$ can be solved in time $\bigO((|w| + |\X|) |M| \log(|\X|))$, where $M$ is an $\NFA$ that represents a refl-spanner $S = \sem{\lang(M)}$ over $\Sigma$ and $\X$, $w\in\Sigma^*$, and $t$ is an $(\X,w)$-tuple. 
\end{thm}

\begin{proof}
For convenience, we will assume that every symbol in $\Gamma_\X$ occurs in some transition of $M$ (this is w.l.o.g., because if $\open{\x}$ or $\close{\x}$ does not, we can safely remove $\x$ from $\X$ and reject $t$ in case that $\x\in\Dom{t}$). In particular, $|\X|\leq |M|$. \par
The proof is split in two parts. First, we consider the \emph{regular case}, where $\lang(M)$ is not a ref-language, but just a subword-marked language over $\Sigma$ and $\varset$ (i.\,e., $\llbracket \lang(M) \rrbracket$ is a regular spanner). The algorithm for this case is then used for the general \emph{refl case} where $\lang(M)$ is a ref-language over $\Sigma$ and $\varset$ (i.\,e., $\llbracket \lang(M) \rrbracket$ is a refl-spanner).\par
\textbf{Regular case:} We shall denote by $\msrep{M}$ the $\NFA$ $M$ with the modification that we interpret each transition labelled with some $\sigma \in \Gamma_\X$ as a transition that can read any symbol $\Gamma \subseteq \Gamma_\X$ with $\sigma \in \Gamma$. This means that $\msrep{M}$ accepts a language over the alphabet $\Sigma \cup \Pot{\Gamma_{\X}}$ and, moreover, for every $w \in \lang(M)$ and arbitrary subword-marked word $w'$ with $\msrep{w} = \msrep{w'}$ (this includes $w$ itself), we have that $w' \in \lang(\msrep{M})$.

By definition, $t \in \llbracket \lang(M) \rrbracket(w)$ if and only if there exists a subword-marked word $v \in \lang(M)$ such that $\getSpanTuple{v} = t$ and $\getWord{v}=w$. This latter property can be checked as follows.\par

We first combine $t$ and $w$ into a subword-marked word $w'$, i.\,e., $\getSpanTuple{w'} = t$ and $\getWord{w'}=w$ (the order in which we place consecutive occurrences of symbols from $\Gamma_{\X}$ is not relevant). In the following, we shall consider $\msrep{w'}$ (the construction of which we discuss later on), and the automaton $\msrep{M}$ as explained above. 

\begin{clm}\label{checkingTheorem:firstClaim}
The following statements are equivalent:\\
(1) There exists a subword-marked word $v \in \lang(M)$ such that $\getSpanTuple{v} = t$ and $\getWord{v}=w$.\\
(2) $\msrep{w'} \in \lang(\msrep{M})$.
\end{clm}

\begin{proof}
We start with ``(1) $\implies$ (2)'' and assume that there exists a subword-marked word $v \in \lang(M)$ such that $\getSpanTuple{v} = t$ and $\getWord{v}=w$. The fact $\getWord{v}=w$ means that $v$ equals $w$ with some (possibly empty) factors over $\Gamma_{\varset}$ between the symbols of $\Sigma$. Moreover, $\getSpanTuple{v} = t$ means that if in $v$ we replace each maximal factor $\sigma_1 \sigma_2 \ldots \sigma_k$ over $\Gamma_{\X}$ with $\{\sigma_1, \sigma_2, \ldots, \sigma_k\}^k$, we get exactly $\msrep{w'}$, and, by construction of $\msrep{M}$, the fact that $v \in \lang(M)$ also implies that $\msrep{w'} \in \lang(\msrep{M})$ (i.\,e., $\msrep{M}$ accepts $\msrep{w'}$ via the same sequence of states).\par
Next, we prove ``(2) $\implies$ (1)'' and assume that $\msrep{w'} \in
\lang(\msrep{M})$, which means that there is an accepting run of $\msrep{M}$ on $\msrep{w'}$. By construction, each transition of this run that reads a symbol $\Gamma \subseteq \Gamma_{\varset}$ is a re-interpretation of an original $\sigma$-transition of $M$ for some $\sigma \in \Gamma$. Consequently, for every symbol $\Gamma \subseteq \Gamma_{\varset}$ of $\msrep{w'}$, there is a well-defined \emph{representative} $\sigma \in \Gamma$. If we replace every symbol $\Gamma \subseteq
\Gamma_{\X}$ in $\msrep{w'}$ by its representative, then we obtain a
word $v$ over $\Sigma \cup \Gamma_{\X}$, and by construction of
$\msrep{M}$, we also have $v \in \lang(M)$, which means that $v$ is a
subword-marked word. By construction of $\msrep{w'}$, it is also
obvious that $\getWord{v} = \getWord{\msrep{w'}} = w$. Moreover, in
the construction of $v$, every factor in $\msrep{w'}$ of the form
$\Gamma^{|\Gamma|}$ with $\Gamma \subseteq \Gamma_\X$, is replaced by
a factor $\sigma_1 \sigma_2 \ldots \sigma_{|\Gamma|}$ with $\sigma_i
\in \Gamma$ for every $i \in [|\Gamma|]$. Since $v$ is a
subword-marked word, this means that $\{\sigma_1, \sigma_2, \ldots,
\sigma_{|\Gamma|}\} = \Gamma$, which means that $(\sigma_1, \sigma_2,
\ldots, \sigma_{|\Gamma|})$ is some linear ordering of $\Gamma$. By
construction of $\msrep{w'}$, this directly implies that
$\getSpanTuple{v} = t$.  
\end{proof}

Claim~\ref{checkingTheorem:firstClaim} means that in order to check whether $t \in \llbracket \lang(M) \rrbracket(w)$, it is sufficient to check whether $\msrep{w'} \in \lang(\msrep{M})$. Therefore, we shall now discuss how we can efficiently check $\msrep{w'} \in \lang(\msrep{M})$. First, we have to construct (a suitable representation of) $\msrep{w'}$, which we do in the following way. \par
Let us assume that $w' = \gamma_1 w_1 \gamma_2 w_2 \ldots \gamma_n w_n \gamma_{n+1}$, where $w = w_1 w_2 \ldots w_n$ and the $\gamma_1, \gamma_2, \ldots, \gamma_{n+1}$ are (possibly empty) factors over $\Gamma_{\X}$. Consequently, 
\begin{equation*}
\msrep{w'} = (\Gamma_1)^{|\Gamma_1|} w_1 (\Gamma_2)^{|\Gamma_2|} w_2 \ldots (\Gamma_n)^{|\Gamma_n|} w_n (\Gamma_{n+1})^{|\Gamma_{n+1}|}
\end{equation*}
according to the definition of $\msrep{\cdot}$. Let us further assume that we have some fixed order on $\Gamma_{\X}$.
We initialise an array $A$ of size $n + 1$ in time $O(n)$. Then we move over $w'$ from left to right, and whenever we encounter some $\gamma_i = \sigma_1 \sigma_2 \ldots \sigma_{n_i}$, we construct a binary search tree (with respect to the order on $\Gamma_{\X}$) of the elements of $\Gamma_i = \{\sigma_1, \sigma_2, \ldots, \sigma_{n_i}\}$, and we store this search tree in entry $i$ of $A$. For this, we have to consider each symbol from $w'$ only once and we have to insert at most $|\Gamma_{\X}|$ symbols from $\Gamma_{\X}$ into a binary search tree of size at most $|\Gamma_{\X}|$. Consequently, this can be done in time $\bigO(|w'| + (|\X| \log(|\X|)) = 
\bigO(|w| + (|\X| \log(|\X|))$. We then represent $\msrep{w'}$ as the
string $1^{|\Gamma_1|} w_1 2^{|\Gamma_2|} w_2 \ldots n^{|\Gamma_n|}
w_n (n+1)^{|\Gamma_{n+1}|}$ and the array $A$, where the symbols $i$
point to the $i^{\text{th}}$ entry of $A$.

With this representation of $\msrep{w'}$, we can read $\msrep{w'}$ with $\msrep{M}$, and every transition can be evaluated in constant time for symbols from $\Sigma$, and in time $\bigO(\log(|\X|))$ for symbols from $\Gamma_{\X}$. More precisely, transitions labelled with symbols from $\Sigma$ can be evaluated in constant time by directly comparing the next input symbol with the transition label from $\Sigma$, while for evaluating transition labels from $\Gamma_{\X}$, we have to check whether the input symbol $\Gamma \subseteq \Gamma_{\X}$ (from $\msrep{w'}$) contains the transition label, which, since we have $\Gamma$ represented as binary search tree, can be done in time $\bigO(\log(|\X|))$. Consequently, we can check whether $\msrep{M}$ accepts $\msrep{w'}$ in time $\bigO(|\msrep{w'}|\cdot|\msrep{M}| \cdot \log(|\X|)) = \bigO((|w| + |\X|) |M| \log(|\X|))$. Since this also dominates the time needed for constructing the representation of $\msrep{w'}$, the total running time is $\bigO((|w| + |\X|) |M| \log(|\X|))$.\par
\textbf{Refl case:} In this case, $t \in \llbracket \lang(M) \rrbracket(w)$ if and only if there is some $v \in \lang(M)$ such that $\getSpanTuple{\deref{v}} = t$ and $\getWord{\deref{v}}=w$. We next show how this latter property can be checked.\par
For every $\x\in \Dom{t}$ let $w_\x$ be the subword of $w$ that corresponds to $t(\x)$, i.e., $w_\x= w\spann{i}{j}$ if $\spann{i}{j}=t(\x)$. Next, we obtain an $\NFA$ $M'$ from $M$ by replacing every $\varsx$-transition for every $\varsx \in \X$ by a path of transitions labelled with the word $w_{\varsx}$. It is important to observe that $M'$ accepts a subword-marked language over $\Sigma$ and $\varset$ (indeed, replacing the references of ref-words by some factors over $\Sigma$ necessarily yields valid subword-marked words).

\begin{clm}\label{checkingTheorem:secondClaim}
The following statements are equivalent:\\
(1) There exists a ref-word $v \in \lang(M)$ such that $\getSpanTuple{\deref{v}} = t$ and $\getWord{\deref{v}}=w$.\\
(2) There exists a subword-marked word $v' \in \lang(M')$ such that $\getSpanTuple{v'} = t$ and $\getWord{v'}=w$. 
\end{clm}

\begin{proof}
We start with ``(1) $\implies$ (2)'' and assume that there exists some $v \in \lang(M)$ such that $\getSpanTuple{\deref{v}} = t$ and $\getWord{\deref{v}} = w$. Since $\getSpanTuple{\deref{v}} = t$, we know that $\deref{v}$ equals the word obtained from $v$ by replacing each occurrence of a reference $\varsx$ by the word $w_{\varsx}$. Hence, defining $v' := \deref{v}$, we have $v' \in \lang(M')$, and, by assumption, $\getSpanTuple{v'} = t$ and $\getWord{v'} = w$.\par
Next, we prove ``(2) $\implies$ (1)'' and assume that there exists some $v' \in \lang(M')$ such that $\getSpanTuple{v'} = t$ and $\getWord{v'}=w$. We obtain a word $v$ from the accepting run of $M'$ on $v'$ as follows. We simply trace all the symbols in the order read by the transitions, but whenever we read some $w_{\varsx}$ by some of the paths that $M$'s $\varsx$-transitions have been replaced with, we use the symbol $\varsx$ instead of the factor $w_{\varsx}$. 
By construction of $M'$, we have $v \in \lang(M)$, which means that $v$ is a ref-word and therefore $\deref{v}$ is well-defined. Furthermore, since $\getSpanTuple{v'} = t$ and by construction of the $w_{\varsx}$, we also have $\deref{v} = v'$. Finally, by assumption, we have $\getSpanTuple{v'} = t$ and $\getWord{v'}=w$, which implies $\getSpanTuple{\deref{v}} = t$ and $\getWord{\deref{v}}=w$.
\end{proof}

We observe that Claim~\ref{checkingTheorem:secondClaim} means that in
order to check $t \in \llbracket \lang(M) \rrbracket(w)$, it is
sufficient to check whether there exists a subword-marked word $v \in
\lang(M')$ (recall that $M'$ accepts a subword-marked language) such that $\getSpanTuple{v} = t$ and $\getWord{v}=w$. This latter task is exactly the \emph{regular case} discussed above; thus, it can be done in time $\bigO((|w| + |\X|) |M'| \log(|\X|))$.
However, since we obtained $M'$ from $M$ by replacing some transitions by paths of $\bigO(|w|)$ transitions, we can only bound $|M'| = \bigO(|w|\cdot|M|)$. This yields an overall upper bound of $\bigO((|w| + |\X|) |w|\cdot|M|\cdot \log(|\X|))$, which is quadratic in $|w|$. We shall now discuss how we can use a standard data-structure for strings in order to implement our algorithmic idea more efficiently.\par
Let us recall that the algorithm for the refl case first constructs $M'$ by replacing $\varsx$-transitions by $w_{\varsx}$-paths, then we re-interpret the $\Gamma_\X$-transitions of $M'$ as described in the \emph{regular case}, and then we check whether this $\NFA$ accepts $\msrep{w'}$ (for a suitable $w'$ obtained from $w$ and $t$). The whole improvement consists in just keeping the $\varsx$-transitions instead of replacing them by $w_{\varsx}$-paths, but then, when checking whether $\msrep{w'}$ is accepted, we nevertheless want to treat $\varsx$-transitions as $w_{\varsx}$-transitions, but in constant time (i.\,e., we want to check whether the remaining input starts with $w_{\varsx}$ in constant time), for which we need the data structure. More precisely, we will compute a data structure that allows us to check, for every position $i$ of $\msrep{w'}$ and every $\varsx \in \varset$, whether $w_{\varsx}$ is a prefix of $\msrep{w'}[i..|\msrep{w'}|]$ in constant time. 
With this addition to the algorithm, the overall running time is still $\bigO((|w| + |\X|) |M'| \log(|\X|))$
(note that for this it is vital that we can check \emph{in constant time} whether the remaining input starts with $w_{\varsx}$), but we have $|M'| = \bigO(|M|)$, which leads to the overall running time $\bigO((|w| + |\X|) |M| \log(|\X|))$. Let us now give the details for the improvement.\par
For a word $u$ and for every $i \in [|w|{+}1]$ let $\suff_i(u)$ be the suffix of $u$ starting at position $i$ of $u$. A \emph{longest common extension} data structure $\lce_{u}$ for a word $u$ is defined such that $\lce_{u}(i, j)$, for $i, j$ with $1 \leq i < j \leq |u|$, is the length of the longest common prefix of $\suff_i(u)$ and $\suff_j(u)$. It is known that (see, e.\,g.,~\cite{FischerHeun2006}) $\lce_{u}$ can be constructed in time $\bigO(|u|)$ and afterwards $\lce_{u}(i, j)$ can be retrieved in constant time for arbitrary given $i$ and $j$. \par
We construct $\lce_{w\hat{w}}$, where $\hat{w}$ is obtained from
$\msrep{w'}$ by replacing all symbols from $\Pot{\Gamma_{\X}}$ by a
new symbol $\#$, i.\,e., $\# \notin \Sigma \cup \Gamma_{\X}$ (note
that $|\msrep{w'}| = |\hat{w}|$ and these words only differ with
respect to symbols $\Pot{\Gamma_{\X}}$ and $\#$). Now, for a given
$\varsx \in {\varset}$ with $t(\varsx) =
\spann{i_{\varsx}}{j_{\varsx}}$, and position $i$ of $\msrep{w'}$, we
have that $w_{\varsx}$ is a prefix of $\msrep{w'}[i..|\msrep{w'}|]$ if
and only if $\lce_{w\hat{w}}(i_{\varsx}, |w| + i) \geq j_{\varsx} -
i_{\varsx}$. Finally, observe that $\lce_{w\hat{w}}$ can be
constructed in time $\bigO(|w\hat{w}|) = \bigO(|w| + |\varset|)$.
This completes the proof of Theorem~\ref{checkingTheorem}.
\end{proof}

We discuss a few particularities about Theorem~\ref{checkingTheorem}. The result points out that $\TestProb$ for refl-spanners has the same complexity as for regular spanners,\footnote{To the best knowledge of the authors, the bound that is provided by Theorem~\ref{checkingTheorem} is also the currently best upper bound for model checking of regular spanners in the literature.} which can be considered surprising, given the fact that refl-spanners cover a large class of core spanners and are generally much more expressive than regular spanners. Moreover, for core spanners the problem $\TestProb$ is $\npclass$-complete. Another interesting fact is that in data complexity, $\TestProb$ can be solved in linear time for both regular as well as refl-spanners (whereas the latter result is slightly more complicated, since it depends on using the longest common extension data structure).\par
Next, we consider the problem $\NonemptProb$ (i.\,e., to decide whether $S(w) \neq \emptyset$ for given $S$ and $w$), for which we can obtain the following theorem by standard methods.

\begin{thm}\label{thm:nonempRefl}
$\NonemptProb$ for $\reflspanners$ is $\npclass$-complete.
\end{thm}

\begin{proof}
The $\npclass$-hardness can be proven in a similar way as Theorem~3.3 in~\cite{FreydenbergerHolldack2018}. It follows easily by a reduction from the problem of matching patterns with variables, which is as follows: for a given pattern $\alpha \in (\Sigma \cup \varset)^*$ and a word $w \in \Sigma^*$, decide whether there is a mapping $h \colon \varset \to \Sigma^*$ such that $\widehat{h}(\alpha) = w$, where $\widehat{h}$ is the natural extension of $h$ to a morphism $(\varset \cup \Sigma)^* \to (\varset \cup \Sigma)^*$, by letting $h$ be the identity on $\Sigma$. This problem is $\npclass$-complete (for more information see, e.\,g.,~\cite{FernauEtAl2016, FernauSchmid2015, FreydenbergerHolldack2018}).\par
Let $\alpha \in (\varset \cup \Sigma)^*$ be a pattern with variables and let $w \in \Sigma^*$. Moreover, let $\beta$ be the regular expression obtained from $\alpha$ by replacing each first occurrence of a variable $\varsx$ by $\open{\varsx} \Sigma^* \close{\varsx}$. It can be easily seen that $\lang(\beta)$ is a ref-language over $\Sigma$ and $\varset$, and that $w \in \lang(\alpha)$ if and only if $\llbracket \lang(\beta) \rrbracket(w) \neq \emptyset$. This shows that $\NonemptProb$ for $\reflspanners$ is $\npclass$-hard.\par
For membership in $\npclass$, let $S = \sem{\lang(M)}\in \reflspanners_{\Sigma, \varset}$ and let $w \in \Sigma^*$. In order to check whether $S(w) \neq \emptyset$, we guess some $(\varset, w)$-tuple $t$, which can be done in polynomial time. Then we check whether $t \in S(w)$, which can also be done in polynomial time (see Theorem~\ref{checkingTheorem}).
\end{proof}

\subsection{Static Analysis}

We start with the following straightforward result. Recall that for the problems $\SatProb$, $\HierProb$ and $\FunctProb$, we get a single spanner $S$, represented by an $\NFA$, as input and ask whether there is a $w \in \Sigma^*$ with $S(w) \neq \emptyset$, whether $S$ is hierarchical, or whether $S$ is functional, respectively.

\begin{thm}\label{thm:SatAndHier}\hfill
\begin{mea}
\item\label{item:SatRefl}
$\SatProb$ for $\reflspanners$ can be solved in time $\bigO(|M|)$,
\item\label{item:HierRefl}
$\HierProb$ for $\reflspanners$ can be solved in time $\bigO(|M|\cdot|\varset|^3)$, and
\item\label{item:FunctRefl}
$\FunctProb$ for $\reflspanners$ can be solved in time $\bigO(|M|\cdot|\varset|^2)$,
\end{mea}
\noindent
where $M$ is an $\NFA$ that describes a refl-spanner over $\Sigma$ and $\X$.
\end{thm}

\begin{proof}
\ref{item:SatRefl}: Let $S \in \reflspanners_{\Sigma, \varset}$ be represented by an $\NFA$ $M$. \par
If $\lang(M) \neq \emptyset$, then there is a $v \in \lang(M)$ and therefore, by definition, $\getSpanTuple{\deref{v}} \in \llbracket \lang(M) \rrbracket(\getWord{v})$, which means that $S(\getWord{v}) \neq \emptyset$.
On the other hand, if there is some $w \in \Sigma^*$ with $S(w) \neq \emptyset$, then there is a span-tuple $t \in \llbracket \lang(M) \rrbracket(w)$, which means that there is a ref-word $v \in \lang(M)$ with $\getWord{\deref{v}} = w$ and $\getSpanTuple{\deref{v}} = t$. Thus, $\lang(M) \neq \emptyset$. Consequently, there is a word $w \in \Sigma^*$ with $S(w) \neq \emptyset$ if and only if $\lang(M) \neq \emptyset$. Checking whether $\lang(M) \neq \emptyset$ can be done in time $\bigO(|M|)$, by searching $M$ for a path from the start state to an accepting state.

\ref{item:HierRefl}:
Let $S \in \reflspanners_{\Sigma, \varset}$ be represented by an $\NFA$ $M$. It is straightforward to see that $S$ is hierarchical if and only if $\lang(M)$ is hierarchical. According to Proposition~\ref{CheckSubwordMarkedLanguagesProposition} the latter can be checked in time $\bigO(|M|\cdot|\varset|^3)$.

\ref{item:FunctRefl}:
Let $S \in \reflspanners_{\Sigma, \varset}$ be represented by an $\NFA$ $M$. It is straightforward to see that $S$ is functional if and only if $\lang(M)$ is functional. According to Proposition~\ref{CheckSubwordMarkedLanguagesProposition} the latter can be checked in time $\bigO(|M|\cdot|\varset|^2)$.
\end{proof}

With respect to Theorem~\ref{thm:SatAndHier}, it is worth recalling that for core spanners, $\SatProb$ and $\HierProb$ are $\pspaceclass$-complete, even for restricted classes of core spanners (see~\cite{FreydenbergerHolldack2018}).\par
Finally, we investigate the problems $\ContProb$ and $\EquProb$, which consist in deciding whether $S_1 \subseteq S_2$ or $S_1 = S_2$, respectively, for given spanners $S_1$ and $S_2$. \par
Recall that for core spanners $\ContProb$ and $\EquProb$ are not semi-decidable (see~\cite{FreydenbergerHolldack2018}). We now show that for refl-spanners we can achieve decidability of $\ContProb$ and $\EquProb$ by imposing suitable restrictions. \par
Let us first briefly discuss the case of regular spanners. Let $S_1$ and $S_2$ be regular spanners represented by $\NFA$s $M_1$ and $M_2$. Due to Proposition~\ref{setRepInclusionProp}, we can check whether $S_1 \subseteq S_2$ by checking whether $\msrep{\lang(M_1)} \subseteq \msrep{\lang(M_2)}$. Checking whether $\msrep{\lang(M_1)} \not \subseteq \msrep{\lang(M_2)}$ can be done in $\pspaceclass$ (i.\,e., we guess a witness $w'$ and check whether $w' \in \msrep{\lang(M_1)} \setminus \msrep{\lang(M_2)}$); we skip all the details here, since the result will follow from our more general result with respect to a subclass of refl-spanners; moreover, the $\pspaceclass$-completeness of containment of regular spanners has already been mentioned in~\cite{MaturanaEtAl2018}.\par
Let us give some intuition of why the situation is more complicated for refl-spanners. We first note that also for ref-words (since they are subword-marked words), we can use the function $\msrep{\cdot}$. However, while $\msrep{w_1} = \msrep{w_2}$ for ref-words $w_1, w_2$ is still sufficient for $\getWord{\deref{w_1}} = \getWord{\deref{w_2}}$ and $\getSpanTuple{\deref{w_1}} = \getSpanTuple{\deref{w_2}}$, it is not a necessary condition. For example, the ref-words 
\begin{align*}
w_1 = &\open{\varsx} \ta \tb \close{\varsx} \tb \open{\varsy} \ta \tb \tb \close{\varsy} \varsy \ta \tb\,, \\
w_2 = &\open{\varsx} \ta \tb \close{\varsx} \tb \open{\varsy} \varsx \tb \close{\varsy} \varsx \tb \varsx\,, \\
w_3 = &\open{\varsx} \ta \tb \close{\varsx} \tb \open{\varsy} \ta \tb \tb \close{\varsy} \varsy\varsx
\end{align*}
have all the same image $\open{\varsx} \ta \tb \close{\varsx} \tb \open{\varsy} \ta \tb \tb \close{\varsy} \ta \tb \tb \ta \tb$ under the function $\deref{\cdot}$ (and therefore they all describe the same document and the same span-tuple), but their marker-set representations are obviously pairwise different.\par
The main idea in the following is that we impose a restriction to ref-words (and therefore to ref-languages and refl-spanners) that allows us to obtain an analogue of Proposition~\ref{setRepInclusionProp} for refl-spanners. Let us first informally explain our approach. Intuitively speaking, we require all variable references to be extracted by their own private extraction variable, i.\,e., in the ref-words we encounter all variable references $\varsx$ in the form $\open{\varsy_{\varsx}} \varsx \close{\varsy_{\varsx}}$, where $\varsy_{\varsx}$ has in all ref-words the sole purpose of extracting the content of some reference of variable $\varsx$. With this requirement, the positions of the repeating factors described by variables and their references must be explicitly present as spans in the span-tuples. This seems like a strong restriction for refl-spanners, but we should note that for core spanners we necessarily have a rather similar situation: if we want to use string-equality selections on some spans, we have to explicitly extract them by variables first. \par
In the following, we assume that the set of variables is partitioned into the following sets. There is a set $\varset_{r}$ of \emph{reference-variables}, there is a set $\varset_{e}$ of \emph{extraction-variables}, and, for each reference-variable $\varsx \in \varset_r$, there is a set $\varset_{e, \varsx}$ of \emph{extraction-variables for reference-variable $\varsx$}.
The intuitive idea is that for every $\varsx \in \varset_r$, every reference of $\varsx$ is extracted by some $\varsy \in \varset_{e, \varsx}$, i.\,e., it occurs directly between $\open{\varsy}$ and $\close{\varsy}$; moreover, every $\varsy \in \varset_{e, \varsx}$ either does not occur at all, or it is used as extractor for $\varsx$, i.\,e., $\varsy$'s definition contains exactly one occurrence of a symbol from $\Sigma \cup \varset$ which is $\varsx$. The variables $\varset_e$ are also only used for extraction (i.\,e., they are never referenced), but their respective markers can parenthesise any kind of factor. In addition, we also require that for every $\varsx \in \varset_{r}$ with at least one reference, the image of $\varsx$ under $\deref{\cdot}$ is non-empty.

\begin{defi}\label{def:extracting}
A ref-word $w$ over $\Sigma$ and $\varset$ is a \emph{strongly reference extracting} ref-word over $\Sigma$ and $(\varset_{r}, \varset_{e}, \{\varset_{e, \varsx} \mid \varsx \in \varset_r\})$, if it satisfies the following:
\begin{itemize}
%\item $\{\varset_{r}, \varset_e, \varset_{e, \varsx} \mid \varsx \in \varset_r\}$ is a partition of $\varset$.
\item $\varset_{r}, \varset_{e}, \{\varset_{e, \varsx} \mid \varsx \in \varset_r\}$ is a partition of $\varset$.
\item For every $\varsx \in \varset_r$, if $w\spann{i}{i+1} = \varsx$, 
\begin{itemize}
\item 
then $w\spann{i-1}{i+2} = \open{\varsy} \varsx \close{\varsy}$ with $\varsy \in \varset_{e, \varsx}$, and
\item $\getSpanTuple{\deref{w}}(\varsx) = \spann{i'}{j'}$ with $j' - i' \geq 1$.
\end{itemize}
\item For every $\varsy \in \varset_e \cup \bigcup_{\varsx \in \varset_r} \varset_{e,\varsx}$ we have that $|w|_{\varsy} = 0$. 
\item For every $\varsy \in \varset_{e, \varsx}$, if $w\spann{i}{i+1} = \open{\varsy}$ then $w\spann{i}{i+3} = \open{\varsy} \varsx \close{\varsy}$.
\end{itemize}
\end{defi}

A ref-language $L$ over $\Sigma$ and $\varset$ is a \emph{strongly reference extracting} ref-language over $\Sigma$ and $(\varset_{r}, \varset_{e}, \{\varset_{e, \varsx} \mid \varsx \in \varset_r\})$ if every ref-word $w \in L$ is a strongly reference extracting ref-word over $\Sigma$ and $(\varset_{r}, \varset_{e}, \{\varset_{e, \varsx} \mid \varsx \in \varset_r\})$.

\begin{obs}\label{stronglyRefMsrepMainObs}
Let $w$ be a strongly reference extracting ref-word over $\Sigma$ and\linebreak $(\varset_{r}, \varset_{e}, \{\varset_{e, \varsx} \mid \varsx \in \varset_r\})$. 
\begin{enumerate}
\item\label{stronglyRefMsrepMainObsPointOne} For every $\varsx \in \varset_r$, every occurrence of $\varsx$ in $\msrep{w}$ is directly preceded by (and followed by) a symbol $\Gamma \subseteq \Gamma_{\varset}$ such that $\Gamma$ contains some $\open{\varsy}$ (some $\close{\varsy}$, respectively) with $\varsy \in \varset_{e, \varsx}$ and no other $\open{\varsy'}$ ($\close{\varsy'}$, respectively) with $\varsy' \in \bigcup_{\varsx' \in \varset_{r}} \varset_{e, \varsx'}$ and $\varsy \neq \varsy'$.
\item\label{stronglyRefMsrepMainObsPointTwo} 
Every occurrence of a symbol $\Gamma \subseteq \Gamma_{\varset}$ with $\open{\varsy} \in \Gamma$ (or $\close{\varsy} \in \Gamma$) for some $\varsy \in \varset_{e, \varsx}$, is followed by (preceded by, respectively) a sequence of occurrences of $\Gamma$ followed by (preceded by) a reference $\varsx$. 
\end{enumerate}
\end{obs}

\begin{lem}\label{setRepSimpleLemmaRefVersion}
Let $w_1$ and $w_2$ be strongly reference extracting ref-words over $\Sigma$ and $(\varset_{r}, \varset_{e}, \{\varset_{e, \varsx} \mid \varsx \in \varset_r\})$. Then $\getWord{\deref{w_1}} = \getWord{\deref{w_2}}$ and $\getSpanTuple{\deref{w_1}} = \getSpanTuple{\deref{w_2}}$ if and only if $\msrep{w_1} = \msrep{w_2}$.
\end{lem}

\begin{proof}
If $\msrep{w_1} = \msrep{w_2}$, then we obviously also have $\msrep{\deref{w_1}} = \msrep{\deref{w_2}}$. Since both $\deref{w_1}$ and $\deref{w_2}$ are subword-marked words, we can use Proposition~\ref{setRepSimpleProp} to conclude that $\getWord{\deref{w_1}} = \getWord{\deref{w_2}}$ and $\getSpanTuple{\deref{w_1}} = \getSpanTuple{\deref{w_1}}$. \par
Next, we assume that $\getWord{\deref{w_1}} = \getWord{\deref{w_2}}$ and $\getSpanTuple{\deref{w_1}} = \getSpanTuple{\deref{w_2}}$. Since both $\deref{w_1}$ and $\deref{w_2}$ are subword-marked words, we can conclude with Proposition~\ref{setRepSimpleProp} that $\msrep{\deref{w_1}} = \msrep{\deref{w_2}}$. In order to conclude the proof, we have to show that $\msrep{w_1} = \msrep{w_2}$. To this end, we show that the assumption $\msrep{w_1} \neq \msrep{w_2}$ leads to a contradiction. \par
We prove by induction that, for every $i \in \{0, 1, \ldots, \max\{|\msrep{w_1}|, |\msrep{w_2}|\}\}$, $\msrep{w_1}\spann{1}{i+1} = \msrep{w_2}\spann{1}{i+1}$. As the base of the induction, we observe that $\msrep{w_1}\spann{1}{0+1} = \msrep{w_1}\spann{1}{1} = \emptyword$ and $\msrep{w_2}\spann{1}{0+1} = \emptyword$. Next, we assume that, for some $i \in \{0, 1, \ldots, \max\{|\msrep{w_1}|, |\msrep{w_2}|\} - 1\}$, we have that $\msrep{w_1}\spann{1}{i+1} = \msrep{w_2}\spann{1}{i+1}$. We let 
$\sigma_1 = \msrep{w_1}\spann{i+1}{i+2}$ and $\sigma_2 = \msrep{w_2}\spann{i+1}{i+2}$. We make a case distinction with respect to $\sigma_1$ (note that by symmetry this will cover all possible cases).\medskip \par
\textbf{Case $\sigma_1 \subseteq \Gamma_{\varset}$}: In this case, $\sigma_1 \neq \sigma_2$ would directly imply that for some $\varsx \in \varset$, we have that $\getSpanTuple{\deref{w_1}}(\varsx) \neq \getSpanTuple{\deref{w_2}}(\varsx)$ (note that this holds for all possible choices of $\sigma_2$). Hence, we would get the contradiction that $\getSpanTuple{\deref{w_1}} \neq \getSpanTuple{\deref{w_2}}$, which implies that we have $\sigma_1 = \sigma_2$.\par
\textbf{Case $\sigma_1 \in \Sigma$ and $\sigma_2 \in \Sigma$}: This directly implies that there is some position $j$ such that $\msrep{\deref{w_1}}\spann{j}{j+1} = \sigma_1$ and $\msrep{\deref{w_2}}\spann{j}{j+1} = \sigma_2$. Thus, since $\msrep{\deref{w_1}} = \msrep{\deref{w_2}}$, we can directly conclude that $\sigma_1 = \sigma_2$.\par
\textbf{Case $\sigma_1 = \varsx \in \varset_r$}: If $\sigma_2 \subseteq \Gamma_{\varset}$, then we can proceed like in the first case, but with $\sigma_2$ playing the role of $\sigma_1$. Hence, we may assume that $\sigma_2 \in \varset \cup \Sigma$. By Point~\ref{stronglyRefMsrepMainObsPointOne} of Observation~\ref{stronglyRefMsrepMainObs}, we can conclude that $\msrep{w_1}\spann{i}{i+1} = \Gamma \subseteq \Gamma_{\varset}$ such that $\open{\varsy} \in \Gamma$ for some $\varsy \in \varset_{e, \varsx}$. By assumption, we therefore also have $\msrep{w_2}\spann{i}{i+1} = \Gamma$. Hence, Point~\ref{stronglyRefMsrepMainObsPointTwo} of Observation~\ref{stronglyRefMsrepMainObs} implies that symbol $\Gamma$ at position $i$ of $\msrep{w_2}$ is followed by a sequence of occurrences of $\Gamma$ followed by $\varsx$. Since we made the assumption that $\sigma_2 \in \varset \cup \Sigma$, we know that $\sigma_2 \neq \Gamma$, which means that $\sigma_2 = \varsx$. Consequently, $\sigma_1 = \sigma_2$.\medskip\par
We have shown that
$\msrep{w_1}\spann{1}{i+1} = \msrep{w_2}\spann{1}{i+1}$ holds 
for every $i \in \{0, 1, \ldots,
\max\{|\msrep{w_1}|, |\msrep{w_2}|\}\}$. This means that either $\msrep{w_1} =
\msrep{w_2}$, or one of these words is a proper prefix of the
other. Let us assume that the latter case applies and, without loss of
generality, that $\msrep{w_1}$ is a proper prefix of
$\msrep{w_2}$. More formally, for some $j \in \{|\msrep{w_1}|,
|\msrep{w_1}|+1, \ldots, |\msrep{w_2}|\}$, we have that
$\msrep{w_2}\spann{1}{j+1} = \msrep{w_1}$, and $|\msrep{w_2}| >
|\msrep{w_1}|$. If $\msrep{w_2}\spann{j+1}{j+2} \in \Sigma$ or
$\msrep{w_2}\spann{j+1}{j+2} \in \varset$, then $|\deref{w_2}| >
|\deref{w_1}|$, which leads to the contradiction that
$\msrep{\deref{w_1}} \neq \msrep{\deref{w_2}}$ (note that in the
latter case, i.\,e., $\msrep{w_2}\spann{j+1}{j+2} \in \varset$, it is
important that for strongly reference extracting ref-words we require
for every $\varsx \in \varset_r$ with $|w|_{\varsx} \neq 0$ that
$\getSpanTuple{\deref{w}}(\varsx) = \spann{i}{j}$ with $j - i \geq
2$). Moreover, if $\msrep{w_2}\spann{j+1}{j+2} = \Gamma \subseteq
\Gamma_{\varset}$, then, for some $\varsx \in \varset$, we have that
$\getSpanTuple{\deref{w_1}}(\varsx) \neq
\getSpanTuple{\deref{w_2}}(\varsx)$. This leads to the contradiction
that $\getSpanTuple{\deref{w_1}} \neq
\getSpanTuple{\deref{w_2}}$. Consequently, it is not possible that one
of $\msrep{w_1}$ or $\msrep{w_2}$ is a proper prefix of the other;
thus, $\msrep{w_1} = \msrep{w_2}$.
This concludes the proof of Lemma~\ref{setRepSimpleLemmaRefVersion}.
\end{proof}

\begin{lem}\label{setRepInclusionLemmaRefVersion}
Let $L_1, L_2$ be strongly reference extracting ref-languages over $\Sigma$ and $(\varset_{r}, \varset_{e},\allowbreak \{\varset_{e, \varsx} \mid \varsx \in \varset_r\})$. Then $\llbracket L_1 \rrbracket \subseteq \llbracket L_2 \rrbracket$ if and only if $\msrep{L_1} \subseteq \msrep{L_2}$.
\end{lem}

\begin{proof}
Let us first assume that $\llbracket L_1 \rrbracket \subseteq \llbracket L_2 \rrbracket$. Let $v_1 \in \msrep{L_1}$ and let $v'_1 \in L_1$ be such that $\msrep{v'_1} = v_1$. Moreover, let $w = \getWord{\deref{v'_1}}$, which also means that $\getSpanTuple{\deref{v'_1}} \in \llbracket L_1 \rrbracket(w)$. Since $\llbracket L_1 \rrbracket \subseteq \llbracket L_2 \rrbracket$, we have $\llbracket L_1 \rrbracket(w) \subseteq \llbracket L_2 \rrbracket(w)$ and therefore $\getSpanTuple{\deref{v'_1}} \in \llbracket L_2 \rrbracket(w)$. By definition, this means that there is some $v'_2 \in L_2$ with $\getWord{\deref{v'_2}} = w = \getWord{\deref{v'_1}}$ and $\getSpanTuple{\deref{v'_2}} = \getSpanTuple{\deref{v'_1}}$. By Proposition~\ref{setRepSimpleLemmaRefVersion}, this means that $\msrep{v'_1} = \msrep{v'_2}$, and, by definition, $\msrep{v'_2} \in \msrep{L_2}$. Thus, $\msrep{v'_1} = v_1 \in \msrep{L_2}$. Since $v_1 \in \msrep{L_1}$ has been chosen arbitrarily, we can conclude that $\msrep{L_1} \subseteq \msrep{L_2}$. \par
Next, we assume that $\msrep{L_1} \subseteq \msrep{L_2}$. Let $w \in
\Sigma^*$ and $t \in \llbracket L_1 \rrbracket(w)$ be arbitrarily
chosen. This means that there is some $v_1 \in L_1$ with
$\getWord{\deref{v_1}} = w$ and $\getSpanTuple{\deref{v_1}} = t$. By
assumption, $\msrep{v_1} \in \msrep{L_2}$, which means that there is some $v_2 \in L_2$ with $\msrep{v_1} = \msrep{v_2}$, which, by Proposition~\ref{setRepSimpleLemmaRefVersion}, means that $\getWord{\deref{v_1}} = \getWord{\deref{v_2}} = w$ and $\getSpanTuple{\deref{v_1}} = \getSpanTuple{\deref{v_2}} = t$. Thus, $t \in \llbracket L_2 \rrbracket(w)$. Since $w \in \Sigma^*$ and $t \in \llbracket L_1 \rrbracket(w)$ have been chosen arbitrarily, we can conclude that $\llbracket L_1 \rrbracket \subseteq \llbracket L_2 \rrbracket$. 
\end{proof}

\begin{thm}\label{mainContainmentDecidabilityTheorem}
$\ContProb$ and $\EquProb$ for strongly reference extracting refl-spanners over $\Sigma$ and $(\varset_{r}, \varset_{e}, \{\varset_{e, \varsx} \mid \varsx \in \varset_r\})$ are $\pspaceclass$-complete.
\end{thm}

\begin{proof}
We first note that hardness follows from the fact that the inclusion and equivalence problem for $\NFA$s is $\pspaceclass$-hard (see~\cite{MeyerStockmeyer1972,StockmeyerMeyer1973}). Next, we prove the upper bound for the case of $\ContProb$ (which also implies the upper bound for $\EquProb$). Note that the general proof idea is very similar to \cite[Thm. 6.4]{MaturanaEtAl2018_arxiv}, i.\,e., the $\pspaceclass$-completeness of regular spanners with schemaless semantics.\par
Let $L_1, L_2$ be strongly reference extracting ref-languages over
$\Sigma$ and $(\varset_{r}, \varset_{e}, \{\varset_{e, \varsx} \mid
\varsx \in \varset_r\})$, represented by $\NFA$ $M_1$ and $M_2$. We
wish to decide whether $\llbracket L_1 \rrbracket \subseteq \llbracket
L_2 \rrbracket$. By Lemma~\ref{setRepInclusionLemmaRefVersion}, this
can be done by checking whether $\msrep{L_1} \subseteq
\msrep{L_2}$. We shall now devise a nondeterministic algorithm that
checks whether $\msrep{L_1} \not \subseteq \msrep{L_2}$. Note that we cannot modify automata for $L_1$ and $L_2$ in polynomial space such that they accept the languages $\msrep{L_1}$ and $\msrep{L_2}$, respectively. \par 
Without loss of generality, we assume that $M_1$ and $M_2$ are complete. The algorithm guesses words $w_1, w_2$ over $\Sigma \cup \Gamma_{\varset} \cup \varset$ letter by letter, which satisfy $\msrep{w_1} = \msrep{w_2}$, and checks whether $w_1$ is accepted by $M_1$ and rejected by $M_2$. To this end, we run the automata in parallel and, since the automata are nondeterministic, we have to explore all possible paths labelled with $w_1$ and $w_2$, which can be done by maintaining the sets $S_1$ and $S_2$ of current active states of $M_1$ and $M_2$, respectively. More precisely, we guess the next symbols $\sigma_1$ and $\sigma_2$ of $w_1$ and $w_2$, and then we update the sets of active states by applying the $\sigma_1$- and $\sigma_2$-transitions for the active states of $M_1$ and $M_2$, respectively. Obviously, in order to satisfy $\msrep{w_1} = \msrep{w_2}$, we must have $\sigma_1 = \sigma_2$ whenever $\sigma_1 \in \Sigma \cup \varset$. If, however, in $w_1$ we guess that instead of a symbol from $\Sigma$ or $\varset$ a sequence over $\Gamma_{\varset}$ follows, then we must process with $M_1$ and $M_2$ all possible permutations of this sequence. This can be done as follows. If we guess the next $\sigma_1$ to be neither from $\Sigma$ or from $\varset$, we guess a set $\Gamma \subseteq \Gamma_{\varset}$ instead. Then we enumerate all linear orders $\vec{\Gamma}$ of the symbols from $\Gamma$, and for each such $\vec{\Gamma}$, we update the sets of active states by exploring all $\vec{\Gamma}$-labelled paths in $M_1$ and $M_2$ that start in states of $S_1$ and $S_2$, respectively. Moreover, after such a step, the algorithm must next again guess some $\sigma_1$ from $\Sigma$ or $\varset$.\par
We start this procedure with $S_1 = \{s_{0, 1}\}$ and $S_2 = \{s_{0, 2}\}$, where $s_{0, 1}$ and $s_{0, 2}$ are the initial states of $M_1$ and $M_2$, respectively. The algorithm terminates with output $\mathsf{yes}$ as soon as $S_1$ contains an accepting state and $S_2$ does not contain an accepting state. Note that if $S_1$ contains an accepting state, then this also means that the currently read input $w_1$ is a valid ref-word, which, since $\msrep{w_1} = \msrep{w_2}$, also means that $w_2$ is a valid ref-word. \par
We shall now consider the correctness of the algorithm. The basic observation is that, for some (nondeterministically chosen) word over $w \in \Sigma \cup \Gamma_{\varset} \cup \varset$, the algorithm simulates $M_1$ and $M_2$ in parallel on \emph{all} inputs $w'$ that satisfy $\msrep{w'} = \msrep{w}$. 
If $\msrep{L_1} \not \subseteq \msrep{L_2}$, then there is some $w \in \lang(M_1)$ such that for every $w'$ over $\Sigma \cup \Gamma_{\varset} \cup \varset$ with $\msrep{w} = \msrep{w'}$, we have that $w' \notin \lang(M_2)$. Consequently, if the algorithm guesses $w$ (this means it guesses $w$'s symbols from $\Sigma \cup \varset$ in the right order, and whenever some sequence $\gamma$ over $\Gamma_{\varset}$ follows, it guesses the set $\Gamma$ that contains exactly the symbols from $\gamma$), then, after having completely read $w$, we will reach a set of active states $S_1$ with an accepting state, while, due to $w' \notin \lang(M_2)$ for all $w'$ with $\msrep{w} = \msrep{w'}$, the set $S_2$ of active states cannot contain an accepting state. This means that the algorithm produces output $\mathsf{yes}$. On the other hand, if the algorithm produces output $\mathsf{yes}$, then it has guessed some $w$ such that, once it its completely consumed, $S_1$ contains an accepting state while $S_2$ does not contain any accepting state. Consequently, $\lang(M_1)$ accepts some $w'$ with $\msrep{w'} = \msrep{w}$, while there is no $w'' \in \lang(M_2)$ with $\msrep{w''} = \msrep{w}$. This directly implies that $\msrep{w} \in \msrep{L_1}$ and $\msrep{w} \notin \msrep{L_2}$; thus, $\msrep{L_1} \not \subseteq \msrep{L_2}$.\par
With respect to the complexity of the algorithm, we observe that we only store sets $S_1$ and $S_2$ that are subsets of the states of $M_1$ and $M_2$, respectively. Furthermore, updating these steps can be easily done in polynomial space by consulting the transition functions of $M_1$ and $M_2$. In particular, we observe that guessing a set $\Gamma \subseteq \Gamma_{\varset}$ and enumerating all linear orders $\vec{\Gamma}$ of the symbols from $\Gamma$ can also be done in polynomial space. 
\end{proof}

\section{Expressive Power of Refl-Spanners}\label{sec:ExpressivePower}

It is a straightfoward observation that the expressive power of refl-spanners properly exceeds the one of regular spanners, but it is less clear which refl-spanners are also core spanners and which core spanners are refl-spanners. 

\subsection{From Refl-Spanners to Core Spanners}

A ref-language $L$ over $\Sigma$ and $\varset$ is \emph{reference-bounded} if there is a number $k$ with $|w|_{\varsx} \leq k$ for every $\varsx \in \varset$ and every $w \in L$. A refl-spanner is \emph{reference-bounded} if it is represented by a reference-bounded ref-language. The following is easy to see.

\begin{thm}\label{reflToCoreTheorem}
Every reference-bounded refl-spanner is a core spanner.
\end{thm}

\begin{proof}
Let $S$ be a refl-spanner over $\Sigma$ and $\varset$ represented by an $\NFA$ $M$, and let $k \in \mathbb{N}$ be such that $|w|_{\varsx} \leq k$ for every $\varsx \in \varset$ and every $w \in \lang(M)$. For every $\varsx \in \varset$, let $\mathcal{Y}_{\varsx} = \{\varsy_{\varsx, 1}, \varsy_{\varsx, 2}, \ldots, \varsy_{\varsx, k}\}$ be $k$ fresh variables. We define an $\NFA$ $M'$ as follows. The automaton $M'$ simulates $M$, but for every $\varsx \in \varset$, it maintains a counter $c_{\varsx}$ in its finite state control that is initially $0$ and that can hold values from $\{0\} \cup [k]$. Whenever $M$ takes an $\varsx$-transition, $M'$ increments $c_{\varsx}$ and then can read any word from $\lang(\open{\varsy_{\varsx, c_{\varsx}}} \Sigma^* \close{\varsy_{\varsx, c_{\varsx}}})$. It can be easily seen that $\lang(M')$ is a subword-marked language over $\Sigma$ and $\varset$. Furthermore, it can be shown with moderate effort that $\llbracket \lang(M) \rrbracket = \pi_{\varset} \varsigma^{=}_{\{\mathcal{Y}_{\varsx} \mid \varsx \in \varset\}}(\llbracket \lang(M') \rrbracket)$. 
\end{proof}

It is interesting to note that the refl-spanner $\llbracket \lang(\ta^+  \open{\varsx} \tb^+ \close{\varsx} (\ta^+ \varsx)^* \ta^+) \rrbracket$, which is \emph{not} reference-bounded, is provably not a core spanner (see \cite[Theorem 6.1]{FaginEtAl2015})).

\subsection{From Core Spanners to Refl-Spanners}\label{nonOverlappingCoreSpanner}

The question of which core spanners can be represented as refl-spanners is a much more difficult one and we shall investigate it in more detail. There are simple core spanners which translate to refl-spanners in an obvious way, e.\,g., $\pi_{\{\varsx\}} \varsigma_{\{\varsx, \varsy\}}^{=} \llbracket L \rrbracket$ with $L = \lang(\open{\varsx} (\ta^* \altop \tb^*) \close{\varsx} \tc \open{\varsy} (\ta^* \altop \tb^*) \close{\varsy})$ can be represented as $\llbracket L' \rrbracket$ where $L' = \lang(\open{\varsx} (\ta^* \altop \tb^*) \close{\varsx} \tc \varsx)$. However, if we change $L$ to $\lang(\open{\varsx} \Sigma^* \ta \Sigma^* \close{\varsx} \tc \open{\varsy} \Sigma^* \tb \Sigma^* \close{\varsy})$, then neither of the ref-languages $\lang(\open{\varsx} \Sigma^* \ta \Sigma^* \close{\varsx} \tc \varsx)$ nor $\lang(\open{\varsx} \Sigma^* \tb \Sigma^* \close{\varsx} \tc \varsx)$ yield an equivalent refl-spanner and we have to use $\lang(\open{\varsx} r \close{\varsx} \tc \varsx)$, where $r$ is a regular expression for $\lang(\Sigma^* \ta \Sigma^*) \cap \lang(\Sigma^* \tb \Sigma^*)$. \par
Another problem is that core spanners can also use string-equality selections on spans that contain start or end positions of other spans. For example, it seems difficult to transform $\varsigma_{\{\varsx, \varsy\}}^{=} \llbracket \lang(\open{\varsx} \ta^* \close{\varsx} \open{\varsy} \open{\varsz} \ta^* \close{\varsz} \ta^* \close{\varsy}) \rrbracket$ into a refl-spanner. The situation gets even more involved if we use the string-equality selections directly on overlapping spans, e.\,g., as in core spanners of the form $\varsigma^{=}_{\{\varsx, \varsy\}}(\llbracket \lang(\open{\varsx} \ldots \open{\varsy} \ldots \close{\varsx} \ldots \close{\varsy})\rrbracket)$. For an in-depth analysis of the capability of core spanners to describe word-combinatorial properties, we refer to~\cite{FreydenbergerHolldack2018, Freydenberger2019}. \par
These considerations suggest that the refl-spanner formalism is less powerful than core spanners, which is to be expected, since we have to pay a price for the fact that we can solve many problems for refl-spanners much more efficiently than for core spanners (see our results presented in Section~\ref{sec:decisionProblems} and summarised in Table~\ref{comparisonTable}). However, we can show that a surprisingly large class of core spanners can nevertheless be represented by a single refl-spanner along with the application of a simple spanner operation (to be defined in the next subsection) that just combines several variables (or columns in the spanner result) into one variable (or column) in a natural way, and a projection; see Theorem~\ref{coreSpannersToReflSpannersTheorem} for our respective main result.

In this section, we shall also use the marker-set representation of
subword-marked words (and therefore ref-words) as defined at the
beginning of Section~\ref{sec:decisionProblems}. However, in the
following it will be more convenient to simplify this concept as
follows. Let $w$ be a subword-marked word over $\Sigma$ and $\varset$,
let $t = \getSpanTuple{w}$, and let $\getWord{w} = w_1 w_2 \ldots w_n$
with $w_i \in \Sigma$ for every $i \in [n]$. For each $i\in
\set{1,\ldots,n{+}1}$ let $\Gamma_i$ again be the set comprising of all symbols $\open{\x}$ where $\x\in\Dom{t}$ and $t(\x)=\spann{i}{j}$ for some $j$, and all symbols $\close{\x}$ where $\x\in \Dom{t}$ and $t(\x)=\spann{j}{i}$ for some $j$. We now set $\msrep{w} = \Gamma'_1 w_1 \Gamma'_2 w_2 \ldots \Gamma'_n w_n \Gamma'_{n+1}$, where, for every $i \in \{1, 2, \ldots, n+1\}$, $\Gamma'_i = \Gamma_i$ if $\Gamma_i \neq \emptyset$, and $\Gamma'_i = \varepsilon$ otherwise. Note that this corresponds to our original version of the marker-set representation with the only difference that the factors $\Gamma_i^{|\Gamma_i|}$ are replaced by just one occurrence of the symbol $\Gamma_i$, or completely erased if $\Gamma_i = \emptyset$. Moreover, note that we have re-defined $\msrep{\cdot}$ and that from now on, we will only use this definition of the marker-set representation. \par
For example, if $w = \open{\varsx_1} \open{\varsx_3} \ta \tb \close{\varsx_1} \open{\varsx_2} \tc \tb \close{\varsx_2} \ta \tb \tc \ta \close{\varsx_3}$, then we have $\msrep{w} = \{\open{\varsx_1}, \open{\varsx_3}\} \ta \tb \{\close{\varsx_1}, \open{\varsx_2}\} \tc \tb \{\close{\varsx_2}\} \ta \tb \tc \ta \{\close{\varsx_3}\}$.\par
We again extend the function $\msrep{\cdot}$ in the natural way to a subword-marked language $L$ by setting $\msrep{L} = \{\msrep{w} \mid w \in L\}$. 

\begin{rem}\label{msrepRemark}
In this section, we assume that all our regular and refl-spanners are represented as (automata that accept) subword-marked languages and ref-languages, respectively, all the words of which are given in marker-set representation. It can be easily seen that any $\NFA$ accepting a subword-marked language or a ref-language can be transformed into an $\NFA$ that accepts the $\msrep{\cdot}$ image of that language. This transformation may cause an exponential blow-up, which is no problem, since we are here only concerned with questions of expressive power, and not with complexity issues. 
\end{rem}

We next give an intuitive explanation of the result to be proven in
this subsection. As the central question of this subsection, we
investigate which core spanners of the form $S =
\varsigma_{E}^{=}(\llbracket L \rrbracket)$ (where $L$ is a regular
subword-marked language) can be described by $\llbracket L'
\rrbracket$ for some regular ref-language $L'$. Our first respective
observation is that $S$ can in fact be described in the form
$\llbracket L' \rrbracket$ for a regular ref-language $L'$, provided
that all the variables that are subject to the string-equality
selection $\varsigma_{E}^{=}$ are \emph{simple} (with respect to
$L$). A variable $\varsx$ is called simple with respect to a subword-marked language $L$, if any definition for $\varsx$ (in any $w \in L$) is always of the form $\Gamma_1 u \Gamma_2$ with $\Gamma_1, \Gamma_2 \subseteq \Gamma_{\varset}$, $\open{\varsx} \in \Gamma_1, \close{\varsx} \in \Gamma_2$ and $u \in \Sigma^+$ (or it defines the empty string, i.\,e., there is an occurrence of $\Gamma_1 \subseteq \Gamma_{\varset}$ with $\{\open{\varsx}, \close{\varsx}\} \subseteq \Gamma_1$). 

This result -- i.\,e., that a core spanner $\varsigma_{E}^{=}(\llbracket L \rrbracket)$ is a refl-spanner, if all the variables that are subject to the string-equality selection are simple -- is in fact a generalisation of the following example already mentioned above (note that both $\varsx$ and $\varsy$ are simple): 
\begin{equation*}
\varsigma_{\{\varsx, \varsy\}}^{=} \llbracket \lang(\open{\varsx}
(\ta^* \altop \tb^*) \close{\varsx} \tc \open{\varsy} (\ta^* \altop
\tb^*) \close{\varsy}) \rrbracket \ \ = \ \ \llbracket \lang(\open{\varsx} (\ta^* \altop \tb^*) \close{\varsx} \tc \open{\varsy} \varsx \close{\varsy}) \rrbracket\,.
\end{equation*}

Next, we show that we can use this result in order to prove that a much larger class of
core spanners can also be described by ref-languages, if we also allow
an additional operation on spanners, which we call the
\emph{span-fusion} $\bigspanfusion$. Intuitively, for a set $\lambda =
\{\varsy_1, \varsy_2, \ldots, \varsy_k\} \subseteq \varset$ of
variables, and a completely new variable $\varsx$ with $\varsx \notin
\varset$, the span-fusion $\bigspanfusion_{\lambda \to \varsx}$
replaces in a span-relation the set of columns $\{\varsy_1, \varsy_2,
\ldots, \varsy_k\}$  by a single new column $\varsx$ that, for each
row $t$ (i.\,e., span-tuple $t$), contains a single span that
corresponds to the union of all the spans $t(\varsy_1), t(\varsy_2),
\ldots, t(\varsy_k)$. For example, $\bigspanfusion_{\{\varsx_1,
  \varsx_2\} \to \varsz}$ turns a $\{\varsx_1, \varsx_2,
\varsx_3\}$-tuple $t = (\spann{3}{17}, \spann{17}{23}, \spann{4}{20})$
into a $\{\varsz, \varsx_3\}$-tuple $t'$ with $t'(\varsz) =
\spann{3}{23}$ and $t'(\varsx_3) = \spann{4}{20})$. The span-fusion operation will be formally defined below in Section~\ref{sec:spanFusion}.\par
We can show that provided that $\varsigma_{E}^{=}$ satisfies a certain
\emph{non-overlapping} property (basically, we require that any two
variables that are subject to the string-equality selection are
non-overlapping with respect to $L$), we can represent
$\varsigma_{E}^{=}(\llbracket L \rrbracket)$ in the form $\bigspanfusion_{\lambda_1 \to \varsx_1} \bigspanfusion_{\lambda_2 \to \varsx_2} \ldots \bigspanfusion_{\lambda_k \to \varsx_k}(\llbracket L' \rrbracket)$ for a regular ref-language $L'$ over $\Sigma$ and a set $\varset'$ of variables with $|\varset'| = \bigO(|\varset|^3)$.

The proof idea for this result is as follows (formal details are given later on). Every regular
subword-marked language $L$ over $\Sigma$ and $\varset$ can be
transformed into a subword-marked language $L'$ over $\Sigma$ and
$\varset'$ such that, for every $w \in L$, there is a $w' \in L'$ with
$\getWord{w} = \getWord{w'}$ and $\getSpanTuple{w'}$ is a \emph{split}
of $\getSpanTuple{w}$. An $\varset'$-tuple $t'$ is a split of an
$\varset$-tuple $t$, if every span $t(\varsx) =
\spann{\ell_{\varsx}}{r_{\varsx}}$ is factorised into $\varno$ factors
with respect to $t'$, i.\,e., there are variables $\varsx^1, \varsx^2,
\ldots, \varsx^{\varno} \in \varset'$ such that $t'(\varsx^1) =
\spann{\ell_{\varsx}}{k_1}, t'(\varsx^2) = \spann{k_1}{k_2}, \ldots,
t'(\varsx^{\varno - 1}) = \spann{k_{\varno - 1}}{k_{\varno}},
t'(\varsx^{\varno}) = \spann{k_{\varno}}{r_{\varsx}}$ (the number
$\varno$ is explained later and satisfies $\varno =
\bigO(|\varset|^2)$). Moreover, we require these factorisations to be
such that the spans for any two variables are non-overlapping. The
subword-marked language $L'$ of all these splits of subword-marked
words from $L$ will be called the \emph{split} of $L$, and it can be
shown that if $L$ is regular, then so is its split $L'$. In
particular, by definition of the span-fusion, we immediately have that
$\bigspanfusion_{\Lambda} \llbracket L' \rrbracket = \llbracket L
\rrbracket$, where $\Lambda = \{\lambda_\varsx \to \varsx \mid \varsx
\in \varset\}$ with $\lambda_{\varsx} = \{\varsx^1, \varsx^2, \ldots,
\varsx^{\varno}\}$ for every $\varsx \in \varset$. The split-operation will be formally defined below in Section~\ref{sec:splitOp}.\par
The next idea is that we can simulate the string-equality selection
$\varsigma_{E}^{=}$ (for $L$) with $E = \{\mathcal{Z}_1,
\mathcal{Z}_2, \ldots, \mathcal{Z}_k\}$ directly on the split
$L'$. Intuitively speaking, for every $\varsx, \varsy \in
\mathcal{Z}_i$, we require the pairs $\varsx^j, \varsy^j$ to be
subject to a string-equality selection for every $j \in [\varno]$, or,
more formally, we define the string-equality selection
$\varsigma^=_{E'}$, where $E' = \{\mathcal{Y}_i^j \mid i \in [k], j
\in [\varno]\}$ and $\mathcal{Y}_i^{j} = \{\varsx^j \mid \varsx \in
\mathcal{Z}_i\}$ for $i \in [k]$ and $j \in [\varno]$. It can be
easily seen that if $t \in \bigspanfusion_{\Lambda}
\varsigma_{E'}^{=}(\llbracket L' \rrbracket)(w)$ for some $w \in
\Sigma^*$, then we also have $t \in \varsigma_{E}^{=}(\llbracket L
\rrbracket)(w)$. The converse, however, is only true if $t$ has a
split $t'$ that is in $\varsigma_{E'}^{=}(\llbracket L'
\rrbracket)(w)$. A split $t'$ of $t$ is only in
$\varsigma_{E'}^{=}(\llbracket L' \rrbracket)(w)$ if, for every $i \in
[k]$, all $t(\varsx)$ with $\varsx \in \mathcal{Z}_i$ are factorised
into the parts $t'(\varsx^1), t'(\varsx^2), \ldots,
t'(\varsx^{\varno})$ in exactly the same way, since otherwise $t'$
would not satisfy $\varsigma_{E'}^{=}$. In the general case, the set
$\varset'$ of variables is not large enough to make such a split $t'$
possible for every $t \in \varsigma_{E}^{=}(\llbracket L
\rrbracket)(w)$. The problem is that due to the property of $t'$ that
all the spans for each two variables of $\varset'$ are
non-overlapping, the number of variables needed for this property may
depend on $|w|$ (as will also be demonstrated below by an
example). This, however, can only happen if $E$ is \emph{overlapping},
i.\,e., there are variables $\varsx, \varsy \in \bigcup_{i \in [k]}
\mathcal{Z}_i$ such that the spans of $\varsx$ and $\varsy$ may refer
to overlapping spans. Consequently, we have $\bigspanfusion_{\Lambda} \varsigma_{E'}^{=}(\llbracket L' \rrbracket) = \varsigma_{E}^{=}(\llbracket L \rrbracket)$, if $E$ is \emph{non-overlapping} (see Lemmas~\ref{splitEasyDirectionLemma}~and~\ref{splitHardDirectionLemma}). \par
Finally, since the split $L'$ of $L$ has necessarily only simple variables, we can express $\varsigma_{E'}^{=}(\llbracket L' \rrbracket)$ as $\llbracket L'' \rrbracket$ for a regular ref-language $L''$ (recall that this was our first result, sketched at the beginning of this subsection). Consequently, $\varsigma_{E}^{=}(\llbracket L \rrbracket)$ can be expressed as $\bigspanfusion_{\Lambda} (\llbracket L'' \rrbracket)$, where $L''$ is a regular ref-language.\par
In the following subsections, we shall formally carry out this proof roadmap.

\subsubsection{Simple Variables}

As sketched above, a variable is called simple with respect to a
subword-marked language if all its definitions do not contain parts of
other variable definitions, i.\,e., symbols $\Gamma \subseteq
\varset$. We will show next that string-equality selections for simple variables can be expressed by the ref-language mechanism, i.\,e., by using variable references. 

Let us now define simple variables formally.\footnote{Recall that by Remark~\ref{msrepRemark} we assume that all subword-marked and all ref-languages  are given in marker-set representation.} Let $L$ be a subword-marked language over $\Sigma$ and $\varset$, let $w \in L$ and $t = \getSpanTuple{w}$. We say that  a variable $\varsx \in \varset$ is \emph{simple} (\emph{with respect to $w$}) if $t(\varsx)$ is undefined or the empty span, or $w = u_1 \Gamma u_2 \Gamma' u_3$ with $\open{\varsx} \in \Gamma$ and $\close{\varsx} \in \Gamma'$ implies that $u_2 \in \Sigma^+$. For example, $\varsx$ and $\varsz$ are simple in $\ta \tb \{\open{\varsx}, \open{\varsy}\} \ta \tb \{\close{\varsx}\} \tc \{\close{\varsy}, \open{\varsz}\} \ta \{\close{\varsz}\}$, whereas $\varsy$ is not simple. We extend this definition to subword-marked languages in the obvious way, i.\,e., $\varsx$ is simple \emph{with respect to $L$} if $\varsx$ is simple with respect to every $w \in L$. \par

Our next goal is to show that for a subword-marked language $L$ over $\Sigma$ and $\varset$, and an $E = \{\mathcal{Z}_1, \mathcal{Z}_2, \ldots, \mathcal{Z}_k\} \subseteq \mathcal{P}(\varset)$ such that all variables from $\bigcup_{i \in [k]} \mathcal{Z}_i$ are simple with respect to $L$, we can construct a regular ref-language $L'$ such that $\llbracket L' \rrbracket = \varsigma_{\mathcal{Z}}^{=}(\llbracket L \rrbracket)$; see Theorem~\ref{stronglyNonOverlappingTheorem} below. We first sketch the intuitive idea of the construction and then give a formal proof. \par
For simplicitly, let us assume that $E = \{\mathcal{Z}\}$. Every $w \in L$ may contain some \emph{$\mathcal{Z}$-definitions}, which are definitions that define at least one variable $\varsz \in \mathcal{Z}$. Since all variables from $\mathcal{Z}$ are simple, such $\mathcal{Z}$-definitions are either \emph{non-empty} and of the form $\Gamma_1 u \Gamma_2$ with 
$\{\varsz \in \mathcal{Z} \mid \open{\varsz} \in \Gamma_1\} = \{\varsz \in \mathcal{Z} \mid \close{\varsz} \in \Gamma_2\} \neq \emptyset$, and $u \in \Sigma^+$, or they are \emph{empty}, which means that they are represented by a single symbol $\Gamma$ with $\{\open{\varsz}, \close{\varsz}\} \subseteq \Gamma$ for at least one $\varsz \in \mathcal{Z}$. The first step is to remove all subword-marked words that have both non-empty and empty $\mathcal{Z}$-definitions, since they describe span-tuples that are filtered out by the string-equality selection $\varsigma_{\mathcal{Z}}^{=}$. Moreover, this can be easily done by simply storing in the finite state control whether the first $\mathcal{Z}$-definition is empty or not, and then abort any computation that reads both types of $\mathcal{Z}$-definition.\par
The next task is to handle all other subword-marked words, i.\,e.,
those $w \in L$ in which some $\mathcal{Z}$-definitions assign to
variables from $\mathcal{Z}$ spans referring to different non-empty
strings (which we somehow have to get rid of, since they represent
span-tuples filtered out by the string-equality selection), and those
that assign the same non-empty string to all variables from
$\mathcal{Z}$. These properties cannot be explicitly checked by a
finite automaton. The idea of the construction is based on the
observation that for all subword-marked words $w \in L$ that assign to
every variable from $\mathcal{Z}$ spans referring to the same string,
we can as well only keep the very first $\mathcal{Z}$-definition
$\Gamma_1 u \Gamma_2$ of $w$, and replace every further $\mathcal{Z}$-definition $\Gamma'_1 u
\Gamma'_2$ in $w$ by $\Gamma'_1 \widehat{\varsz} \Gamma'_2$, where
$\widehat{\varsz}$ is some arbitrarily chosen variable defined by the
first $\mathcal{Z}$-definition. Such a modified $w'$ satisfies $w =
\deref{w'}$ and therefore also $\getSpanTuple{w} =
\getSpanTuple{\deref{w'}}$. Consequently, this modification yields a
ref-language $L'$ with $\llbracket L' \rrbracket =
\varsigma_{\mathcal{Z}}^{=}(\llbracket L \rrbracket)$. Obviously, we
have to show that $L'$ is regular. To this end, we can directly modify
the $\NFA$ $M$ for $L$ as follows. We simulate $M$ up to the point
where we read the first $\mathcal{Z}$-definition $\Gamma_1 u
\Gamma_2$. However, before reading $u$, we guess a state for every
variable from $\mathcal{Z}$ that might be defined by a later
$\mathcal{Z}$-definition. Then, while reading $u$, we compute some
states that are reachable from the guessed states by reading $u$
(since $M$ is nondeterministic, these target states are also
implicitly guessed by the specific path that is chosen). The idea is
of course that if the currently read input turns out to be of the kind
that is not filtered out by $\varsigma_{\mathcal{Z}}^{=}$, then it
must also be possible to guess exactly the state pairs between which
all the further occurrences of $u$ are read in the following
$\mathcal{Z}$-definitions of $w$ (after all, these state pairs must be
such that $u$ can be read between them, so it is possible to guess
them while processing the string $u$ of the first
$\mathcal{Z}$-definition). Then, in the rest of the computation,
whenever we encounter a $\mathcal{Z}$-definition $\Gamma'_1 u'
\Gamma'_2$, we simply verify if we have guessed valid state pairs
(i.\,e., whether the computation of $M$ also ends up in the states
guessed for the variables of this $\mathcal{Z}$-definition when we
encounter this $\mathcal{Z}$-definition), and then we simply read
$\Gamma'_1 \widehat{\varsz} \Gamma'_2$ instead of $\Gamma'_1 u'
\Gamma'_2$. Let us now give a formal proof.

\begin{thm}\label{stronglyNonOverlappingTheorem}
Let $L$ be a regular subword-marked language over $\Sigma$ and $\varset$, let $E = \{\mathcal{Z}_1, \mathcal{Z}_2, \ldots, \mathcal{Z}_k\} \subseteq \mathcal{P}(\varset)$ be such that all variables from $\bigcup_{i \in [k]} \mathcal{Z}_i$ are simple with respect to $L$. There is a regular ref-language $L'$ over $\Sigma$ and $\varset$ with $\llbracket L' \rrbracket = \varsigma_{E}^{=}(\llbracket L \rrbracket)$. Furthermore, $L'$ is reference-bounded.
\end{thm}

\begin{proof}
Let $\mathcal{Z} = \bigcup_{i \in [k]} \mathcal{Z}_i$. Since all
variables from $\mathcal{Z}$ are simple with respect to $L$, we know
that, for every $w \in \lang(M)$ and $i \in [k]$, all definitions of
variables from $\mathcal{Z}_i$ are either \emph{non-empty
  $\mathcal{Z}_i$-definitions} of the form $\Gamma_1 u \Gamma_2$ with
$\{\varsz \in \mathcal{Z} \mid \open{\varsz} \in \Gamma_1\} = \{\varsz \in \mathcal{Z} \mid \close{\varsz} \in \Gamma_2\} \neq \emptyset$, and $u \in \Sigma^+$, or they are \emph{empty
  $\mathcal{Z}_i$-definitions} represented by a single symbol $\Gamma$
with $\{\open{\varsz}, \close{\varsz}\} \subseteq \Gamma$ for at least
one $\varsz \in \mathcal{Z}_i$. \par 
We now make two assumptions. Firstly, we assume that $\mathcal{Z}_i \cap \mathcal{Z}_{i'} = \emptyset$ for all $i, i' \in [k]$ with $i \neq i'$. Secondly, we assume that, for every $w \in L$ and for every non-empty $\mathcal{Z}$-definition $\Gamma_1 u \Gamma_2$, there is some $i \in [k]$ such that $\{\varsz \in \mathcal{Z} \mid \open{\varsz} \in \Gamma_1\} \subseteq \mathcal{Z}_i$, and, likewise, for every empty $\mathcal{Z}$-definition $\Gamma_1$ in $w$, there is some $i \in [k]$ such that $\{\varsz \in \mathcal{Z} \mid \open{\varsz} \in \Gamma_1\} \subseteq \mathcal{Z}_i$. This means that every $\mathcal{Z}_i$-definition cannot be at the same time a $\mathcal{Z}_{i'}$-definition for some $i' \in [k]$ with $i \neq i'$. However, $|\{\varsz \in \mathcal{Z} \mid \open{\varsz} \in \Gamma_1\}| > 1$ is possible. These assumptions are obviously not without loss of generality, and we will explain later on how our construction can be adapted to the general case.\par 
Let $L_1$ be the subset of $L$ containing those $w \in L$ that, for every $i \in [k]$, have either only non-empty $\mathcal{Z}_i$-definitions or only empty $\mathcal{Z}_i$-definitions (i.\,e., not both types occur in the same subword-marked word). All $w \in L$ that, for some $i \in [k]$, have both empty and non-empty $\mathcal{Z}_i$-definitions, represent span-tuples that are filtered out by the string-equality selection. Thus, we have that $\varsigma_{E}^{=}(\llbracket L \rrbracket) = \varsigma_{E}^{=}(\llbracket L_1 \rrbracket)$. Moreover, it can be seen that $L_1$ is a regular subword-marked language. To this end, let $M$ be an $\NFA$ that represents $L$. 
We describe how $M$ can be changed into an $\NFA$ $M_1$ that accepts $L_1$. The $\NFA$ $M_1$ simulates $M$, but whenever we read the first $\mathcal{Z}_i$-definition for some $i \in [k]$, we store in the finite state control whether it is empty or non-empty. Then, we check whether all the following $\mathcal{Z}_i$-definitions are also empty or non-empty, respectively, and we stop the computation in a non-accepting state, if this is not the case. \par

For every $w \in L_1$, we now define a ref-word $\mu(w)$ as
follows. If, for every $i \in [k]$, $w$ contains at most one non-empty
$\mathcal{Z}_i$-definition (this includes the case that it contains
only empty $\mathcal{Z}_i$-definitions, or no
$\mathcal{Z}_i$-definition at all), then we set $\mu(w) = w$. Note
that in this case the span-tuple represented by $w$ is not filtered
out by $\varsigma_{E}^{=}$. If, for some $i \in [k]$, $w$ contains at least two non-empty $\mathcal{Z}_i$-definitions $\Gamma_1 u \Gamma_2$ and $\Gamma'_1 u' \Gamma'_2$ with $u \neq u'$, then $\mu(w) = \bot$ is undefined. Note that the string-equality selection $\varsigma_{E}^{=}$ filters out the span-tuple represented by such subword-marked words. If for a subword-marked word $w$ neither of these two cases apply, then the span-tuple represented by $w$ is not filtered out by $\varsigma_{E}^{=}$, and in this case we define $\mu(w)$ as follows. \par
For every $i \in [k]$ such that $w$ contains at least two non-empty
$\mathcal{Z}_i$-definitions, we do the following replacement. Let
$\Gamma_1 u \Gamma_2$ be the leftmost $\mathcal{Z}_i$-definition, and choose one fixed $\widehat{\varsz} \in \mathcal{Z}_i$ such that $\open{\widehat{\varsz}} \in \Gamma_1$. Then we replace every
$\mathcal{Z}_i$-definition $\Gamma'_1 u' \Gamma'_2$ that follows the
initial $\Gamma_1 u \Gamma_2$ by $\Gamma'_1 \widehat{\varsz}
\Gamma'_2$ (note that, by assumption, we have $u = u'$). This is
well-defined, since, by assumption, we have $\{\varsz \in \mathcal{Z}
\mid \open{\varsz} \in \Gamma_1\} \subseteq \mathcal{Z}_i$ for every
$\mathcal{Z}_i$-definition $\Gamma_1 u \Gamma_2$. The string obtained
by these replacements will be denoted by $\mu(w)$. It can be easily seen that, for every $w \in L_1$, if $\mu(w) \neq \bot$, then $\mu(w)$ is a ref-word with $w = \deref{\mu(w)}$. Finally, we define $\mu(L_1) = \{\mu(w) \mid w \in L_1, \mu(w) \neq \bot\}$. \par

We next observe that $\llbracket \mu(L_1) \rrbracket = \varsigma_{E}^{=}(\llbracket L_1 \rrbracket)$. To this end, let $t \in (\varsigma_{E}^{=}(\llbracket L_1 \rrbracket))(v)$ for some $v \in \Sigma^*$. This means that there is some $w \in L_1$ with $v = \getWord{w}$ and $t = \getSpanTuple{w}$. Moreover, due to the string-equality selection $\varsigma_{E}^{=}$, we know that $\mu(w) \neq \bot$, and therefore $\deref{\mu(w)} = w$, which also implies that $\getWord{\deref{\mu(w)}} = \getWord{w} = v$ and $\getSpanTuple{\deref{\mu(w)}} = \getSpanTuple{w} = t$. Consequently, $t \in \llbracket \mu(L_1) \rrbracket(v)$.\par
For the other direction, let $t \in \llbracket \mu(L_1) \rrbracket(v)$ for some $v \in \Sigma^*$. This means that there is some $w \in \mu(L_1)$ with $v = \getWord{\deref{w}}$ and $t = \getSpanTuple{\deref{w}}$. By definition of $\mu(L_1)$, we have $w = \mu(w')$ for some $w' \in L_1$ with $\mu(w') \neq \bot$, which also means that $w' = \deref{\mu(w')} = \deref{w}$. Hence, $\getWord{\deref{w}} = \getWord{w'} = v$ and $\getSpanTuple{\deref{w}} = \getSpanTuple{w'} = t$. Consequently, $t \in \llbracket L_1 \rrbracket(v)$. Moreover, since $\mu(w') \neq \bot$, we know that all (if any) $\mathcal{Z}$-definitions of $w'$ define the same factor, which implies that $t \in (\varsigma_{E}^{=}(\llbracket L_1 \rrbracket))(v)$. \par

Next, we have to show that $\mu(L_1)$ is regular. Let $M$ be an $\NFA$
that represents $L_1$ with a set $Q$ of states. We describe how $M$ can be changed into an $\NFA$ $\mu(M)$ that accepts $\mu(L_1)$. The $\NFA$ $\mu(M)$ has the same states as $M$, but in each state it also stores, for every $\varsz \in \mathcal{Z}$, a pair $(s_{\varsz, \openmark}, s_{\varsz, \closemark}) \in (Q \cup \{\bot\})^2$. Initially, we have $(s_{\varsz, \openmark}, s_{\varsz, \closemark}) = (\bot, \bot)$ for every $\varsz \in \mathcal{Z}$. \par
The $\NFA$ $\mu(M)$ simulates $M$, but every encountered non-empty $\mathcal{Z}$-definition is processed in one of two possible ways, depending on whether or not it is the \emph{first} $\mathcal{Z}_i$-definition for some $i \in [k]$.\par
Let us assume that for some $i \in [k]$ we reach the first non-empty $\mathcal{Z}_i$-definition $\Gamma_1 u \Gamma_2$. Moreover, let $\widehat{\mathcal{Z}}_i$ be the variables from $\mathcal{Z}$ defined by this definition, i.\,e., $\widehat{\mathcal{Z}}_i = \{\varsz \in \mathcal{Z}_i \mid \open{\varsz} \in \Gamma_1\}$ (it is helpful to note that $\mathcal{Z}_i \setminus \widehat{\mathcal{Z}}_i$ is therefore exactly the set of variables that might be defined by any other $\mathcal{Z}_i$-definition to appear later in the input). Let $\widehat{\varsz}_i \in \widehat{\mathcal{Z}}_i$ be arbitrarily chosen (this variable $\widehat{\varsz}_i$ will play a central role later on, when we describe how $\mathcal{Z}_i$-definitions are processed that are not \emph{first} $\mathcal{Z}_i$-definitions). 
Let us further assume that after reading $\Gamma_1$, we are in some state $s$. 
Now, $\mu(M)$ sets every $s_{\varsz, \openmark}$ with $\varsz \in
\mathcal{Z}_i \setminus \widehat{\mathcal{Z}}_i$ to some
nondeterministically guessed state. It then reads $u$ simultaneously
from the states $s$ and all the guessed states $s_{\varsz, \openmark}$
with $\varsz \in \mathcal{Z}_i \setminus \widehat{\mathcal{Z}}_i$. The states that are reached from the guessed states $s_{\varsz, \openmark}$ with $\varsz \in \mathcal{Z}_i \setminus \widehat{\mathcal{Z}}_i$ by reading $u$ are stored in the components $s_{\varsz, \closemark}$ for all $\varsz \in \mathcal{Z}_i \setminus \widehat{\mathcal{Z}}_i$ (or $s_{\varsz, \closemark} = \undefin$ if no state can be reached from $s_{\varsz, \openmark}$ by reading $u$). This means that after reading $u$, $\mu(M)$ is in some state $t$ (the state reached from $s$ by reading $u$) and the additionally stored pairs $(s_{\varsz, \openmark}, s_{\varsz, \closemark})$ are such that, for every $\varsz \in \mathcal{Z}_i \setminus \widehat{\mathcal{Z}}_i$, reading $u$ can change $M$ from state $s_{\varsz, \openmark}$ to $s_{\varsz, \closemark}$ (unless $s_{\varsz, \closemark} = \undefin$). In particular, note that the start states $s_{\varsz, \openmark}$ are non-deterministically guessed, while the states  $s_{\varsz, \closemark}$ (unless we have $s_{\varsz, \closemark} = \undefin$) are also non-deterministically determined, since there might be several states rechable from $s_{\varsz, \openmark}$ by reading $u$. \par
Whenever we encounter a non-empty $\mathcal{Z}_i$-definition $\Gamma'_1 u' \Gamma'_2$, for some $i \in [k]$, that is \emph{not} a first $\mathcal{Z}_i$-definition, then we proceed as follows. We let again $\widehat{\mathcal{Z}}'_i$ be the variables from $\mathcal{Z}$ defined by this definition. By the way we have processed the first $\mathcal{Z}_i$-definition $\Gamma_1 u \Gamma_2$, we know that for every $\varsz' \in \widehat{\mathcal{Z}}'_i$, the pair $(s_{\varsz', \openmark}, s_{\varsz', \closemark})$ is such that we can reach $s_{\varsz', \closemark}$ from $s_{\varsz', \openmark}$ by reading $u$ (or $s_{\varsz', \closemark} = \undefin$). Moreover, recall that $\widehat{\varsz}_i$ is an arbitrary variable defined by this first $\mathcal{Z}_i$-definition $\Gamma_1 u \Gamma_2$.
Let $s'$ be the state reached by $M$ after having read $\Gamma'_1$. 
We check whether $s_{\varsz', \openmark} = s'$ for every $\varsz' \in \widehat{\mathcal{Z}}'_i$, and whether $s_{\varsz', \closemark} = t'$ for every $\varsz' \in \widehat{\mathcal{Z}}'_i$ and the same state $t'$, and, if this is the case, we read the variable reference $\widehat{\varsz}_i$ and move to state $t'$. \par
This concludes the definition of $\mu(M)$. We will show next the correctness of the construction, i.\,e., $\lang(\mu(M)) = \mu(L_1)$.\par
Let $w$ be accepted by $M$ and assume that $\mu(w) \neq \bot$, i.\,e., $\mu(w) \in \mu(L_1)$. If, for every $i \in [k]$, $w$ has at most one non-empty $\mathcal{Z}_i$-definition, then $\mu(M)$ accepts $\mu(w) = w$ (by the same run as $M$ accepts $w$). So let us assume that $i_1, i_2, \ldots, i_{p} \in [k]$ are such that $w$ has at least two non-empty $\mathcal{Z}_{i_j}$-definitions for $j \in [p]$. Since $\mu(w) \neq \bot$, we know that, for every $j \in [p]$, all $\mathcal{Z}_{i_j}$-definitions of $w$ define the same factor $u_{j}$, i.\,e., every non-empty $\mathcal{Z}_{i_j}$-definition has the form $\Gamma_1 u_{j} \Gamma_2$. We have to show that $\mu(M)$ accepts $\mu(w)$. \par
We assume that, for every $j \in [p]$, $\Gamma^{j}_1 u_{j} \Gamma^j_2$ is the first $\mathcal{Z}_{i_j}$-definition of $w$, and let $\widehat{\mathcal{Z}}_{i_j}$ be the variables from $\mathcal{Z}_{i_j}$ defined by this $\mathcal{Z}_{i_j}$-definition $\Gamma^{j}_1 u_{j} \Gamma^j_2$.
For every $\varsz \in \mathcal{Z}_{i_j} \setminus
\widehat{\mathcal{Z}}_{i_j}$ that is defined by some
$\mathcal{Z}_{i_j}$-definition in $w$, let $\widetilde{s}_{\varsz,
  \openmark}$ and $\widetilde{s}_{\varsz, \closemark}$ be the states
between which $M$ reads the factor $u_j$ in the
$\mathcal{Z}_{i_j}$-definition for $\varsz$. The $\NFA$ $\mu(M)$ can
now simulate $M$ until it encounters the first
$\mathcal{Z}_{i_j}$-definition $\Gamma^j_1 u_j \Gamma^j_2$. Then,
after reading $\Gamma_1^j$, $\mu(M)$ can guess for every $\varsz \in \mathcal{Z}_{i_j} \setminus
\widehat{\mathcal{Z}}_{i_j}$ the states $\widetilde{s}_{\varsz,
  \openmark}$ and store them in the components $s_{\varsz,
  \openmark}$, if $w$ contains some $\mathcal{Z}_{i_j}$-definition
that defines $z$, and it can guess some arbitrary state
otherwise. Then, it can read $u_j$ in such a way that exactly the
states $\widetilde{s}_{\varsz, \closemark}$ are stored in the
components $s_{\varsz, \closemark}$, if $w$ contains some
$\mathcal{Z}_{i_j}$-definition that defines $z$, and
$\widetilde{s}_{\varsz, \closemark} = \undefin$ otherwise. This must
be possible by definition of these states $\widetilde{s}_{\varsz,
  \openmark}$ and $\widetilde{s}_{\varsz, \closemark}$, and by
$\mu(M)$'s definition. The $\NFA$ $\mu(M)$ can now proceed with
simulating $M$, but whenever it encounters another
$\mathcal{Z}_{i_j}$-definition $\Gamma'^j_1 u_j \Gamma'^j_2$, due to
the stored states, it will read $\Gamma'^j_1 \widehat{\varsz}_{j}
\Gamma'^j_2$ instead of $\Gamma'^j_1 u_j \Gamma'^j_2$ (recall that $\widehat{\varsz}_i$ is the fixed variable from $\widehat{\mathcal{Z}}_{i_j}$ selected earlier). Hence, it accepts the word $\mu(w)$.\par

On the other hand, let some $w'$ be accepted by $\mu(M)$. If, for every $i \in [k]$, $w'$ has at most one non-empty $\mathcal{Z}_i$-definition, then $M$ accepts $w'$, and $\mu(w') = w'$, i.\,e., $w' \in \mu(L_1)$. So let us assume that $i_1, i_2, \ldots, i_{p} \in [k]$ are such that $w'$ has at least two non-empty $\mathcal{Z}_{i_j}$-definitions for $j \in [p]$. For every $j \in [p]$, the first $\mathcal{Z}_{i_j}$-definition has the form $\Gamma^j_1 u_j \Gamma^j_2$, and all other $\mathcal{Z}_{i_j}$-definitions have the form $\Gamma'^j_1 \widehat{\varsz}_j \Gamma'^j_2$, where $\widehat{\varsz}_j$ is from the set $\widehat{\mathcal{Z}}_{i_j}$ of those variables from $\mathcal{Z}_{i_j}$ that are defined by the first $\mathcal{Z}_{i_j}$-definition $\Gamma^{j}_1 u_{j} \Gamma^j_2$.
Moreover, for every such $\mathcal{Z}_{i_j}$-definition $\Gamma'^j_1 \widehat{\varsz}_j \Gamma'^j_2$ where $\widehat{\varsz}_j$ is read by going from some state $s$ to some state $t$, we know that $M$ can read $u_j$ by going from state $s$ to $t$ (this follows from the definition $\mu(M)$). This directly implies that the word $w$ obtained from $w'$ by replacing each $\Gamma'^j_1 \widehat{\varsz}_j \Gamma'^j_2$ by $\Gamma'^j_1 u_j \Gamma'^j_2$ satisfies $\mu(w) = w'$ and $w \in \lang(M) = L_1$. Thus, $w' \in \mu(L_1)$.

We conclude that we have shown that $\varsigma_{E}^{=}(\llbracket L \rrbracket) = \varsigma_{E}^{=}(\llbracket L_1 \rrbracket)$, that
$\llbracket \mu(L_1) \rrbracket = \varsigma_{E}^{=}(\llbracket L_1 \rrbracket)$, and that $\mu(L_1)$ is a regular ref-language over $\Sigma$ and $\varset$. Thus, $L' := \mu(L_1)$ is a regular ref-language over $\Sigma$ and $\varset$ with $\llbracket L' \rrbracket = \varsigma_{E}^{=}(\llbracket L \rrbracket)$. Moreover, $L'$ is obviously reference-bounded. \par

In order to conclude this proof, we have to show how to adapt the construction to the general case, i.\,e., the case where, for some $i, i' \in [k]$ with $i \neq i'$, it is possible that $\mathcal{Z}_i \cap \mathcal{Z}_j \neq \emptyset$, and $w \in L$ might contain $\mathcal{Z}$-definitions that are both $\mathcal{Z}_i$- and $\mathcal{Z}_{i'}$-definitions. 

For a fixed $w \in L$, we can define an equivalence relation over
$\mathcal{Z}$ as follows. We set $\varsx \sim'_w \varsz$ if $\varsx,
\varsz \in \mathcal{Z}_i$ for some $i \in [k]$, or if $w$ has a
non-empty $\mathcal{Z}$-definition $\Gamma_1 u \Gamma_2$ or an empty
$\mathcal{Z}$-definition $\Gamma_1$ such that $\{\open{\varsx},
\open{\varsz}\} \subseteq \Gamma_1$, and we let $\sim_w$ be the transitive closure of $\sim'_w$. Let $E_{\sim_w} = \{\mathcal{Z}^{\sim_w}_1, \mathcal{Z}^{\sim_w}_2, \ldots, \mathcal{Z}^{\sim_w}_{k_{\sim_w}}\}$ be the set of the equivalence classes of $\sim_w$. We note that, for every $w \in L$, we have that $E_{\sim_w}$ satisfies the two properties from above: (1) $\mathcal{Z}^{\sim_w}_i \cap \mathcal{Z}^{\sim_w}_{i'} = \emptyset$ for all $i, i' \in [k_{\sim_w}]$ with $i \neq i'$, and (2), for every $w \in L$ and for every non-empty $\mathcal{Z}$-definition $\Gamma_1 u \Gamma_2$ and every empty $\mathcal{Z}$-definition $\Gamma_1$ in $w$, there is some $i \in [k_{\sim_w}]$ such that $\{\varsz \in \mathcal{Z} \mid \open{\varsz} \in \Gamma_1\} \subseteq \mathcal{Z}^{\sim_w}_i$. \par
For the following, it is a crucial observation that, for every $w \in L$, the tuple described by $w$ is filtered out by $\varsigma_{E}^{=}$ if and only if it is filtered out by $\varsigma_{E_{\sim_w}}^{=}$, i.\,e., $\varsigma_{E}^{=}(\{\getSpanTuple{w}\}) = \varsigma_{E_{\sim_w}}^{=}(\{\getSpanTuple{w}\})$. Indeed, if the tuple described by $w$ is filtered out by $\varsigma_{E}^{=}$, then it must be filtered out by $\varsigma_{E_{\sim_w}}^{=}$ as well, since $E$ is a refinement of $E_{\sim_w}$. On the other hand, if the tuple described by $w$ satisfies $\varsigma_{E}^{=}$, then, by definition of $\sim_w$, it must also satisfy all the string-equality selections described by $E_{\sim_w}$.

For every $w \in L$, we have defined above the condition for $w$ to be in $L_1$, and we have defined the ref-word $\mu(w)$. For every fixed $w \in L$, these definitions depend only on $w$ and the sets $\mathcal{Z}_1, \mathcal{Z}_2, \ldots, \mathcal{Z}_k$. However, we had to require that these sets satisfied our two conditions: (1) $\mathcal{Z}_i \cap \mathcal{Z}_{i'} = \emptyset$ for every $i, i' \in [k]$ with $i \neq i'$, and (2), for every $w \in L$ and for every non-empty $\mathcal{Z}$-definition $\Gamma_1 u \Gamma_2$ and every empty $\mathcal{Z}$-definition $\Gamma_1$ in $w$, there is some $i \in [k]$ such that $\{\varsz \in \mathcal{Z} \mid \open{\varsz} \in \Gamma_1\} \subseteq \mathcal{Z}_i$. Since the sets from $E_{\sim_w}$ also satisfy these conditions, we can also define, for every $w \in L$, the condition for $w$ to be in $L_1$, and the ref-word $\mu(w)$, but with respect to the sets $\mathcal{Z}^{\sim_w}_1, \mathcal{Z}^{\sim_w}_2, \ldots, \mathcal{Z}^{\sim_w}_{k_{\sim_w}}$.\par

It is not hard to see that for this adapted version of $L_1$, we also have that $\varsigma_{E}^{=}(\llbracket L \rrbracket) = \varsigma_{E}^{=}(\llbracket L_1 \rrbracket)$. Indeed, every word $w \in L_1$ that, for some $i \in [k_{\sim_w}]$, contains both a non-empty $\mathcal{Z}^{\sim_w}_i$-definition and an empty $\mathcal{Z}^{\sim_w}_i$-definition must be filtered out by $\varsigma_{E}^{=}$. \par

Moreover, the adapted version of $\mu$ also satisfies that, if $w \in L_1$ with $\mu(w) \neq \bot$, then $\mu(w)$ is a ref-word with $w = \deref{\mu(w)}$. In particular, this also means that $\llbracket \mu(L_1) \rrbracket = \varsigma_{E}^{=}(\llbracket L_1 \rrbracket)$.\par

Consequently, also in the general case, we have that $\varsigma_{E}^{=}(\llbracket L \rrbracket) = \llbracket \mu(L_1) \rrbracket$. However, it remains to show that with the adapted definitions of $L_1$ and $\mu(w)$, $L_1$ and $\mu(L_1)$ are still regular. In the constructions of the $\NFA$s for $L_1$ and $\mu(L_1)$, we have explained how the computation should proceed on a fixed input $w$, and in dependency of the sets $\mathcal{Z}_1, \mathcal{Z}_2, \ldots, \mathcal{Z}_k$. Consequently, for a fixed input $w$, we can carry out the constructions also for the sets $\mathcal{Z}^{\sim_w}_1, \mathcal{Z}^{\sim_w}_2, \ldots, \mathcal{Z}^{\sim_w}_{k_{\sim_w}}$, as long as the automaton has knowledge of these sets, which, after all, depend on the input. Hence, we let the $\NFA$s simply initially guess the equivalence relation $\sim_{w}$, and then we will verify this guess while processing the input $w$, and we will abort the computation if the guess of $\sim_{w}$ has been incorrect. In this way, we can assume that for any input $w$ the sets $\mathcal{Z}^{\sim_w}_1, \mathcal{Z}^{\sim_w}_2, \ldots, \mathcal{Z}^{\sim_w}_{k_{\sim_w}}$ are known to the $\NFA$, which means that we can carry out the $\NFA$ constructions described above, just with respect to the sets $\mathcal{Z}^{\sim_w}_1, \mathcal{Z}^{\sim_w}_2, \ldots, \mathcal{Z}^{\sim_w}_{k_{\sim_w}}$ instead of the set $\mathcal{Z}_1, \mathcal{Z}_2, \ldots, \mathcal{Z}_{k}$.\par
This completes the proof of Theorem~\ref{stronglyNonOverlappingTheorem}.
\end{proof}

\subsubsection{The Span Fusion Operation}\label{sec:spanFusion}

The \emph{span-fusion}\label{spanfusionDef} $\spanfusion$ is a binary
operation $\spans \times \spans \to \spans$ defined by $\spann{i}{j}
\spanfusion \spann{i'}{j'} = \spann{\min\{i, i'\}}{\max\{j, j'\}}$ and
$\spann{i}{j} \spanfusion \undefin = \undefin \spanfusion \spann{i}{j}
= \spann{i}{j}$.
Intuitively speaking, the operation $\spanfusion$
constructs the set-union of two spans and fills in the gaps to turn it
into a valid span.
Note that $\spanfusion$ is obviously associative.
For a set $K \subseteq \spans(w)$, we define
$\bigspanfusion(K) = \undefin$ if $K = \emptyset$ and
$\bigspanfusion(K) = \bigspanfusion(K\setminus\{s\}) \spanfusion s$ if
$s \in K$.
\par 
We next lift this operation to an operation on spanners with the following intended meaning. In a table $S(w)$ for some spanner $S$ over $\Sigma$ and $\varset$, and $w \in \Sigma^*$, we want to replace a specified set of columns $\{\varsy_1, \varsy_2, \ldots, \varsy_k\} \subseteq \varset$ by a single new column $\varsx$ that, for each row $t$ (i.\,e., span-tuple $t$) in $S(w)$, contains the span $\bigspanfusion(\{t(\varsy_i) \mid i \in [k]\})$.

\begin{defi}
Let $\lambda \subseteq \varset$ and let $\varsx$ be a new variable
with $\varsx \notin \varset \setminus \lambda$. For any
$\varset$-tuple $t$, $\bigspanfusion_{\lambda \to \varsx}(t)$ is the
$((\varset \setminus \lambda) \cup \{\varsx\})$-tuple with
$(\bigspanfusion_{\lambda \to \varsx}(t))(\varsx) =
\bigspanfusion(\{t(\varsy) \mid \varsy \in \lambda\})$ and
$(\bigspanfusion_{\lambda \to \varsx}(t))(\varsz) = t(\varsz)$ for
every $\varsz \in \varset \setminus \lambda$. For a set $R$ of
$\varset$-tuples, $\bigspanfusion_{\lambda \to \varsx}(R) =
\{\bigspanfusion_{\lambda \to \varsx}(t) \mid t \in R\}$. Moreover,
for a spanner $S$ over $\Sigma$ and $\varset$, the spanner
$\bigspanfusion_{\lambda \to \varsx}(S)$ over $\Sigma$ and $((\varset \setminus \lambda)\cup \{\varsx\})$ is defined by $(\bigspanfusion_{\lambda \to \varsx}(S))(w) = \bigspanfusion_{\lambda \to \varsx}(S(w))$ for every word $w$.
\end{defi}

We use the following generalised application of the operation $\bigspanfusion_{\lambda \to \varsx}$. For $\Lambda = \{\lambda_1, \lambda_2, \ldots, \allowbreak \lambda_k\} \subseteq \mathcal{P}(\varset)$ such that all $\lambda_i$ with $i \in [k]$ are pairwise disjoint, a spanner $S$ over $\varset$ and fresh variables $\varsx_1, \varsx_2, \ldots, \varsx_k$, we define $\bigspanfusion_{\{\lambda_i \to \varsx_i \mid i \in [k]\}}(S) = \bigspanfusion_{\lambda_1 \to \varsx_1}(\bigspanfusion_{\lambda_2\to \varsx_2}(\ldots \bigspanfusion_{\lambda_k\to \varsx_k}(S)\ldots))$. If the new variables $\varsx_i$ are clear from the context, we also write $\bigspanfusion_{\lambda}$ or $\bigspanfusion_{\{\lambda_i \mid i \in [k]\}}$ instead of $\bigspanfusion_{\lambda \to \varsx}$ or $\bigspanfusion_{\{\lambda_i \to \varsx_i \mid i \in [k]\}}$, respectively.

\begin{exa}
Let $L = \lang(\open{\varsx_1} \ta^* \open{\varsx_2} \tb^* \close{\varsx_1} \ta^* \close{\varsx_2})$ be a subword-marked language over $\{\ta, \tb\}$ and $\{\varsx_1, \varsx_2\}$. For $w = \ta \ta \tb \ta \ta \ta$, we have $\llbracket L \rrbracket(w) = \{(\spann{1}{4}, \spann{3}{7})\}$. 
Moreover, let $L' = \lang(\open{\varsy_1} \ta^* \close{\varsy_1} \open{\varsy_2} \open{\varsy_3} \tb^* \close{\varsy_3} \close{\varsy_2} \open{\varsy_4} \ta^* \close{\varsy_4})$ be a subword-marked language over $\{\ta, \tb\}$ and $\{\varsy_1, \varsy_2, \varsy_3, \varsy_4\}$, and let $\Lambda = \{\{\varsy_1, \varsy_2\} \to \varsx_1, \{\varsy_3, \varsy_4\} \to \varsx_2\}$. Then 
\begin{equation*}
\bigspanfusion_{\Lambda}\llbracket L' \rrbracket(w) = \{\bigspanfusion_{\Lambda}(\spann{1}{3}, \spann{3}{4}, \spann{3}{4}, \spann{4}{7})\} = \{(\spann{1}{4}, \spann{3}{7})\} = \llbracket L \rrbracket(w)
\end{equation*}
In fact, it can be easily verified that $\bigspanfusion_{\Lambda}\llbracket L' \rrbracket = \llbracket L \rrbracket$.  
\end{exa}

\subsubsection{The Split Operation}\label{sec:splitOp}

Intuitively speaking, a split $t'$ of a span-tuple $t$ is any
span-tuple obtained from $t$ by splitting each span $t(\varsx)$ into
several spans, e.\,g., splitting $\spann{3}{24}$ into spans
$\spann{3}{5}, \spann{5}{17}, \spann{17}{17}, \spann{17}{24}$. Any
such split $t'$ can be easily translated back into $t$ by applications
of the span fusion, i.\,e., $\bigspanfusion_{\{\lambda_{\varsx} \to
  \varsx \mid \varsx \in \varset\}}(t') = t$, where the sets
$\lambda_{\varsx}$ are the variables used for the split version of
$t(\varsx)$. There are two important aspects of this
operation. Firstly, if we allow large enough splits, i.\,e., if we
choose the sets $\lambda_{\varsx}$ large enough, then we can always
split in such a way that all variables are non-overlapping. Secondly,
the split of a regular spanner (i.\,e., the spanner that extracts all splits of the original span-tuples) is also a regular spanner. Let us next define this formally.\par
Let us first recall the definition of non-overlapping variables, which
has already been introduced
in Section~\ref{sec:prelim}. Let $L$ be a subword-marked
language over $\Sigma$ and $\varset$, let $w \in L$ and $t =
\getSpanTuple{w}$. We say that variables $\varsx, \varsy \in \varset$
are \emph{non-overlapping} (\emph{with respect to $w$} (\emph{or
  $t$})) if $t(\varsx) = t(\varsy)$ or $t(\varsx)$ and $t(\varsy)$ are
disjoint (i.\,e., $t(\varsx) = \spann{\ell_{\varsx}}{r_{\varsx}}$ and
$t(\varsy) = \spann{\ell_{\varsy}}{r_{\varsy}}$ with $r_{\varsx} \leq
\ell_{\varsy}$ or $r_{\varsy} \leq \ell_{\varsx}$). The variables
$\varsx$ and $\varsy$ are non-overlapping with respect to the
subword-marked language $L$, if $\varsx$ and $\varsy$ are
non-overlapping with respect to every $w \in L$.\par 
For any set $\varset$ of variables, we define the \emph{extended variable set} $\extVarset{\varset}$ for $\varset$ by $\extVarset{\varset} := \bigcup_{\varsx \in \varset} \{\varsx^1, \varsx^2, \ldots, \varsx^{\varno}\}$, where $\varno := 4|\varset|^2 - 1$ (the choice of the number $\varno$ is crucial and will become clear later). Intuitively, for every $\varsx \in \varset$, the extended variable set for $\varset$ contains $\varno$ new variables, e.\,g., $\extVarset{\{\varsx, \varsy, \varsz\}} = \{\varsx^1, \ldots, \varsx^{37}, \varsy^1, \ldots, \varsy^{37}, \varsz^1, \ldots, \varsz^{37}\}$. With respect to $\extVarset{\varset}$, we can also make the following observation.

\begin{obs}\label{extVarsetSizeObs}
For every set $\varset$ of variables, we have that $|\extVarset{\varset}| = \bigO(|\varset|^3)$.
\end{obs}

Let $t$ be an $\varset$-tuple. Any $\extVarset{\varset}$-tuple $t'$ is called a \emph{split} of $t$ if it satisfies the following properties:
\begin{itemize}
\item For every $\varsx \in \varset$, 
\begin{itemize}
\item if $t(\varsx) = \undefin$, then $t'(\varsx^1) = t'(\varsx^2) = \ldots = t'(\varsx^{\varno}) = \undefin$,
\item if $t(\varsx) = \spann{k}{\ell}$, then, for every $i \in [\varno]$, $t'(\varsx^i) = \spann{k_i}{\ell_i}$ such that $k_1 = k$, $\ell_{\varno} = \ell$, and, for every $i \in [\varno-1]$, $\ell_i = k_{i+1}$.
\end{itemize}
\item Any two variables of $\extVarset{\varset}$ are non-overlapping with respect to $t$.
\end{itemize}
For example, let $t$ be an $\{\varsx, \varsy\}$-tuple with $t(\varsx)
= \spann{3}{24}$ and $t(\varsy) = \spann{7}{15}$. Then a possible
split of $t$ would be the $\extVarset{\varset}$-tuple $t'$ with
$t'(\varsx^1) = \spann{3}{7}, t'(\varsx^2) = \spann{7}{12},
t'(\varsx^3) = \spann{12}{15}, t'(\varsx^4) = \spann{15}{24}$ and $t'(\varsx^j) = \spann{24}{24}$ for
every $j \in \{5, 6, \ldots, 17\}$, and $t'(\varsy^1) = \spann{7}{12},
t'(\varsy^2) = \spann{12}{12}, t'(\varsy^3) = \spann{12}{15}$ and $t'(\varsy^j) = \spann{15}{15}$ for every $j \in
\{4, 5 \ldots, 17\}$.\par
For every $\varsx \in \varset$, we define $\lambda^{\varno}_{\varsx} = \{\varsx^1, \varsx^2, \ldots, \varsx^{\varno}\}$. Hence, $\extVarset{\varset} = \bigcup_{\varsx \in \varset} \lambda^{\varno}_{\varsx}$. By definition, every split $t'$ of $t$ satisfies $\bigspanfusion_{\{\lambda^{\varno}_\varsx \to \varsx \mid \varsx \in \varset\}}(t') = t$. \par
We can define splits analogously for subword-marked words. Let $L$ be a subword-marked language over $\Sigma$ and $\varset$. For any $w \in L$, we say that a subword-marked word $w'$ over $\Sigma$ and $\extVarset{\varset}$ is a \emph{split} of $w$, if $\getWord{w} = \getWord{w'}$ and $\getSpanTuple{w'}$ is a split of $\getSpanTuple{w}$. \par  
For an $\varset$-tuple $t$ (or subword-marked word $w$), let $\splitset{\varset}{t}$ ($\splitset{\varset}{w}$, respectively) be the set of all splits of $t$ (splits of $w$, respectively). 
For any subword-marked language over $\Sigma$ and $\varset$, let $\splitset{\varset}{L} = \{\splitset{\varset}{w} \mid w \in L\}$, and let $\splitset{\varset}{\llbracket L \rrbracket} = \llbracket \splitset{\varset}{L} \rrbracket$.\par
We will next see that the split of any regular spanner is also a
regular spanner. We will show that any $\NFA$ $M$ for a
subword-marked language $L$ can be modified such that it accepts $\splitset{\varset}{L}$. To this end, we simply non-deterministically use the markers $\open{\varsx^1}, \close{\varsx^1}, \open{\varsx^2}, \close{\varsx^2}, \ldots, \open{\varsx^{\varno}}, \close{\varsx^{\varno}}$ for processing the part of the input that by $M$ is read between markers $\open{\varsx}$ and $\close{\varsx}$. Moreover, to ensure that each two variables are non-overlapping, we have to make sure that all these non-deterministic split points agree with each other in the sense that whenever some split point is created by reading $\close{\varsx^j}$ and $\open{\varsx^{j+1}}$, then we also must create a split-point for all other variables by reading $\close{\varsy^{j'}}$ and $\open{\varsy^{j' + 1}}$ for the correct $j' \in [\varno]$. This construction is formally carried out in the proof of the next lemma.

\begin{lem}
Let $L$ be a regular subword-marked language over $\Sigma$ and $\varset$. Then $\splitset{\varset}{L}$ is a regular subword-marked language over $\Sigma$ and $\extVarset{\varset}$.
\end{lem}

\begin{proof}
Let $M$ be an $\NFA$ for $L$. To transform $M$ into an $\NFA$ for $\splitset{\varset}{L}$, we perform two transformation steps. In general, we use the following terminology. Reading a symbol $\Gamma \subseteq \Gamma_{\varset}$ with $\open{\varsx} \in \Gamma$ is called \emph{opening} variable $\varsx$, and reading a symbol $\Gamma \subseteq \Gamma_{\varset}$ with $\close{\varsx} \in \Gamma$ is called \emph{closing} variable $\varsx$. Moreover, at any step of a computation, we say that variable $\varsx$ is \emph{open}, if we have already read a symbol $\Gamma \subseteq \Gamma_{\varset}$ with $\open{\varsx} \in \Gamma$, but we have not yet read a symbol $\Gamma' \subseteq \Gamma_{\varset}$ with $\close{\varsx} \in \Gamma$. We assume that any $\NFA$ always stores in its finite state control which variables are currently open (this information can be maintained during any computation).

\medskip

\noindent\textbf{Step 1} (Changing $M$ into $M'$): The $\NFA$ $M'$ simulates the computation of $M$ on $w$, but with the following differences. Whenever $M$ opens some variable $\varsx \in \varset$, then $M'$ opens variable $\varsx^1$. At any point in the computation, $M'$ can nondeterministically choose some currently open variable $\varsx^j$ with $j < \varno$, and then close $\varsx^j$ and open $\varsx^{j+1}$ (observe that closing $\varsx^{\varno}$ is not included here). Whenever $M$ closes some variable $\varsx \in \varset$, then $M'$ closes the currently open variable $\varsx^j$ (which may be $\varsx^{\varno}$) and then opens and closes all remaining variables $\varsx^{j+1}, \ldots, \varsx^{\varno}$. \par 
Let us say that a subword-marked word $w'$ over $\Sigma$ and $\extVarset{\varset}$ is a \emph{pseudo-split} of a subword marked word $w$ over $\Sigma$ and $\varset$, if the first property of the definition of splits is satisfied:
\begin{itemize}
\item if $\getSpanTuple{w}(\varsx) = \undefin$, then $\getSpanTuple{w'}(\varsx^1) = \getSpanTuple{w'}(\varsx^2) = \ldots = \getSpanTuple{w'}(\varsx^{\varno}) = \undefin$,
\item if $\getSpanTuple{w}(\varsx) = \spann{k}{\ell}$, then, for every $i \in [\varno]$, $\getSpanTuple{w'}(\varsx^i) = \spann{k_i}{\ell_i}$ such that $k_1 = k$, $\ell_{\varno} = \ell$, and, for every $i \in [\varno-1]$, $\ell_i = k_{i+1}$.
\end{itemize}
It can be easily seen that $M'$ accepts exactly all subword-marked words over $\Sigma$ and $\extVarset{\varset}$ that are pseudo-splits of subword-marked words accepted by $M$.\par
We now have to filter out those pseudo-splits which are not valid splits, i.\,e., those that do not satisfy the second property of the definition of splits: any two variables of $\extVarset{\varset}$ are non-overlapping with respect to $\getSpanTuple{w'}$.

\medskip

\noindent\textbf{Step 2} (Changing $M'$ into $M''$): The $\NFA$ $M''$ simulates the computation of $M'$ on $w$, but it can interrupt certain computations and reject the input as follows. Let us assume that $M'$ reads a symbol $\Gamma \subseteq \Gamma_{\varset}$, which, since $\Gamma$ is non-empty, means that some variable is either opened or closed at this step. Then $M'$ checks whether there is some variable $\varsy$ that is open right before $M'$ reads $\Gamma$ and $\close{\varsy} \notin \Gamma$ and, if this is the case, $M'$ interrupts the computation and rejects. \par
Obviously, if some $w' \in \lang(M')$ is a valid split of some $w \in \lang(M)$, then upon reading $w'$ the situation that causes $M''$ to interrupt will never occur, which means that $w' \in \lang(M'')$. On the other hand, if some $w' \in \lang(M')$ is only a pseudo-split of some $w \in \lang(M)$, but not a valid split, then there are two variables in $\extVarset{\varset}$ that are overlapping with respect to $w'$, which means that any computation of $M''$ on $w'$ must lead to the situation that causes $M'$ to interrupt and reject. Consequently, $M''$ accepts the subset of $\lang(M')$ of valid splits. Thus, $\lang(M'') = \splitset{\varset}{L}$
\end{proof}

Let $L$ be a regular subword-marked language over $\Sigma$ and $\varset$, and let $\varsigma_{E}^{=}$ be a string-equality selection with $E = \{\mathcal{Z}_1, \mathcal{Z}_2, \ldots, \mathcal{Z}_k\} \subseteq \mathcal{P}(\varset)$. The next question is how we can represent the spanner $\varsigma_{E}^{=}(\llbracket L \rrbracket)$ by first applying the split version of $\llbracket L \rrbracket$, then applying some string-equality selections tailored to the split span relation, and then, finally, using the span-fusion operation to translate back the split span-tuples into the original span-tuples. More precisely, we want to represent $\varsigma_{E}^{=}(\llbracket L \rrbracket)$ in the form $(\bigspanfusion_{\Lambda}(\varsigma_{E'}^{=}(\llbracket \splitset{\varset}{L} \rrbracket)))$, for some suitable string-equality selection $\varsigma_{E'}^{=}$. \par
Intuitively speaking, $\varsigma_{E'}^{=}$ simulates $\varsigma^=_{E}$ in the sense that for any $\varsx, \varsy \in \mathcal{Z}_i$ it checks the string-equality of $t(\varsx)$ and $t(\varsy)$ by checking the string-equalities for all the variable pairs $t(\varsx_j)$ and $t(\varsy_j)$ for every $j \in [\varno]$. 
More formally, for every $i \in [k]$ and every $j \in [\varno]$, we define $\mathcal{Y}_i^{j} = \{\varsx^j \mid \varsx \in \mathcal{Z}_i\}$, 
and we set $\splitset{\varset}{\varsigma_{E}^{=}} = \varsigma_{\{\mathcal{Y}_i^{j} \mid j \in [\varno], i \in [k]\}}^{=}$. \par

The following lemma follows more or less directly from the
definitions. Recall that, for every $\varsx \in \varset$, we have defined $\lambda^{\varno}_{\varsx} = \{\varsx^1, \varsx^2, \ldots, \varsx^{\varno}\}$.

\begin{lem}\label{splitEasyDirectionLemma}
Let $L$ be a regular subword-marked language over $\Sigma$ and $\varset$, let $\varsigma_{E}^{=}$ be a string-equality selection with $E \subseteq \mathcal{P}(\varset)$, let $\Lambda = \{\lambda^{\varno}_\varsx \to \varsx \mid \varsx \in \varset\}$ and let $w \in \Sigma^*$. Then $(\bigspanfusion_{\Lambda}(\splitset{\varset}{\varsigma_{E}^{=}}(\llbracket \splitset{\varset}{L} \rrbracket)))(w) \subseteq (\varsigma_{E}^{=}(\llbracket L \rrbracket))(w)$.
\end{lem}

\begin{proof}
Let $t \in
(\bigspanfusion_{\Lambda}(\splitset{\varset}{\varsigma_{E}^{=}}(\llbracket
\splitset{\varset}{L} \rrbracket)))(w)$. This means that there is a
split $t'$ of $t$ with $t' \in (\splitset{\varset}{\varsigma_{E}^{=}}(\splitset{\varset}{\llbracket L \rrbracket}))(w)$ and $\bigspanfusion_{\Lambda}(t') = t$. Due to the string-equality selection $\splitset{\varset}{\varsigma_{E}^{=}}$,
this means that, for every $i \in [k]$ and all $\varsx, \varsy \in
\mathcal{Z}_i$, if $t'(\varsx^1) \neq \undefin$ and $t'(\varsy^1) \neq
\undefin$, then the spans $t'(\varsx^j)$ and $t'(\varsy^j)$ refer to equal factors of $w$ for every $j \in [\varno]$. This directly
implies that also $(\bigspanfusion_{\Lambda}(t'))(\varsx)$ and
$(\bigspanfusion_{\Lambda}(t'))(\varsy)$ refer to equal factors in $w$,
which means that $\bigspanfusion_{\Lambda}(t') \in
(\varsigma_{E}^{=}(\llbracket L \rrbracket))(w)$. Since
$\bigspanfusion_{\Lambda}(t') = t$, this means that $t \in
(\varsigma_{E}^{=}(\llbracket L \rrbracket))(w)$.  
\end{proof}

The converse of Lemma~\ref{splitEasyDirectionLemma} does not
necessarily hold. This is due to the fact that for a $t \in (\varsigma_{E}^{=}(\llbracket L \rrbracket))(w)$ and $\varsx, \varsy \in \mathcal{Z}_i$ such that $t(\varsx)$ and $t(\varsy)$ refer to equal factors, there might not be any split $t'$ of $t$ with the property that, for every $j \in [\varno]$, $t(\varsx^j)$ and $t(\varsy^j)$ refer to equal factors, which is necessary for $t$ being in $(\bigspanfusion_{\Lambda}(\splitset{\varset}{\varsigma_{E}^{=}}(\llbracket \splitset{\varset}{L} \rrbracket)))(w)$. Let us clarify this with an example.

\begin{exa}
Let $L = \lang(\open{\varsx} \ta \open{\varsy} \ta^* \close{\varsx} \ta \close{\varsy})$ be a regular subword-marked language over $\Sigma = \{\ta\}$ and $\varset = \{\varsx, \varsy\}$. Then $t = (\spann{1}{19}, \spann{2}{20})$ is in $(\varsigma_{\{\varsx, \varsy\}}^{=} \llbracket L \rrbracket)(\ta^{19})$. Now consider an arbitrary split $t'$ of $t$. By definition, $t'$ is a span-tuple over $\extVarset{\varset} = \{\varsx^1, \ldots, \varsx^{17}, \varsy^1, \ldots, \varsy^{17}\}$, such that 
\begin{itemize}
\item $t'(\varsx^1) = \spann{1}{\ell_1}, t'(\varsx^2) = \spann{\ell_1}{\ell_2}, t'(\varsx^3) = \spann{\ell_2}{\ell_3}, \ldots, t'(\varsx^{17}) = \spann{\ell_{16}}{19}$, 
\item $t'(\varsy^1) = \spann{2}{k_1}, t'(\varsy^2) = \spann{k_1}{k_2}, t'(\varsy^3) = \spann{k_2}{k_3}, \ldots, t'(\varsy^{17}) = \spann{k_{16}}{20}$, 
\item any two variables from $\extVarset{\varset}$ are non-overlapping. 
\end{itemize}
Now assume that $t'$ also has the property that, for every $i \in
[17]$, $t(\varsx^i)$ and $t(\varsy^i)$ refer to equal factors of $\ta^{19}$. For simplicity, let us assume that none of the spans are
empty (we will briefly discuss the general case below). Since $\varsx^1$ and $\varsy^1$ are non-overlapping, we must have $\ell_1 = 2$, which means that $t'(\varsx^1)$ refers to a factor $\ta$. Since, by assumption, $t'(\varsx^1)$ and $t'(\varsy^1)$ refer to equal factors, we therefore must also have $k_1 = 3$. Analogously, we can conclude that $\ell_2 = 3$ and $k_2 = 4$, and 
inductively repeating this argument implies that every span must refer
to a factor of size $1$, which is not possible since $t(\varsx)$ and
$t(\varsy)$ both refer to a factor of size $18$, and we only have $17$
variables $\varsx^1, \ldots, \varsx^{17}$ and $17$ variables
$\varsy^1, \ldots, \varsy^{17}$. Note that the same kind of
pigeon-hole argument also applies if some of the spans of the
variables from $\extVarset{\varset}$ are empty.
Consequently, no split $t'$ can have
the property that, for every $i \in [17]$, $t(\varsx^i)$ and
$t(\varsy^i)$ refer to equal factors of $\ta^{19}$, which means that
$t$ cannot be in
$(\bigspanfusion_{\Lambda}(\splitset{\varset}{\varsigma_{E}^{=}}(\llbracket
\splitset{\varset}{L} \rrbracket)))(\ta^{19})$.\par 
Note that this problem cannot be resolved by simply enlarging the
extended variable set $\extVarset{\varset}$ to $\{\varsx^1, \ldots,
\varsx^p, \varsy^1, \ldots, \varsy^p\}$ for some sufficiently large
$p$, since then the same contradiction can be obtained with the
document $\ta^{p+2}$. 
\end{exa}

This example points out that using a string-equality selection with respect to variables that can overlap is the reason that the converse of Lemma~\ref{splitEasyDirectionLemma} is not necessarily true. We will next restrict the string-equality selections accordingly.

Let $L$ be a regular subword-marked language over $\Sigma$ and $\varset$, and let $E =$\linebreak$\{\mathcal{Z}_1, \mathcal{Z}_2, \ldots, \mathcal{Z}_k\} \subseteq \mathcal{P}(\varset)$. We say that $E$ is \emph{non-overlapping} (\emph{with respect to $L$}) if all $\varsx, \varsy \in \bigcup_{i \in [k]} \mathcal{Z}_i$ are non-overlapping with respect to $L$.

\begin{lem}\label{splitHardDirectionLemma}
Let $L$ be a regular subword-marked language over $\Sigma$ and
$\varset$, let $\varsigma_{E}^{=}$ be a string-equality selection such
that $E \subseteq \mathcal{P}(\varset)$ is non-overlapping with
respect to $L$, let $\Lambda = \{\lambda^{\varno}_\varsx \to \varsx
\mid \varsx \in \varset\}$ and let $w \in \Sigma^*$. Then
$(\varsigma_{E}^{=}(\llbracket L \rrbracket))(w) \subseteq
(\bigspanfusion_{\Lambda}(\splitset{\varset}{\varsigma_{E}^{=}}(\llbracket
\splitset{\varset}{L} \rrbracket)))(w)$. 
\end{lem}

\begin{proof}
Let $E = \{\mathcal{Z}_1, \mathcal{Z}_2, \ldots, \mathcal{Z}_k\}$, and
let $t \in (\varsigma_{E}^{=}(\llbracket L \rrbracket))(w)$. We will
show that there is a split $t'$ of $t$, such that $t' \in
(\splitset{\varset}{\varsigma_{E}^{=}}(\llbracket
\splitset{\varset}{L} \rrbracket))(w)$. If such a split exists, then,
due to the fact that $\bigspanfusion_{\Lambda}(t') = t$, we have $t
\in
(\bigspanfusion_{\Lambda}(\splitset{\varset}{\varsigma_{E}^{=}}(\llbracket
\splitset{\varset}{L} \rrbracket)))(w)$. \par
Since $E$ is a non-overlapping string equality selection, we know that all the spans $t(\varsx)$ with $\varsx \in \bigcup^k_{i = 1} \mathcal{Z}_i$ are pairwise non-overlapping. However, it is not sufficient to simply create the span-tuple $t'$ where, for every $\varsx \in \bigcup^k_{i = 1} \mathcal{Z}_i$, $t'(\varsx^1) = t(\varsx) = \spann{\ell}{k}$ and $t'(\varsx^i) = \spann{k}{k}$ for every $i \in \{2, 3, \ldots, \varno\}$. In general, this definition would yield a span-tuple $t'$ that is not a valid split of $t$, since $t$ may have overlapping spans with respect to variables that are not in $\bigcup^k_{i = 1} \mathcal{Z}_i$, i.\,e., for a variable $\varsx \in \varset \setminus (\bigcup^k_{i = 1} \mathcal{Z}_i)$ and some other variable $\varsy \in \varset$, the spans $t(\varsx)$ and $t(\varsy)$ can be overlapping. 

For every $\varsx \in \varset$ with $t(\varsx) \neq \undefin$, let $t(\varsx) = \spann{\ell_{\varsx}}{k_{\varsx}}$. We observe that every $t(\varsx)$ can be interpreted as the interval $\{\ell_{\varsx}, \ell_{\varsx}+1, \ldots, k_{\varsx}\}$.
A \emph{split point of $t(\varsx)$} is any element from $\{\ell_{\varsx}, \ell_{\varsx}+1, \ldots, k_{\varsx}\}$. Now let us assume that, for every $\varsx \in \varset$, we have a set $\{p_{\varsx, 1}, p_{\varsx, 2}, \ldots, p_{\varsx, s_{\varsx}}\}$ of split points of $t(\varsx)$, which satisfy the following properties. 
\begin{enumerate}
\item For every $\varsx \in \varset$, $p_{\varsx, 1} = \ell_{\varsx}$, $p_{\varsx, s_{\varsx}} = k_{\varsx}$ and $s_{\varsx} \leq \varno + 1$.
\item For every $\varsx, \varsy \in \varset$ and every $i_{\varsx} \in [s_{\varsx} - 1]$ and $i_{\varsy} \in [s_{\varsy} - 1]$, the spans $\spann{p_{\varsx, i_{\varsx}}}{p_{\varsx, i_{\varsx} + 1}}$ and $\spann{p_{\varsy, i_{\varsy}}}{p_{\varsy, i_{\varsy} + 1}}$ are non-overlapping.
\item For every $i \in [k]$ and $\varsx, \varsy \in \mathcal{Z}_i$, $s_{\varsx} = s_{\varsy}$ and, for every $j \in [s_{\varsx}]$, $p_{\varsx, j} - \ell_{\varsx} = p_{\varsy, j} - \ell_{\varsy}$. 
\end{enumerate}

It can be easily seen that such split points induce a split $t'$ of $t$ such that $t' \in (\splitset{\varset}{\varsigma_{E}^{=}}(\llbracket \splitset{\varset}{L} \rrbracket))(w)$. Indeed, for every $\varsx \in \varset$, we define $t'(\varsx^1) = \spann{p_{\varsx, 1}}{p_{\varsx, 2}}$, $t'(\varsx^2) = \spann{p_{\varsx, 2}}{p_{\varsx, 3}}$, $\ldots$, $t'(\varsx^{s_{\varsx}-1}) = \spann{p_{\varsx, s_{\varsx} - 1}}{p_{\varsx, s_{\varsx}}}$, and we define $t'(\varsx^{q}) = \spann{p_{\varsx, s_{\varsx}}}{p_{\varsx, s_{\varsx}}}$ for every $q$ with $s_{\varsx} \leq q \leq \varno$. Since $s_{\varsx} \leq \varno + 1$ for every $\varsx \in \varset$, this is well-defined. Due to the first property of the split points, the span-tuple $t'$ satisfies the first condition of a split, and due to the second property of the split points, it also satisfies the second condition of a split. Thus, $t'$ is a split of $t$. Since $t$ satisfies the string-equality selection $\varsigma_{E}^{=}$, we must have $k_{\varsx} - \ell_{\varsx} = k_{\varsy} - \ell_{\varsy}$ for every $\varsx, \varsy \in \mathcal{Z}_i$ and $i \in [k]$. The fact that $t$ satisfies the string equality selection $\varsigma_{E}^{=}$ together with the third property of the split points directly implies that $t'$ satisfies the string-equality selection $\splitset{\varset}{\varsigma_{E}^{=}}$. This means that $t' \in (\splitset{\varset}{\varsigma_{E}^{=}}(\llbracket \splitset{\varset}{L} \rrbracket))(w)$.\par
In order to complete the proof, we have to show that split points with the above properties can be created. Let us first observe that if we declare some elements of $\{1, 2, \ldots, |w|+1\}$ to be \emph{general split points}, then this defines a set $\{p_{\varsx, 1}, p_{\varsx, 2}, \ldots, p_{\varsx, s_{\varsx}}\}$ of split points for every $\varsx \in \varset$. More precisely, for every $\varsx \in \varset$, $\{p_{\varsx, 1}, p_{\varsx, 2}, \ldots, p_{\varsx, s_{\varsx}}\}$ is exactly the set of all general split points $p$ that satisfy $p \in \{\ell_{\varsx}, \ell_{\varsx} + 1, \ldots, k_{\varsx}\}$. We declare the general split points in two steps.

\begin{itemize}
\item \emph{First step}: For every $\varsx \in \varset$ with $t(\varsx) \neq \undefin$, both $\ell_{\varsx}$ and $r_{\varsx}$ are declared general split points. \par
\item \emph{Second step}: Let us call the general split points added in the first step \emph{old} general split points, and the ones to be added in this step \emph{new} general split points. For every $i \in [k]$ we now proceed as follows.\par
Let $\mathcal{Z}_i = \{\varsz_1, \varsz_2, \ldots, \varsz_{k_i}\}$. For every $j \in [k_i]$, let $\{q_{j, 1}, q_{j, 2}, \ldots, q_{j, s_j}\}$ be the old general split points in the interval $\{\ell_{\varsz_j}, \ell_{\varsz_j} + 1, \ldots, k_{\varsz_j}\}$, i.\,e., the old general split points defined in the first step that will be split points of $t(\varsz_j)$, and let $A_j = \{q_{j, 1} - \ell_{\varsz_j}, q_{j, 2} - \ell_{\varsz_j}, \ldots, q_{j, s_j} - \ell_{\varsz_j}\}$ be the set of these split points shifted to the left by $\ell_{\varsz_j}$. Then, we join all these sets into $B_i = \bigcup^{k_i}_{j = 1} A_j$ and shift all these points back to their corresponding positions in  each of the intervals $\{\ell_{\varsz_j}, \ell_{\varsz_j} + 1, \ldots, k_{\varsz_j}\}$, i.\,e., we declare all elements of $\{p + \ell_{\varsz_j} \mid j \in [k_i], p \in B_i\}$ to be new split points.
\end{itemize}
This concludes the definition of general split points and therefore defines, for every $\varsx \in \varset$, a set $\{p_{\varsx, 1}, p_{\varsx, 2}, \ldots, p_{\varsx, s_{\varsx}}\}$ of split points of $t(\varsx)$. By definition, for every $\varsx \in \varset$, $p_{\varsx, 1} = \ell_{\varsx}$ and $p_{\varsx, s_{\varsx}} = k_{\varsx}$, since both $\ell_{\varsx}$ and $k_{\varsx}$ are declared general split points in the first step.\par
We now estimate the total number of general split points introduced by the two steps from above. We first note that in the first step we have introduced at most $2|\varset|$ general split points (the maximum $2|\varset|$ is reached if $t(\varsx) \neq \undefin$ for all variables $\varsx \in \varset$). In the second step, for every $i \in [k]$, we create at most $k_i|B_i|$ new general split points, where $B_i = \bigcup^{k_i}_{j = 1} A_j$. Since each $A_j$ only contains old general split points, we know that $|A_j| \leq 2|\varset|$ for every $j \in [k_i]$. Thus, $|B_i| \leq k_i 2|\varset|$. This means that in the second step, we create at most $\sum^{k}_{i = 1} k_i 2|\varset| = 2|\varset| \sum^{k}_{i = 1} k_i \leq 2|\varset|^2$. Finally, we conclude that the total number of general split points is at most $2|\varset| + 2|\varset|^2 \leq 4|\varset|^2$. In particular, this means that $s_{\varsx} \leq 4|\varset|^2 = \varno + 1$ for every $\varsx \in \varset$. We therefore conclude that the defined split points satisfy the first property mentioned above.\par
That the split points satisfy the second and third property mentioned above is obvious by construction.
\end{proof}

\subsubsection{Core-Spanners with Non-Overlapping String-Equality Selections}

We can now plug together the previously proven lemmas in the way sketched at the beginning of Section~\ref{sec:ExpressivePower} in order to obtain this section's main result. In particular, note that due to the normal form of Lemma~\ref{lemma:CoreSimplificationLemma}, every core spanner (with non-overlapping string equality selections) has a representation as in the statement of the following theorem.

\begin{thm}\label{coreSpannersToReflSpannersTheorem}
Let $S$ be a core spanner with $S =
\pi_{\mathcal{Y}}\varsigma_{E}^{=}(S')$, where $S'$ is a regular
spanner over $\Sigma$ and $\varset$, $\mathcal{Y} \subseteq \varset$
and $E \subseteq \mathcal{P}(\varset)$ such that $E$ is
non-overlapping with respect to $S'$. There
are a variable set $\varset'$,
a reference-bounded refl-spanner $S''$ over $\Sigma$ and $\varset'$, and a set $\Lambda \subseteq \mathcal{P}(\varset')$ such that $S = \pi_{\mathcal{Y}}\bigspanfusion_{\Lambda}(S'')$. Moreover, $|\varset'| = \bigO(|\varset|^3)$.
\end{thm}

\begin{proof}
Let $L$ be a regular subword-marked language with $S' = \llbracket L
\rrbracket$. Since $E$ is non-overlapping with respect to $L$,
Lemmas~\ref{splitEasyDirectionLemma}~and~\ref{splitHardDirectionLemma}
imply that 
\[
\varsigma_{E}^{=}(\llbracket L \rrbracket) \ \ = \ \ 
\bigspanfusion_{\Lambda}
\splitset{\varset}{\varsigma_{E}^{=}}(\llbracket \splitset{\varset}{L}
\rrbracket),
\]
where 
$\Lambda = \{\lambda^{\varno}_\varsx \to \varsx
\mid \varsx \in \varset\}$.
\par 
Moreover, $\splitset{\varset}{L}$ is a subword-marked language over
$\Sigma$ and $\extVarset{\varset}$ such that any two variables from $\extVarset{\varset}$ are
non-overlapping with respect to $\splitset{\varset}{L}$. This directly
implies that every variable from $\extVarset{\varset}$ is simple with
respect to $\splitset{\varset}{L}$. By applying
Theorem~\ref{stronglyNonOverlappingTheorem}, we can conclude that
there is a reference-bounded regular ref-language $L'$ over $\Sigma$ and
$\extVarset{\varset}$ with $\llbracket L' \rrbracket =
\splitset{\varset}{\varsigma_{E}^{=}}(\llbracket \splitset{\varset}{L}
\rrbracket)$. Hence, $S=\pi_{\mathcal{Y}}\varsigma_{E}^{=}(\llbracket L
\rrbracket) = \pi_{\mathcal{Y}}\bigspanfusion_{\Lambda}
\splitset{\varset}{\varsigma_{E}^{=}}(\llbracket \splitset{\varset}{L}
\rrbracket) = \pi_{\mathcal{Y}}\bigspanfusion_{\Lambda}(\llbracket L'
\rrbracket)$.\par 
Finally, by Observation~\ref{extVarsetSizeObs}, we have that $|\extVarset{\varset}| = \bigO(|\varset|^3)$.
\end{proof}

\section{Conclusion}\label{sec:conclusions}

In this work, we introduced reference-bounded refl-spanners, a new fragment of core spanners. In terms of expressive power, this fragment 
is slightly less powerful than 
the class of core spanners, but has lower evaluation complexity (see
Table~\ref{comparisonTable}, and note further that these upper bounds even hold for refl-spanners that are \emph{not} necessarily reference-bounded). If we add the \emph{span-fusion} -- a natural binary operation on spanners -- followed by a projection,
to reference-bounded
refl-spanners (see
Section~\ref{sec:ExpressivePower})
then they 
have the same expressive power as 
core spanners with \emph{non-overlapping} string-equality
selections. This demonstrates that our formalism covers all aspects of
core spanners except for the possibility of applying string-equality
selections on variables with overlapping spans. Moreover, since we
achieve better complexities for refl-spanners compared to core
spanners, this also shows that overlapping string-equality selections, followed by a projection, are a source of complexity for core spanners. \par 
From a conceptional point of view, our new angle was to treat the classical two-stage approach of core spanners, i.\,e., first producing the output table of a regular spanner and then filtering it by applying the string-equality selections, as a single $\NFA$. This is achieved by using ref-words in order to represent a document along with a span-tuple \emph{that satisfies the string-equality selections}, instead of just using subword-marked words to represent a document along with a span-tuple, which might not satisfy the string-equality selections and therefore will be filtered out later in the second evaluation stage.\par
A question that is left open for further research is about enumeration for some 
fragment of core spanners 
strictly more powerful than the regular spanners. In this regard, we note
that for efficient enumeration for (some fragments of) 
core spanners, 
we must overcome the general intractability of $\NonemptProb$ (which we have for core spanners as well as for (reference bounded) refl-spanners). 
We believe that refl-spanners are a promising candidate for further
restrictions that may lead to a fragment of 
core spanners 
that enables constant delay enumeration.

\section*{Acknowledgment}
\noindent
We thank the anonymous reviewers for their valuable comments.
In particular, the differences between
this paper and the preliminary conference version \cite{SchmidSchweikardt2021}, as
described in Section~\ref{subsection:DifferencesToPreviousVersion},
mainly go back to the reviewers' suggestions. 
They helped to considerably improve the paper both in terms of results and
in terms of the presentation.

%\bibliographystyle{alphaurl}
%\bibliography{bibfile}

\newcommand{\etalchar}[1]{$^{#1}$}

\end{document}